\documentclass[reprint,prd,onecolumn,longbibliography,eqsecnum,floats,aps,amsmath,%
amssymb,nofootinbib,showpacs,notitlepage,superscriptaddress]{revtex4-1}

\usepackage{amsmath,amssymb,amsfonts,amsthm}
\usepackage{slashed}
\usepackage{graphicx}
\usepackage{enumerate}
\usepackage{psfrag}   
\usepackage{calc}

\input cyracc.def
\newfam\cyrfam
\font\tencyr=wncyr10
\font\sevencyr=wncyr7
\font\fivecyr=wncyr5
\def\cyr{\fam\cyrfam\tencyr\cyracc}
\textfont\cyrfam=\tencyr \scriptfont\cyrfam=\sevencyr
\scriptscriptfont\cyrfam=\fivecyr
\newcommand{\ppi}{{\cyr p}}

\usepackage{IHsymbols}

\begin{document}
\allowdisplaybreaks

\title{Neighbourhoods of Isolated Horizons and their stationarity}

\author{Jerzy \surname{Lewandowski}}
	\email{Jerzy.Lewandowski@fuw.edu.pl}
	\affiliation{Instytut Fizyki Teoretycznej, Wydzia{\l} Fizyki Uniwersytetu 
		Warszawskiego, Ho\.za 69, 00-681 Warszawa, Poland.}
\author{Tomasz \surname{Paw{\l}owski}}
	\email{tomasz.pawlowski@unab.cl}
	\affiliation{Departamento de Ciencias F\'isicas, Facultad de Ciencias Exactas,
    Universidad Andres Bello, Av.~Rep\'ublica 220,  Santiago 8370134, Chile.}
	\affiliation{Instytut Fizyki Teoretycznej, Wydzia{\l} Fizyki Uniwersytetu 
		Warszawskiego, Ho\.za 69, 00-681 Warszawa, Poland.}

\begin{abstract}
	A distinguished (invariant) Bondi-like coordinate system is defined in the 
	spacetime neighbourhood of a non-expanding horizon of arbitrary dimension via 
	geometry invariants of the horizon. With its use, the radial expansion of a 
	spacetime metric about the horizon is provided and the free data needed to 
	specify it up to given order are determined in spacetime dimension $4$. 
	For the case of an electro-vacuum 
	horizon in $4$-dimensional spacetime the necessary and sufficient conditions for the 
	existence of a Killing field at its neighbourhood are identified as differential 
	conditions on the horizon data and data on null surface transversal to the horizon.
\end{abstract}


\maketitle

\section{Introduction}
\label{sec:intro}

Systematic studies of black holes in various approaches to quantum gravity as well 
as accurate description of the dynamical evolution of these exotic objects require 
a quasi-local description formalism -- where a black hole can be treated as 
``object in the lab'' and the global spacetime structure of the universe far away 
from it can be neglected. Among several approaches to construct such formalism 
\cite{pn-bh,h-1993} one of considerable success is the theory of Isolated Horizons 
\cite{abdfklw,ak-idh-rev,gj-4d}. This approach was originally inspired by ideas 
of Pejerski and Newman \cite{pn-bh}, next shaped into a solid formalism by Ashtekar 
\cite{ack-g} and subsequently developed by many researchers. Its main feature is 
the representation of a black hole in equilibrium through its surface -- the 
non-expanding (or isolated) horizon -- a null cylinder of codimension $1$ and of 
compact spatial slices embedded in a Lorentzian spacetime. The black hole is 
characterized by the geometry (and possibly matter fields) data on this surface only. 
Both geometry aspects \cite{ack-g,abl-g} and the mechanics \cite{abf-m,abl-m} have 
been systematically studied in spacetime dimension $4$ and then extended to general 
dimension \cite{aw-3d,lp-g,klp-m}, the latter including in particular asymptotically 
Anti-deSitter spacetimes \cite{apv-ads}. Also various matter content at the horizon 
has been considered \cite{ac-dil,*acs-hair,*acs-scalar} as well as the properties 
of symmetric horizons \cite{lp-symm} and their relation with standard black hole 
solutions \cite{lp-kerr,lp-extr}. The formalism has been further extended to the 
non-equilibrium situations through the \emph{Dynamical Horizons} 
\cite{ak-dh-prl,*ak-dh-det} (see also \cite{ak-idh-rev}). It is vastly applied in 
numerical relativity (see for example \cite{s-dh-intro,*bshh-bbh,*jvg-inter,*ch-bump}) 
as well as in black hole description in loop quantum gravity \cite{al-status} 
-- especially as the basis for entropy calculations (see for example 
\cite{aev-entropy,*cdf-entropy,*abfdv-entropy}). The extension to this formalism has 
found applications also in supergravity \cite{lb-super,*bl-super} and string theory 
inspired gravity \cite{l-string,*l-sext1,*lb-sext,*l-sext2}.

The quasi-locality of the theory is a great advantage, as only the geometry objects 
at the horizon are relevant in the description, however for this very reason one 
misses the information about the black hole neighbourhood. However, the success of 
the formalism of \emph{near horizon geometries} \cite{kl-near1,*kl-near2} shows clearly, 
that there is a strong demand for any black hole description 
method to be able to also ``handle'' its neighbourhood, a feature particularly 
relevant for the studies of black hole spacetimes in context of AdS/CFT correspondence 
\cite{p-ads-cft} and in numerical probing of the late stages of black hole mergers. 

This article is dedicated to supplying such extension within the isolated horizon 
formalism. Its main ideas have been originally published in \cite{p-phd}. 
The principal part of presented work is providing the convenient way of 
describing the spacetime geometry in the neighbourhood of the horizon through the 
\emph{Bondi-like} coordinate system originally introduced in \cite{l-sem}. Here, 
this construction is extended to arbitrary spacetime dimension and horizon spatial 
slice topology and its properties are studied. It is shown to be relatively 
convenient in use, providing for example a well defined invariant radial spacetime 
metric expansion about the horizon. It provides a solid frame for addressing the 
questions like the conditions for the existence of a Killing vector fields in the 
horizon neighbourhood, which problem is studied here in detail in context of the 
electro-vacuum black hole in $4$-dimensional spacetime.
\\

The paper is organized as follows: We start in section~\ref{sec:intro-geom} with 
short introduction of the geometric structure of general non-expanding horizons
in arbitrary dimension and discuss the definition and those properties of the horizon
symmetries, which will be needed in the further studies of Killing horizons. Next, 
in section~\ref{sec:Bondi}, still keeping the general dimension, we construct the 
Bondi-like coordinate system which later will provide the basis to describe the
structures at the neighbourhood. In the same section we study the properties and form
of Killing fields possibly present there. The material of these two sections is then used
in section~\ref{sec:4d-neigh}, where we restrict the studies to the horizons in 
$4$-dimensional electrovac spacetime, introducing in particular convenient Newman-Penrose
null frame (consistent with Bondi-like coordinate system) for any matter type and studying
the initial value problem for the ``radial'' metric expansion. The interest is then 
restricted just to horizons embeddable in the evectrovac spacetime. There the Maxwell
evolution equations change the mentioned initial value problem, making it stronger 
(less initial data needed). This structure is then applied in section~\ref{sec:evac-kvf},
where we formulate the conditions necessary and sufficient for the non-expanding horizon
to be Killing horizon. Analogous conditions for several classes of Killing fields are
derived in section~\ref{sec:evac-kvf-axihel}. We summarize the results in the concluding
section~\ref{sec:concl}. In order to not disrupt the reasoning flow some more complicated 
proofs as well as the definition of Newman-Penrose frame have been moved to 
appendices~\ref{app:NP} through~\ref{app:4d-kill}.

\section{Geometry of a non-expanding horizon}
  \label{sec:intro-geom}

\subsection{Non-expanding horizons}
  \label{sec:neh}

In this section we introduce the notation, recall the definition and
outline the properties of non-expanding horizons in $n$ dimensions 
\cite{abdfklw,abl-g,lp-g,lp-symm}. 
We often use the abstract index notation
and the following convention: the $n$-spacetime indexes are denoted
by $\alpha,...,\nu$,  every $n-1$ dimensional vector spaces indexes
are denoted by $a,...,d$ and each $n-2$ vector space indexes are
denoted by $A,...,D$.

\subsubsection{Definition, the induced degenerate metric, 2-volume and Hodge $*$}
  \label{ssec:neh-def}

We start our consideration with an $(n-1)$-dimensional null surface $\ih$
embedded in an $n$-dimensional, time oriented spacetime $\M$. In Section
\ref{sec:4d-neigh} and further
$n=4$. The spacetime metric tensor $g_{\mu\nu}$ of the signature
$(-,+,\cdots,+)$ is assumed to satisfy the Einstein field equations
(possibly with matter and cosmological constant). 
Throughout the paper, we assume that all the matter fields possibly
present at the surface $\ih$  satisfy the following:
\begin{cond}\label{c:energy}(Stronger Energy Condition)
  At every point of $\ih$, for every future oriented null vector
  $\ell$ tangent to $\ih$, the vector $-T^{\mu}{}_{\nu}\ell^{\nu}$
  is causal and future oriented, where $T_{\mu\nu}$ is the energy-momentum 
  tensor.
\end{cond}
We  denote the degenerate metric tensor induced at $\ih$ by
$q_{ab}$. The sub-bundle of the tangent bundle $T(\ih)$ defined by
the null vectors is denoted by $L$ and referred to as the null
direction bundle. Given a vector bundle $P$, the set of sections
will be denoted by $\Gamma(P)$.

To recall the definition of a non-expanding horizon consider the
metric tensor $q_{AB}$ induced by the tensor $q_{ab}$ in the fibers
of the quotient bundle $T(\ih)\, / \,L $. Denote the inverse metric
tensor defined in the fibers of the dual bundle $(T(\ih) / L)^*$ by
$q^{AB}(x)$ (that is $q_{AB}q^{BC} =\delta_A^C$).

We will be assuming that $\ih$ is a non-expanding horizon, in the
following sense:
\begin{defn}\label{def:neh}
  A null submanifold $\ih$ of codimension 1 embedded in spacetime
  satisfying the
  Einstein field equations is called a {\it non-expanding horizon}
  (NEH) if:
  \begin{enumerate}
    \item for every point $x\in\ih$ for every
      null vector $\ell^a$ tangent to $\ih$ at $x$,
      \begin{equation}
        q^{AB}\lie_\ell q_{AB}\ =\ 0 \ ,
      \end{equation} 
      and
    \item $\ih$ has the product structure $\bas\times\I$,
     that is there
      is an embedding
      \begin{equation}
        \bas\times\I\ \rightarrow\ \M
      \end{equation}
      such that:
      \begin{enumerate}[ (i)]
        \litem $\ih$ is the image,
        \litem $\bas$ is an $n-2$ dimensional compact and
          connected \footnote{In the case $\bas$ is not connected all
            the otherwise global constants (like surface gravity) remain
            constant only at maximal connected components of the
            horizon.} manifold (referred to as the horizon base space),
        \litem $\I$ is an open interval,
        \litem for every maximal null curve in $\ih$ there is
          $\hat{x}\in\bas$ such that the curve is the image of
          $\{\hat{x}\}\times \I$.
      \end{enumerate}
  \end{enumerate}
\end{defn}
Each (non-vanishing  in a generic set) null vector field $\ell$
defined in $\ih$ (a section of $L$) determines a function
$\sgr{\ell}$ referred to as surface gravity\footnote{We use
dimensionless coordinates in spacetime, therefore our surface
gravity is also dimensionless.} of $\ell$, such that
\begin{equation}\label{eq:geod}
  \ell^{\mu}\nabla_{\mu}\ell^{\nu} = \sgr{\ell}\ell^{\nu} \ .
\end{equation}
In particular, given a NEH there always exists a nowhere vanishing null
vector field $\ell_o^a$ of the identically vanishing surface
gravity, $\sgr{\ell_o}=0$. One can also choose a null vector field
$\ell^a$ of $\sgr{\ell}$ being an arbitrary constant,\footnote{The
	first one, $\ell_o$  can be defined by fixing appropriately affine
	parameter $v$ at each null curve in
  $\ih$. Then, the second vector field is just $\ell=\sgr{\ell}v\ell_o$.
}
\begin{equation}\label{eq:kappaconst}
  \sgr{\ell}\ =\ \const \ .
\end{equation}
That vector field $\ell^a$ can vanish in a harmless (for our
purposes) way on an $(n-2)$-dimensional section of $\ih$ only.

From Stronger Energy Condition \ref{c:energy} it follows in particular that
$T_{\ell\ell}\geq 0$. 
This in turn implies via the generalized
Raychaudhuri equation that the flow $[\ell]$ preserves the degenerate
metric $q$
\begin{equation}\label{eq:lie_q}
  \lie_{\ell}q_{ab}\ =\ 0 \ ,
\end{equation}
and using again the Stronger Energy Condition it can be shown that at every point  
$p$ of $\ih$ and for every $X\in T_p\ih$
\begin{equation}
	T_{a\beta}X^a\ell^\beta\ =\ \Ricn{n}_{a\beta}X^a\ell^\beta =\ 0 \ .\label{RicclX}
\end{equation}

The property \eqref{eq:lie_q} above combined with  $\ell^aq_{ab}=0$
means that $q_{ab}$ is the pullback of a certain metric tensor field
$\hq_{AB}$ defined on $\bas$. The horizon base space $\bas$ can be identified
with the space of null curves tangent to $\ih$.  Its manifold structure is unique. 
The pullback map is defined by the natural (and also unique) projection,
\begin{subequations}\label{eq:bas}\begin{align}
  \ppi:\ih\ &\rightarrow\ \bas \ , &
  q_{ab}\ &=\ (\ppi^*\hq)_{ab} \ . \tag{\ref{eq:bas}}
\end{align}\end{subequations}

The pullback of the base space $\bas$ $2$-volume form $\hat{\eta}_{AB}$,
\begin{equation} 
\ppi^*\hat{\eta}_{ab}\ =: \eta_{ab} 
\end{equation}
defines the canonical area 2-form on $\ih$ and its restriction $\eta_{AB}$
to a two form in $T(\ih)/L$. $\eta_{AB}$ is used to define
the horizon Hodge dualization $\star_\ih: (T(\ih)/L)^*\rightarrow (T(\ih)/L)^*$, 
\begin{equation}
  \star_{\ih}v_A\ :=\ \eta_{AB}q^{BC}v_C \ .  
\end{equation}

\subsubsection{The induced covariant derivative and the rotation $1$-form}
  \label{ssec:rot}

It can be shown by using \eqref{eq:lie_q}, that the  space time covariant
derivative $\nabla_\alpha$ determined by the metric tensor $g_{\alpha\beta}$  
preserves the tangent bundle $T(\ih)$. Indeed, for every pair of vector fields 
$X,Y\in \Gamma(T(\ih))$, 
\begin{equation} 
	\nabla_XY\in \Gamma(T(\ih)) \ .	
\end{equation}
Therefore, there exists in $T(\ih)$ an induced covariant derivative $D_a$, such 
that for every pair of vector fields $X,Y\in \Gamma(T(\ih))$ the following holds
\begin{equation}
  D_XY^a\ :=\ \nabla_XY^a \ .
\end{equation}
Its action on covectors, sections of the dual bundle $T^*\ih$ is
determined by the Leibniz rule. Together with the induced metric
the covariant derivative  constitutes the {\it geometry of a NEH}
$(q_{ab},D_a)$.

The connection $D_a$ preserves in particular the null direction
bundle $L$, thus for every $\ell\in\Gamma(L)$ the derivative $D_a\ell^b$ is 
proportional to $\ell^b$ itself,
\begin{equation}\label{eq:omega_def}
  D_a\ell^b\ =\  \w{\ell}{}_a\ell^b \ ,
\end{equation}
where  $\w{\ell}{}_a$ is a $1$-form defined uniquely on this subset
of $\ih$ on which $\ell\neq0$ is defined. We call $\w{\ell}{}_a$
the {\it rotation 1-form potential} (see \cite{abl-g,lp-g}).

The evolution of $\w{\ell}{}_a$ along the null flow on $\ih$ is responsible 
for the $0$th Law of the non-expanding horizon thermodynamics:
the rotation 1-form potential $\w{\ell}{}$ and surface gravity of
$\ell$, related to $\w{\ell}_a$ via
\begin{equation}
  \sgr{\ell}\ =\ \ell^a\w{\ell}_a
\end{equation}
satisfy the following constraint:
\begin{equation}\label{eq:0th}
  \lie_{\ell}\w{\ell}{}_a\ =\ D_a\sgr{\ell} \
\end{equation}
implied by (\ref{RicclX}).
This tells us in particular, that there is always a choice of the
section $\ell$ of the null direction bundle $L$ such that
$\w{\ell}{}$ is Lie dragged by $\ell$. Indeed, we can always find a
non-trivial section $\ell$ of $L$ such that $\sgr{\ell}$ is
constant ($0$ for $\ell$ defined by an affine parametrization of the null 
geodesics tangent to $\ih$). Throughout  the remaining part of the article 
we will restrict our consideration to fields $\ell\in \Gamma(L)$ of such 
class, or equivalently satisfying
\begin{equation}\label{eq:lieellomega}
  \lie_\ell \w{\ell}\ =\ 0 \ .
\end{equation}

Upon rescalings $\ell\mapsto\ell'=f\ell$ (where $f$ is a real
function defined at $\ih$) of a section $\ell^a$ of $L$ the rotation
$1$-form changes as follows
\begin{equation}\label{eq:omega'}
  \w{\ell'}{}_a\ =\ \w{\ell}{}_a + D_a \ln f \ .
\end{equation}
The requirement that both $\sgr{\ell}$ and $\sgr{\ell'}$ are constants
restricts the form of $f$ to the following one
\begin{equation}\label{eq:f}
  f\ =\ \begin{cases}
          B e^{-\sgr{\bsl}u} + \fracs{\sgr{\bsl'}}{\sgr{\bsl}}
            & \sgr{\bsl}\neq 0 \\
          \sgr{\bsl'} u - B  &  \sgr{\bsl}= 0 \\
        \end{cases} \ ,
\end{equation}
where $u$ is any function defined on $\ih$ such that
\begin{equation}\label{lu}
  \ell^aD_au\ =\ 1 \ ,
\end{equation}
and $B$ is an arbitrary function constant along null geodesics of $\ih$.

\subsubsection{The constraints}
  \label{ssec:degrees}
  
The non-expanding horizon geometry $(q_{ab},D_a)$  is constrained by
the Einstein equations. We have already used some of them above.
A \emph{complete} set of the constraints on horizon geometry
$(q_{ab},\,D_a)$ is encoded in the following identity
\begin{equation}\label{eq:[l,D]1}
  [\lie_\ell,D_a]X^b\ =\ \ell^b X^c
  \left[ D_{(a}\w{\ell}{}_{c)}\ +\ \w{\ell}{}_a\w{\ell}{}_c
    + \frac{1}{2}\left(\Ricn{n}_{ac}-\ppi^*\Ricn{n-2}_{ac}\right)
  \right] \ =:\ \ell^bN^{(\ell)}_{ac}X^c  \ ,
\end{equation}
which holds for every $\ell\in\Gamma(L)$ and $X\in \Gamma(T(\ih))$,
where $\Ricn{n-2}{}_{AB}$  is the Ricci tensor of the metric tensor
$\hat{q}_{AB}$ induced in base space $\bas$ whereas $\Ricn{n}_{ac}$ is the 
pullback to $\ih$ (by the embedding map $\ih\rightarrow {\cal M}$) of the spacetime 
Ricci tensor. The constraints are given by replacing $\Ricn{n}_{ac}$ with 
$8\pi G (T_{ac} - \frac{1}{2}Tq_{ac})$, where the pullback onto $\ih$ of the 
stress energy tensor $T_{ac}$ satisfies on $\ih$ the condition 
\begin{equation}\label{eq:lie_T-1}
 \ell^aT_{ab}\ = 0
\end{equation}
(see \ref{RicclX}). 
In particular, the $0$th law \eqref{eq:0th} is given by contracting
(\ref{eq:[l,D]1}) with a null vector $\ell$. The remaining constraints determine
the evolution of some other components of $D_a$ along the null geodesics tangent 
to $\ell$ provided the energy momentum tensor is given. In the Einstein vacuum or 
Einstein Maxwell vacuum cases the constraints are solved explicitly \cite{lp-g}. 

Throughout this paper we are making a stronger assumption, namely that on $\ih$ 
in addition to \eqref{eq:lie_T-1} the pull-back onto $\ih$ (by the embedding 
map $\ih\rightarrow {\cal M}$) of the energy-momentum tensor is Lie dragged by 
(any)\footnote{In fact, if this condition is satisfied by any given non-vanishing 
	$\ell^a$ then it is satisfied by every $\ell^a$ null and tangent to $\ih$.
} null vector field $\ell$ tangent to $\ih$
\begin{equation}\label{eq:lie_T-2}
	\lie_\ell T_{ab}\ =\ 0.
\end{equation} 
For the electromagnetic field considered later in this paper this condition is 
a consequence of the Einstein-Maxwell equations, therefore it is satisfied automatically. 
For the time being, however, we just assume it is true.

\subsubsection{The invariants} 
  \label{ssec:invariants}

For every non-expanding horizon $\ih$ which satisfies the assumptions made in
the previous subsection, the geometry $(q,D)$ is analytic along the null
geodesics in any affine coordinate. A given non-expanding horizon $\ih$  can be
incomplete in that coordinate, however one can consider  its (non-embedded)
maximal analytic extension $\bar{\ih}$ endowed with the analytic extension
$(\bar{q},\bar{D})$ of the geometry.  This is what we do in this section. We
will be using the same notation as above however we will mark all the symbols
referring to $\bar{\ih}$ by the extra $\bar{\cdot}$. Finally,  all the invariant
structures introduced on $\bar{\ih}$ determine  unique restriction to an
original, given unextended $\ih$.

Thus far we reduced the freedom in choice of a null vector field $\bar{\ell}$ 
tangent to $\bar{\ih}$ to vector fields which satisfy \eqref{eq:kappaconst} defined
up to the transformations given by \eqref{eq:f}. The freedom can be
further reduced by imposing some condition on the tensor
${N}^{(\bar{\ell})}_{ab}$ (\ref{eq:[l,D]1}). We also notice, that from (\ref{eq:0th},
\ref{eq:lieellomega}) it follows that
\begin{equation}
  \bar{\ell}^a{N}^{(\bar{\ell})}_{ab}\ =\ 0 \ ,
\end{equation}
hence this tensor defines a unique tensor ${N}^{(\bar{\ell})}_{AB}$ in the fibers of
$T(\bar{\ih})\,/\,\bar{L}$.

Let us now introduce a specific class of $\bar{\ell}$, namely:
\begin{defn}\label{def:natural}
  {\it A natural  vector field} $\bar{\ell}$ in $\bar{\ih}$, is a
  tangent null vector field (that is a section of $\bar{L}$) non-vanishing
  on  a generic (open and dense) subset of $\bar{\ih}$ which satisfies the following 
  conditions
  \begin{subequations}\label{invvect}\begin{align}
      \sgr{\bar{\ell}}\ &=\ 1 \ ,  &  \bar{q}^{AB}{N}^{(\bar{\ell})}_{AB}\ &=\ 0 \ .
  \tag{\ref{invvect}}\end{align}\end{subequations} 
\end{defn}
The generic existence and uniqueness of the natural vector field in
a given $\bar{\ih}$, was shown in \cite{lp-g}  (see Eq.~(6.22) 
therein and the following paragraph). We review the argumentation proving these 
properties further below. 
In the case when the natural vector field is unique we will call it 
{\it the invariant vector field of the $\bar{\ih}$ geometry}.

Given a non-zero null vector field $\bar{\ell}$ in $\bar{\ih}$ such that
$\sgr{\bar{\ell}}=\const\not=0$, a unique foliation by spacelike, $n-2$
dimensional sections of $\ih$ can be fixed, by using the rotation
1-form potential $\w{\bar{\ell}}$. In particular, one can choose a section
$\bar{\sigma}:\hat{\ih}\rightarrow {\bar{\ih}}$ of
$\bar{\pi}:{\bar{\ih}}\rightarrow\hat{\ih}$ such that the pullback
$\bar{\sigma}^*\w{\bar{\ell}}$ is divergence free,
\begin{equation}
  d\star\bar{\sigma}^*\w{{\bar{\ell}}}\ 
  =\ \bar{q}^{AB}\bar{D}_A\bar{\sigma}^*\w{{\bar{\ell}}}_B \ 
  =\ 0 \ .
\end{equation}
Those sections of $\bar{\ih}$ are called {\it good cuts} \cite{abl-g} and are defined
uniquely modulo the action of the flow of the vector field ${\bar{\ell}}$.
They set a foliation of ${\bar{\ih}}$. If $\bar{\ell}$ is the invariant vector
field of the $\bar{\ih}$ geometry, then the corresponding good cut
foliation will be called the {\it invariant foliation of the $\bar{\ih}$
geometry}. One has to be aware though, that whereas the slices of the invariant 
foliation of $\bar{\ih}$ are diffeomorphic to the base space $\bas$, the restriction 
of a slice of $\bar{\ih}$ to $\ih$ may be a proper subset of that slice. In other 
words, the slices of $\ih$ may be non-global sections of $\ppi:\ih\rightarrow \hat{\ih}$. 
Given a null vector field $\bar{\ell}$ and a foliation there is a unique
differential 1-form $\bar{n}_a\in \Gamma(T^*(\bar{\ih}))$ orthogonal to the
leaves of the foliation and normalized by
\begin{equation}
  \bar{n}_a\bar{\ell}^a\ =\ -1 \ .
\end{equation}
We will call $\bar{n}_a$  the {\it invariant co-vector of the $\bar{\ih}$
geometry} if $\bar{\ell}$ and the foliation are, respectively, the invariant
vector field of the $\bar{\ih}$ geometry and the invariant foliation of
$\bar{\ih}$.  Finally, the invariant vector field, foliation and covector field
of $(\bar{\ih},\bar{q}_{ab},\bar{D}_a)$ are restricted to a given NEH
$(\ih,q_{ab},D_a)$ the geometry of which was the starting point of an
extension $(\bar{\ih},\bar{q}_{ab},\bar{D}_a)$. The uniqueness of this
extension guaranties the uniqueness of the resulting invariant structures
defined on $(\ih,q_{ab},D_a)$.   

The natural vector field defined above exists and is unique for 
a certain class of NEH geometries, denoted as \emph{generic} and defined as follows:
Let $(\ih,q_{ab},D_a)$ be a NEH. Choose any null vector field $\ell'$ such that
$\sgr{\ell'}=const$ and any global section $\bar{\sigma}:\bas\rightarrow\bar{\ih}$. 
The existence of a natural vector field depends on the invertability of a certain 
operator introduced in \cite{lp-g}. It involves the following  ingredients defined 
on $\bas$: the induced metric tensor $\hq_{AB}$ \eqref{eq:bas}, the corresponding
covariant derivative $\hat{D}_A$ and the Ricci scalar $\hRicn{n-2}$, the
pullback $\hat{\omega}_A$ of the rotation 1-form potential $\w{\ell'}{}_a$ by
the section $\bar{\sigma}$, the trace $\hat{T}$ of the pullback $\hat{T}_{AB}$
of the energy-momentum tensor $T_{\alpha\beta}$, and $T:=T^{\alpha}{}_\alpha$. The
operator is
\begin{equation} \label{eq:Mdef}
  \hat{M}\ =\ \hat{D}^A\hat{D}_A + 2\hat{\omega}^A\hat{D}_A + \hat{\omega}_A\hat{\omega}^A 
  + \hat{D}_A\hat{\omega}^A +4\pi G\left(\hat{T} - T-\hRicn{n-2}\right) \ .   
\end{equation}
If the kernel of this operator is trivial, then there exists exactly one natural
vector field $\ell$. Suppose, for given data the operator has a nontrivial
kernel. Then, the operator corresponding to a new, gently perturbed data, say
$T'_{\alpha\beta}\ =\ T_{\alpha\beta} + \delta \Lambda g_{\alpha\beta}$, has 
trivial kernel for a non-zero perturbation $|\delta\Lambda|\in (0,\epsilon)$.
That shows the genericity of the existence and uniqueness of the natural vector.
It is important to point out, that the operator $\hat{M}$ defined above was 
constructed with use of more data, than $(q_{ab},D_a)$, since the objects
depending on the choice of a null vector field $\ell'$ and section of $\ih$ are 
present on the right hand side of \eqref{eq:Mdef}. However, the dimension of the 
kernel of this operator is independent of those choices. 

If the kernel is nontrivial, on the other hand, then $\ih$ either admits no natural 
vector field or admits more than one. The latter happens for example if $(\ih,q_{ab},D_a)$ 
has $2$-dimensional group of null symmetries \cite{lp-symm}.   
\medskip

In conclusion:
\begin{defn}\label{invgen}
  A NEH  $({\ih},q_{ab}, D_a)$ is invariant-generic if it defines the invariant 
  vector field.
\end{defn}
 
In an invariant-generic case, we have defined  the following unique structures of $({\ih},q_{ab}, D_a)$: 
\begin{itemize}
  \item an invariant tangent null vector field $\ell$, 
  \item an invariant foliation by good cuts
    --sections of $\pi:\ih\rightarrow\hat{\ih}$-- preserved (locally) by the flow 
    of $\ell$, 
  \item a function $v:\ih\rightarrow \mathbb{R}$, constant on the leaves of the 
		foliation and such that
    \begin{equation} 
      \ell^aD_av\ =\ 1 \ ,
    \end{equation}
    invariant up to $v\mapsto v+v_o$, $v_o\in\mathbb{R}$,
  \item an invariant covector field 
    \begin{equation} 
      n\ =\ -\rd v
    \end{equation}
    orthogonal to the leaves of the invariant foliation.
\end{itemize}

The reason for the name `invariant' is that given two NEHs $(\ih,q_{ab},D_a)$ and 
$(\ih',q'_{ab},D'_a)$ related by an isomorphism $\phi:\ih\rightarrow\ih'$, the 
corresponding invariants are mapped to each other by $\phi_*$, $\phi$, $\phi^*$ 
and $\phi^*$ respectively.  
One has to remember though, that the slices of $\ih$ may be non-global sections 
of $\ppi:\ih\rightarrow \hat{\ih}$.

\subsection{Symmetric NEH}
  \label{sec:sym-gen}

\subsubsection{Definitions, known results}\label{sec:sym-known}

Given a non-expanding horizon $\ih$, an {\it infinitesimal symmetry} of it is 
a non-trivial vector field, $X\in\Gamma(T(\ih))$ such that
\begin{equation}\label{eq:def-symm}
  \lie_X q_{ab}\ =\ 0, \quad {\rm and}\quad [\lie_X,\,D_a]\ =\ 0 \ .
\end{equation}
Each Killing field defined in a spacetime neighbourhood of a NEH $\ih$
and tangent to $\ih$ induces an infinitesimal symmetry of $\ih$. Therefore,
recalling the properties of symmetric NEHs (studied in \cite{lp-symm}) is a 
natural starting point for the current paper. Below we will briefly list 
those of their properties which are relevant for our studies. 

Every infinitesimal symmetry $X$ preserves the null direction, that is
for every null vector field $\ell\in\Gamma(T(\ih))$, 
\begin{equation} 
  [X,\ell]\ =\ f\ell, \qquad  f:\ih \rightarrow \mathbb{R}.
\end{equation} 
Due to this property the projection $\pi:\ih\rightarrow \bas$ pushes $X$
forward to a uniquely defined vector field on $\bas$, that is there is a vector 
field $\hat{X}\in\Gamma(T(\bas))$ such that  
\begin{equation}
  \ppi_*X = \hat{X}. 
\end{equation}
The vector field $\hat{X}$ is a Killing vector of the geometry $(\bas,\hat{q}_{ab})$.

Every infinitesimal symmetry $X$ defines a unique analytic extension
$\bar{X}$ to the maximal analytic extension $\bar{\ih}$ of $\ih$. The vector
field $\bar{X}$ defines  a global flow on $\bar{\ih}$, and the flow
preserves the geometry $(\bar{q},\bar{D})$ \cite{lp-symm}.
Using this property, in this subsection we will consider 
symmetric  maximal analytic extensions of the NEHs.  
  
We distinguish several classes of infinitesimal symmetries. 
One of them is  \emph{null infinitesimal symmetry}, corresponding to $X^a$ being 
a null vector field. 
\begin{cor}\label{nullsym} 
  Let $X$ be a null infinitesimal symmetry of $\ih$. Then
  \begin{itemize}
    \item $\rd\sgr{X}\ =\ 0$
    \item if $\sgr{X}\not= 0$, then there is a function $v:{\ih}\rightarrow \mathbb{R}$ 
      such that $dv\not=0$ at every point of $\ih$, and 
      \begin{equation}
				X^aD_av \ =\ \sgr{X}v
      \end{equation}
    \item if $\sgr{X}= 0$, then there is a function $v:{\ih}\rightarrow \mathbb{R}$
      such that
      \begin{equation}
				X^bD_b(X^aD_av)\ =\ 0
      \end{equation}
      on $\ih$, and 
      \begin{equation}
				\rd(X^aD_av)\ \neq\ 0
      \end{equation}
      at each point such that $(X^aD_av)=0$.
  \end{itemize}
\end{cor} 
That means that an infinitesimal null symmetry of non-zero surface gravity
can vanish only on a 2-dimensional slice of $\ih$, whereas in the case of 
the zero surface gravity it may vanish along a finite set of the geodesics. 
The last item is proved in \cite{lp-symm}. 

Furthermore we distinguish the \emph{cyclic} and \emph{helical} infinitesimal 
symmetries, whose definitions are longer, therefore we spell them out more carefully.   

\begin{defn}\label{def:cyc-def}
  Given a NEH $(\ih, q_{ab}, D_a)$, a vector field $\Phi^a\in \Gamma(T(\ih))$ is cyclic
  infinitesimal symmetry whenever the following holds:
  \begin{itemize}
    \litem $\Phi^a$ is an infinitesimal symmetry of $\ih$ (satisfies
      the equations~\eqref{eq:def-symm}),
    \litem the symmetry group of the maximal analytic extension $\bar{\ih}$ it 
			generates  is diffeomorphic to $SO(2)$,
    \litem $\Phi^a$ is spacelike at the points it doesn't
      vanish.
  \end{itemize}
\end{defn}

\begin{defn}\label{def:hel-def}
  An infinitesimal symmetry $X^a$ of a NEH $(\ih,q_{ab},D_a)$ is called helical if
  \begin{itemize}
    \litem The symmetry group generated by the projection $\hat{X}^A$
      of $X^a$ onto the base space $\bas$ is diffeomorphic to $SO(2)$,
    \litem in the maximal analytic extension $\bar{\ih}$ there exists an orbit of 
      the symmetry group generated by the extension $\bar{X}^a$ which is not closed 
      (i.e. it is diffeomorphic to a line).
  \end{itemize}
  A NEH admitting a helical infinitesimal symmetry will be called helical.
\end{defn}

An important property of the latter is that by the local rigidity theorem 
it induces on $\ih$ also a null and cyclic symmetry, that is
\begin{theorem}\label{thm:hel-decomp}
  Suppose the energy-momentum tensor $T_{ab}$ satisfies  the condition
  \eqref{eq:lie_T-1} for a non-vanishing null vector field $\ell^a$ tangent to a 
  non-expanding horizon $\ih$. 
  If  $\ih$ admits a helical infinitesimal symmetry $X^a$,
  then it also admits a cyclic infinitesimal symmetry $\Phi^a$
  and a null infinitesimal symmetry $\bsl^a$ such that
  \begin{equation}
    X^a\ =\ \Phi^a + \bsl^a \ .
  \end{equation}
\end{theorem}

Any cyclic symmetry admits in particular the choice of $\ih$ foliation preserved 
by it, that is
\begin{cor}\label{cor:axi-form}
  Suppose a non-expanding horizon $\ih$ admits a cyclic infinitesimal
  symmetry $\Phi^a$. Then there exists $\ell\in\Gamma(T(\ih))$ such that
  \begin{subequations}\label{eq:c-l}\begin{align}
    \sgr{\ell}\ &=\ \const\ ,  &
    [\Phi,\ell]\ &=\ 0\ .  \tag{\ref{eq:c-l}}
  \end{align}\end{subequations}
  Moreover, for the maximal analytic extension $\bar\ih$ of $\ih$ there exists 
  a diffeomorphism 
  \begin{equation}
    h:\bar\ih\ \rightarrow\ \bas\times\re
  \end{equation}
  such that
  \begin{equation}
    h_*\Phi\ =\ (\hat{\Phi},0) \ ,\quad
    h_*\ell\ =\ (0,\partial_v) \ ,
  \end{equation}
  where $\hat{\Phi}=\ppi_*\Phi$ and $v$ is the coordinate on $\mathbb{R}$.
\end{cor}

\subsubsection{The symmetries and invariants}

An element new with respect to the situations studied in \cite{lp-symm} is the presence 
and uniqueness of the invariants of NEHs introduced in section~\ref{ssec:invariants}.
It leads to the following results:  

\begin{theorem}\label{cor:symm-prop}
  Suppose that $(\ih,q_{ab},D_a)$ is an invariant-generic NEH (see Definition \ref{invgen}) 
  and $X$ is its infinitesimal symmetry. Let $\ell^a$ and $n_a$ are, respectively, the 
  invariant vector field and the invariant covector field. 
  Then
  \begin{enumerate}[a)]
    \item \label{it:cor-symm-prop-comm} $[X,\ell]\ =\ 0$, and $\lie_Xn_a\ =\ 0$.
    \item \label{it:cor-symm-prop-dec} there exists a constant $a\in\mathbb{R}$ and
      a Killing vector field $\hat{X}^A$ of the metric tensor $\hat{q}_{AB}$ induced 
      on the horizon base space $\bas$ such that
      \begin{equation}\label{eq:cor-symm-form}
				X^a\ =\ a\ell^a +\tilde{X}^a
      \end{equation}
      where $\tilde{X}^a$ is tangent to the leafs of the invariant foliation of $\ih$, 
      and
      \begin{equation}
				(\ppi_*\tilde{X})^A\ =\ \hat{X}^A \ .
      \end{equation}
  \end{enumerate}
\end{theorem}

\section{Spacetime neighborhoods of NEHs: induced structures}
  \label{sec:Bondi}

In the previous section we have defined  invariant structures  of 
generic, non-expanding horizons:  the invariant vector $\ell^a$, the
invariant foliation and the invariant covector $n_a$. This intricate structure
provides a natural extension to the spacetime neighbourhood of the horizon,
defining in particular the coordinate system analogous to the Bondi one defined
near null SCRI. This construction has been presented in \cite{l-sem} (and subsequently 
in \cite{k-spc}) in context of horizons in $4$-dimensional spacetime 
and in \cite{p-phd} in general dimension.
Here (Sec.~\ref{sec:neigh-invars}) we perform the $1$st step of 
this construction, extending the horizon invariants to the analogous invariant 
structures of its neighborhood. This extension will be next used in 
Sec.~\ref{sec:neigh-kill} to define a natural way of describing the (possibly 
present) Killing fields on the neighbourhood of $\ih$.

\subsection{The Bondi-like extension of a structure of a NEH.}
  \label{sec:neigh-invars} 

The construction of the coordinate system is graphically presented in fig.~\ref{fig:bondi}. 
The detailed specification of the construction is as follows.
Given a non-expanding horizon $\ih$ let us fix a null,  nowhere vanishing vector
field $\ell^a\in\Gamma(T(\ih))$ (not necessarily the invariant one), a foliation 
transversal to $\ell^a$ and preserved by its flow, and the corresponding covector 
$n_a$ orthogonal to the foliation and normalized by the condition
\begin{equation}
  \ell^a n_a\ =\ -1
\end{equation}
as in the previous section. The covector determines uniquely at
$\ih$ a null vector field $\bn^\mu$, tangent to the spacetime $\M$,
such that the pullback of $\bn_\mu$ onto $\ih$ equals $n_a$. At each
point of $\ih$, the vector $\bn^\mu$ is transversal. We extend the
vector  field $\bn^{\mu}$  to a null vector field defined in some
spacetime neighbourhood of the horizon by the  parallel transport
along the null geodesics tangent to $\bn^\mu$ at $\ih$. In other
words we extend the vector field $n^\mu$ from $\ih$ to a vector
field defined in a neighborhood of $\ih$ such that
\begin{equation}\label{eq:nM}
  \nabla_{\bn}\bn \ =\ 0.
\end{equation}
Assuming the vector  field $\ell^a$ at the horizon is future (past) oriented
the vector $\bn^\mu$ is also future (past) oriented. Due to the finiteness of
the transversal expansion  of the vector field $\bn^\mu$ at the
horizon, there exists some region $\M': \M\supset\M'\supset\ih$ of
the spacetime such that the geodesics generated by $\bn^\mu$ define the
foliation of $\M'$. \footnote{The range of a region $\M'$ strongly
  depends on the choice of the foliation the field $\bn^\mu$ is orthogonal
  to.}
We will denote the maximal (for given foliation of $\ih$) set of that
property as the {\it domain of transversal null
  foliation}.
In this region the field $\bn^\mu$ is defined uniquely.

Via the flow of this vector field, the vector field $\ell^a$ defined
on the horizon is extended to a vector field $\zeta^\mu$ defined on
$\M'$ such that
\begin{subequations}\label{eq:zeta}\begin{align}
  \lie_{\bn}\zeta^{\mu} &= 0  \ ,  &
  \zeta^{\mu}|_{\ih}=\ell^{\mu} \ . \tag{\ref{eq:zeta}}
\end{align}\end{subequations}
We will refer to this field as to the {\it Bondi-like extension} of
the  vector field $\ell$.

Generically, away from the horizon the vector field $\zeta^\mu$ is no longer null. 
Indeed, the Lie derivative of $\zeta^{\mu}\zeta_\mu$ along $\bn^{\mu}$ is at the 
horizon equal to
\begin{equation}\label{eq:zeta-norm}
  \lie_{\bn}\zeta^{\mu}\zeta_{\mu}|_{\ih} = 2\sgr{\bsl} \ .
\end{equation}
If $\sgr{\bsl}>0$ ($<0$) the vector $\zeta^\mu$ becomes timelike  near $\ih$, 
on the side from (into) which the geodesics defined by the vector $\bn^\mu$ are 
incoming (outgoing).  
Hence, it can be treated as a 'time evolution' vector co-rotating
with the horizon.\footnote{Exterior/interior of $\ih$ is undefined at this point.
One can define exterior to be that side of $\ih$ at which near to $\ih$ the vector
field is timelike.}

Finally,  the  foliation of the horizon $\ih$ is mapped by the flow of $\bn^\mu$ 
into a foliation with $n-2$ surfaces diffeomorphic to the corresponding slices of 
$\ih$. The resulting  foliation of the spacetime neighborhood of $\ih$ is preserved 
by the flow of the vector field $\zeta^\mu$ as well.

Also, there exists a uniquely defined function $r$ in the neighborhood of
$\ih$, such that
\begin{subequations}\label{eq:r-M}\begin{align}
  \bn^{\mu}r_{,\mu} &= -1 \ ,  &
  r|_{\ih} &= 0 \ , \tag{\ref{eq:r-M}}
\end{align}\end{subequations}
called  'radial' coordinate on $\M'$. Given a value of $r$ the flow
of $\bn$ maps the horizon $\ih$ into a cylinder $\ih_{(r)}$.
The radial coordinate $r$ provides affine parametrization  of the
null geodesics tangent to $\bn$. Each cylinder $\ih_{(r)}$ is formed by the integral 
curves of the vector field $\zeta$.

A parametrization of the integral curves of $\zeta$ can be fixed
uniquely up to a constant, as a function $v$ determined by its restriction to
$\ih$ and by
\begin{equation}\label{eq:v-M}
  \zeta^{\mu} v_{,\mu}\ =\ 1 \ , \qquad \bn^{\mu} v_{,\mu}\ =\ 0 \ .
\end{equation}
We are assuming that $v$ is constant on the leaves of the foliation
fixed on $\ih$.  
The condition $v=v_o$ defines an $n-1$ dimensional surface $\N_{v_o}$
in the neighborhood of $\ih$ consisting of the null geodesics
tangent to $\bn^\mu$.

Remaining $n-2$ coordinates  $(x^A)$ can be defined in the
neighborhood of $\ih$ as the extension of any properly defined
coordinate system $\hat{x}^A$ given on the base space $\bas$ of the
horizon
\begin{subequations}\label{eq:xA-M}\begin{align}
  \forall_{p\in\ih}\quad x^A(p) &:= \hat{x}^A(\ppi^{*} p) \ ,  &
  \zeta^{\mu}x^A{}_{,\mu} = \bn^{\mu}x^A{}_{\mu} &= 0 \ , \tag{\ref{eq:xA-M}}
\end{align}\end{subequations}
where $\ppi$ is a projection onto base space defined via
\eqref{eq:bas}. 

The set of $n$ functions $(x^A,v,r)$ defined above forms at $\M'$ a well defined 
coordinate system. The coordinates $v$ and $r$ are defined globally on $\M'$. On 
the other hand, the coordinates $\hat{x}^A$ are defined locally on elements of 
an open covering of $\bas$, next pulled back to $\ih$ and finally extended to 
$\M'$. Due to the similarity with the coordinate system defined by Bondi at the 
null SCRI it will be referred to as the {\it Bondi-like coordinate system} of the 
horizon spacetime neighbourhood. 
In these coordinates
\begin{equation}
  \zeta^\mu\ =\ (\partial_v)^\mu \ , \qquad 
  \bn^\mu\ =\ -(\partial_r)^\mu \ , \qquad 
  \bn_\mu\ =\ -(\rd v)_\mu \  .
\end{equation}

\begin{figure}
  \begin{center}
		\psfrag{IH}{$\ih$}
		\psfrag{Slc}{$\slc$}
		\psfrag{M}{$\M$}
		\psfrag{Trans}{$\N$}
		\psfrag{l}{$\ell$}
		\psfrag{n}{$n$}
		\psfrag{zeta}{$\zeta$}
    \includegraphics[width = .5\linewidth, angle = 0]{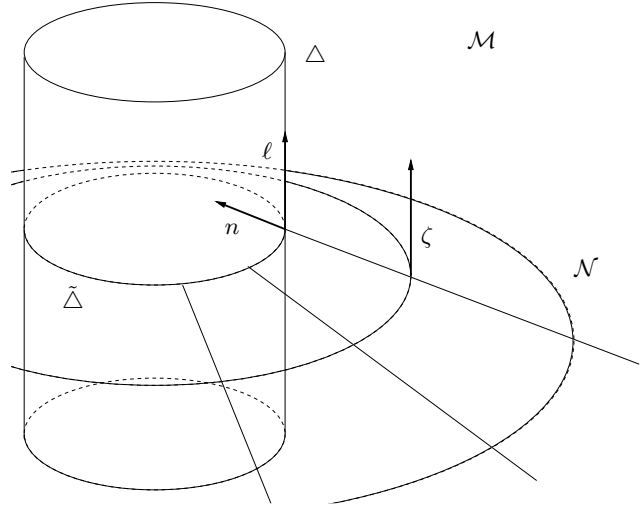}
    \caption{The construction of Bondi-like coordinate system.}
  \end{center} \label{fig:bondi}
\end{figure}
\subsection{Invariants of NEHs neighborhoods} 
We use the Bondi extension to endow the neighborhood of an invariant-generic NEH 
$\ih$ with invariant structures. Let the starting point for the construction of 
the previous subsection be the invariant vector field $\ell^a$  of the geometry 
of $\ih$, the invariant foliation and the invariant covector $n_a$. Then we get 
the following structures invariantly defined in the neighborhood of $\ih$: the
vector fields $\zeta^\mu$, and $\bn^\mu$, the foliation by the surfaces $\ih_{(r)}$, 
the foliation by the surfaces $\N_{v}$.
\begin{defn}\label{invneigh}
	We will call $\zeta^\mu$ the invariant vector field, $\bn_\mu$ the invariant
	covector field, and the foliations,  the invariant foliations,
	respectively of the neighborhood of $\ih$. And we will refer to the
	coordinates $(x^A, v,r)$ as to the Bondi invariant coordinates.
\end{defn}

\subsection{The invariants and the Killing vectors}
  \label{sec:neigh-kill}

The Bondi invariant coordinate system introduced in the previous
subsection in the neighborhood of $\ih$ is very well suited to
identify and characterize the Killing vectors possibly existing in the 
neighborhood of $\ih$.

\begin{lem}\label{[K,zeta]} 
  Consider  a non-expanding horizon $\ih$. Assume it is
  invariant-generic. Let
  $\zeta^\mu$ and $\bn_\mu$ be the invariant, respectively vector field
  and covector field of a neighborhood of $\ih$.  Suppose $K^\mu$ is a
  Killing vector field defined in the spacetime $\M$ and tangent to
  $\ih$. Then,
  \begin{enumerate}
    \item the restriction of $K^\mu$ to $\ih$ is an infinitesimal symmetry
      of the geometry of $\ih$, and
    \item $\hphantom{.}$\vspace{-0.75cm}
      \begin{equation}
        [K,\zeta] \ =\ [K,\bn]\ =\ 0.
      \end{equation}
  \end{enumerate}
\end{lem}

The first item is obvious and the second follows from the fact that there
is an isometric flow of $K^\mu$ defined in a spacetime neighborhood of
any cross-section of $\pi:\ih\rightarrow\hat{\ih}$.

The infinitesimal symmetries were characterized in terms of the
invariant vector $\ell^a$ and the invariant foliation in the previous
section. Owing to Lemma \ref{[K,zeta]} we can characterize the
Killing vectors by the invariant vector field $\zeta^\mu$ and the
invariant foliations, that is by the Bondi invariant coordinates.

\begin{theorem}\label{cor:kill-form}
  Suppose the assumptions of Lemma \eqref{[K,zeta]} are satisfied. Then
  \begin{enumerate}[(i)]
    \item \label{it:cor-kill-form-null} if the restriction of $K^\mu$ to $\ih$ is null 
      everywhere on $\ih$, then there is a constant $a_0\in\mathbb{R}$ such that
      \begin{equation}\label{eq:k-l}
				K^\mu\ =\ a_0\zeta^\mu;
      \end{equation}
    \item If $K^\mu$ is not null at $\ih$, then it  takes the following form in 
			the Bondi-like invariant  coordinate system system $(x^A,r,v)$
      \begin{equation}\label{eq:Kill-chiral-form}
				K^\mu\ =\ a_0\zeta^\mu \ +\ [\Phi^A(x^B)\partial_A]^\mu  \ ,
      \end{equation}
      where 
      \begin{equation} 
				k^a\ =\ a_0\ell^a \ +\ [\Phi^A(\hat{x}^B)\partial_{\hat{x}^A}]^a, 
				\qquad  
				a_0\in\mathbb{R}
      \end{equation}
			is an infinitesimal symmetry of $\ih$.     
  \end{enumerate}
\end{theorem}
\begin{proof}
The conclusion follows from the following calculation 
	\begin{equation}
    0\ =\ [n,K]\ =\ -[\partial_r,K^\mu\partial_\mu]\ =\ -\partial_r(K^\mu)\partial_\mu
  \end{equation}
  of which a direct consequence is 
	\begin{equation}
    K^\mu(x^A,v,r)\ =\ K^\mu(x^A,v,0)=k^\mu(x^A,v) \ .
  \end{equation}
\end{proof}

To see the strength of this result, let us note that, given the Bondi invariant 
coordinates $(x^A,v,r)$, and the Killing vectors 
$\hat{k}_1^A(\hat{x}^B)\partial_{\hat{x}^A},...,\hat{k}_m^A(\hat{x}^B)\partial_{\hat{x}^A}$ 
of the horizon base space $\bas$ geometry $\hat{q}_{AB}$, if we want to know if there
is a spacetime Killing vector field tangent to the horizon, all we have to do
is to check candidates $K^\mu$ of the form 
\begin{equation}
  K\ =\ a_0\partial_v \ +\ (a^1\hat{k}_1^A+...+a^m\hat{k}_m^A)
  \partial_{x^A} 
\end{equation}
for all  $a_0,...,a_m\in \mathbb{R}$.

\subsection{The Killing vectors in non-invariant Bondi coordinates}
	\label{sec:neigh-killnonin}
	
Even if we are not assuming that a NEH $\ih$ is invariant-generic, a Killing 
vector field still takes a simple form in suitably chosen Bondi-like coordinates.  
\begin{theorem}\label{thm:Kill-hel}
  Consider a non-expanding horizon $(\ih,q_{ab},D_a)$ contained in a spacetime 
  $(\M,g)$; suppose $(\M,g)$ admits a Killing field $K^\mu$ tangent to $\ih$; 
  suppose there is a nowhere vanishing vector field $\ell^\mu\in\Gamma(T(\ih))$ 
  such that the restriction $k^a$ of $K^\mu$ to $\ih$ satisfies
  \begin{equation} 
    [\ell,k]\ =\ 0 
  \end{equation}
  and a foliation of $\ih$ by sections of $\pi:\ih\rightarrow\bas$ preserved by 
  the flow of $k^a$.
  Then, there are coordinates $(\hat{x}^A,v)$ on $\ih$ such that 
  \begin{equation} 
    k\ =\ a\partial_v + b\Phi^A(x^B)\partial_A \ ,\qquad  a,b\in\mathbb{R}\ , \label{eq:k}
  \end{equation}
  and in the corresponding Bondi like coordinates $(x^A,r,v)$ the Killing field $K^\mu$ 
  is necessarily of the form: 
  \begin{equation}\label{eq:K-ch}
    K\ =\ a\partial_v + b\Phi^A(x^B)\partial_A \ , \qquad
    a,b = \const \ .
  \end{equation}
\end{theorem}

\section{$4$d Einstein-Maxwell neighbourhoods of 3d NEHs} 
  \label{sec:4d-neigh}

In the remaining part of the article we restrict our studies to the spacetime 
neighbourhood of the horizon being of dimension $4$. The focus point of this 
section is the analysis of the spacetime metric expansion in the radial (the 
vector field $\bn^\mu$) direction first without restrictios on the spacetime 
matter content and next assuming that the only matter admitted is the Maxwell 
field. In the latter case we discuss the characteristic Cauchy problem on $\ih$.

We start with a short summary of the structure and evolution equations of the 
Maxwell field on a NEH. Next we provide a short comprehensive discussion of the 
Cauchy problem on electro-vacuum NEH.

Although the description of an electrovac neighbourhood of a NEH and the initial 
value problem corresponding to it can be (and are) formulated in the geometric 
formalism used so far in the article, it is particularly convenient (especially 
in the task of providing the metric expansion) to perform the analysis using 
appropriately chosen Newman-Penrose null frame (see Appendix \ref{app:NP}). 
Therefore, we introduce such frame (in subsection \ref{sec:null-frame}) adopted 
to the Bondi-like coordinate system defined in Section~\ref{sec:Bondi} further 
providing the dictionary between the formerly used geometric quantities and the
new objects defined by the frame.
This construction is then used to analyze the radial expansion of the spacetime 
metric and determine the initial data (at the horizon) necessary to uniquely specify 
the terms of expansion for any matter content and next, with restriction to 
electrovac spacetime.

Next, in subsection \ref{sec:evac-spc} we restrict our interest solely to electrovac
horizon neighbourhoods. To do so we reintroduce the dictionary between the two formalisms 
(geometric and the null frame one) provided earlier, this time adopting it to context 
at hand. Using it we discuss the properties of Maxwell field equations in the horizon 
neighbourhood. These properties are next used to reexamine the set of initial data 
needed to specify the expansion of metric previously studied in general matter case. 

The section is concluded by the detailed discussion of the characteristic
initial value problem for electrovac case with initial surfaces
defined by Bondi-like coordinate extension (subsection \ref{sec:ivp}) and with use
of the null frame introduced in Section~\ref{sec:null-frame}.
In all the subsections for the reader's convenience we just 
present the results, showing the calculations in the appendices
\ref{app:4d-neigh} and \ref{app:evac-spc}.

\subsection{Electromagnetic field on a NEH.}

Suppose the spacetime $\M$ is $4$ dimensional, and an electromagnetic field $F_{\mu\nu}$
is present in a neighbourhood of $\ih$. We will be considering the Einstein-Maxwell
equations in detail in section \ref{sec:evac-spc} in terms of the Newman-Penrose
components. But before that, let us recall the equivalent geometric description 
of the constraints imposed on the components of $F_{\alpha\beta}$ by the Maxwell 
vacuum equations \cite{lp-extr}.
 
As in the earlier parts of this article, given a spacetime covetor $W_\alpha$, 
its pullback on $\ih$ is denoted by $W_a$. Similarly for any covector $W_a$ 
orthogonal at $p\in \ih$ to $\ell^a$, we denote the corresponding element of 
$(T_p\ih/L)^*$ by $W_A$. 

The electromagnetic field energy-momentum tensor $T_{\alpha\beta}$ satisfies Stronger 
Energy Condition \ref{c:energy}, hence  $T_{ab}$  satisfies (\ref{eq:lie_T-1}). 
The tensor $T_{ab}$ contributes to the constraints (\ref{eq:[l,D]1}), and the condition (\ref{eq:lie_T-1}) amounts to the vanishing of some components of $F_{\alpha\beta}$ 
at $\ih$, namely
\begin{equation} 
  \ell^a(F-i\star_{\M}F)_{ab}\ =\ 0 \ , \label{Fconstr1}
\end{equation}
where $\star_{\M}$ is the spacetime Hodge star. 
 
We assume that $F_{\mu\nu}$ satisfies in a spacetime neighbourhood of $\ih$ the 
Maxwell vacuum equations
\begin{equation}\label{eq:Maxwell-vac} 
  \rd (F-i\star_{\M}F)\ =\ 0 \ . 
\end{equation}

The Maxwell equations constrain further the remaining components of $F_{\mu\nu}$ 
at $\ih$. In the section~\ref{sec:frame-ih} we derive these constraints by using 
a null frame adapted to $\ih$. Here we express the result in a frame independent 
way.

The first of Maxwell constraints reads 
\begin{equation} 
  \ell^a\star_{\M}\rd(F-i\star_{\M}F)_a \ =\ 0 \ . \label{Fconstr2} 
\end{equation}
Therefore the covector $\star_{\M}\rd(F-i\star_{\M}F)_a$ is orthogonal to $\ell$
and can be subject to the horizon Hodge dualization $\star_\ih$ introduced in 
Section~\ref{ssec:neh-def}. The second constraint following from the Maxwell equations 
takes the form of the horizon self-duality condition 
\begin{equation} 
  \star_\ih(\star_{\M}d(F-i\star_{\M}F))_A\ =\ i \star_{\M}d(F-i\star_{\M}F)_A \ .  
  \label{Fconstr3} 
\end{equation}

\subsection{Geometry of the Cauchy problem on $\ih$}
	\label{sec:cauchy-geom}
	
As a null surface, a NEH $\ih$ is not a part of a typical surface to formulate the
initial value problem for the Einstein - matter field equations -- a Cauchy problem 
for general relativity. A suitable formulation is known as the 
characteristic Cauchy problem \cite{Friedrich-evac1,*Friedrich-evac2,Rendall}. 
We will apply it to the NEHs in this subsection, using the null frame approach to 
gravity called the Newman-Penrose framework. Here, we express the outline of the 
results in a geometric, frame independent way. The details are presented in further 
subsections: \ref{sec:null-frame} through \ref{sec:evac-spc} and in 
Appendix~\ref{app:evac-spc} .   
   
Consider a NEH $(\ih,q_{ab},D_a)$ in a spacetime $(\M,g_{\alpha\beta})$.
The Bondi-like coordinate systems (introduced in section 3.1) adapted to $\ih$ 
form an atlas on the neighborhood of $\ih$. Let $(x^A,v,r)$ be one of these coordinate 
systems.

At any point of $\ih$, the spacetime metric tensor can be written in the form
\begin{equation}
  g_{\alpha\beta} \rd x^\alpha\otimes \rd x^\beta \ 
  =\ q_{AB} \rd x^A\otimes \rd x^B\ -\ \rd v\otimes \rd r-\rd r\otimes \rd v \ .
\end{equation}
Furthermore, the derivative $\partial_rg_{\alpha\beta}$ for all $\alpha,\beta =1,2,3,4$ 
is determined at $\ih$ by the components of the connection $D_a$ (see (\ref{eq:N-fc-ih}) 
below). The higher derivatives $\partial^k_rg_{\alpha\beta}$ are determined by the 
Einstein-matter field equations and data defined on a {\it 2-dimensional slice} 
$\tilde{\ih}$ of $\ih$ \cite{Friedrich-evac1,*Friedrich-evac2}.

\subsubsection{The Einstein vacuum case: pure gravitational field}

Suppose, in a neighbourhood of $\ih$ the vacuum Einstein equations hold.
Then, the horizon geometry $q_{ab},D_a$ satisfies the constraints \eqref{eq:[l,D]1}
with arbitrary null vector field $\ell^a = (\partial_v)^a$ and  
\begin{equation}
  \Ricn{4}_{ab}\ =\ 0 \ . \label{Ric=0}
\end{equation} 
To determine  the derivatives $(\lie_{\ell})^kg_{\alpha\beta}$, $k=2,\ldots,n$ at 
$\ih$ it is sufficient to know at a slice of $\ih$,
\begin{equation}\label{eq:slc-def-4d}
  \tilde{\ih}_{v_0} = \{ {\rm x}\in\ih :\  v({\rm x})=v_0 \}
\end{equation} 
(assuming it is a global section of $\ppi:\ih\rightarrow\bas$), the shear (see 
\ref{G32} below) $\lambda$ of the (transversal to $\ih$, orthogonal to the slice, 
null vector field) $n^\mu=(\partial_r)^\mu$, and all its derivatives 
\begin{equation}
  \partial_r^k\lambda, \qquad k=0,\ldots,n-2 \ .
\end{equation} 
The above statement can be demonstrated explicitly (``exactly'') by solving the 
hierarchy of the ordinary differential equations (ODE's) - see the next subsections and 
Appendix~\ref{app:frame-exp}.\footnote{Actually, the initial data the geometry 
	$(q_{ab},D_a)$ defines on the slice is also sufficient to determine $(q_{ab},D_a)$ 
	on the whole $\ih$).
} 

In order to determine  the spacetime metric $g_{\alpha\beta}$ 
in some 4-dimensional region containing $\ih$, we can use the $3$-dimensional surface 
in $\M$: $\N_{v_0}=\{ {\rm x}\in\M :\  v({\rm x})=v_0 \}$ spanned by the null geodesics 
tangent to the vector field $\bn^\mu$ and intersecting the slice  $\tilde{\ih}_{v_0}$ 
of $\ih$ specified in \eqref{eq:slc-def-4d}. Then, the following data:
\begin{enumerate}[\hspace{0.5cm}1)]
  \item on $\ih$:  $q_{ab},D_a$ such that (\ref{eq:[l,D]1}) with $\Ricn{4}_{ab}=0$ 
    \label{it:ivp-vac-geom}
  \item on $\N_{v_0}$: $\lambda$ with certain boundary condition specified on 
		$\lambda|_{\ih\cap\N_{v_0}}$ 
    induced by $D_a$ (see \eqref{eq:[l,D]1}) \label{it:ivp-vac-lambda}
\end{enumerate}
determine the spacetime metric tensor (modulo diffeomorphisms) in the domain of 
dependence of $\ih\cup \tilde{\ih}_{v_0}$. Moreover, as we vary the vacuum spacetime 
metric tensor in such a way that $\ih$ is a NEH, and $\N_{v_0}$ satisfies the definition, 
the data \ref{it:ivp-vac-geom})-\ref{it:ivp-vac-lambda}) above ranges all 
the possible NEH geometries on $\ih$ such that (\ref{eq:[l,D]1},\ref{Ric=0}), and 
all possible functions $\lambda:\N_{v_0}\rightarrow\mathbb{C}$ which satisfy the boundary 
condition at $\ih\cap\N_{v_0}$.

\subsubsection{The Einstein-Maxwell vacuum case: pure gravitational and Maxwell field}  

Suppose now, in the neighbourhood $\M'$ the vacuum Einstein-Maxwell equations hold.
Then, the horizon geometry $(q_{ab},D_a)$ and the pullbacks $F_{ab}$, $\star_\M F_{ab}$
satisfy the constraints (\ref{Fconstr1}, \ref{Fconstr2}, \ref{Fconstr3}).      
At every point of $\ih$, all the transversal derivatives $(\lie_{\ell})^kg_{\alpha\beta}$ 
as well as $(\lie_{\ell})^kF_{\alpha\beta}$ can be determined by certain components 
of the horizon geometry and electromagnetic field on it (see the next subsections and 
Appendix~\ref{app:evac-exp}). 
More precisely, the initial value problem is again the characteristic Cauchy problem 
with the initial data null surfaces $\ih$, $\N_{v_0}$ specified exactly as in previous 
sub-subsection. Then in order to determine the spacetime metric $g_{\alpha\beta}$ and 
the electromagnetic field $F_{\alpha\beta}$ in the domain of dependence of 
$\ih\cup \N_{v_0}$, it is sufficient to specify the following data:
\begin{enumerate}[\hspace{0.5cm}1)]
  \item on $\ih$: $q_{ab},D_a$,  $F_{ab}$ and $\star_\M F_{ab}$ such that 
		(\ref{Fconstr1}, \ref{Fconstr2}, \ref{Fconstr3}) hold,
    \label{it:ivp-em-geom}
  \item on $\N_{v_0}$: $\lambda$ with the boundary data $\lambda|_{\ih\cap\N_{v_0}}$ 
    induced by $D_a$ and the pullbacks $\bn^\mu F_{\mu a}$ and $\bn^\mu\star_\M F_{\mu a}$
    satisfying the boundary conditions induced by the Maxwell equations \eqref{eq:Maxwell-vac}.
    \label{it:ivp-em-lambda}
\end{enumerate}
Moreover, as we vary the vacuum solutions $g_{\mu\nu}$ and $F_{\mu\nu}$ to the 
Einstein-Maxwell equations such that $\ih$ is a NEH and $\N_{v_0}$ satisfies 
the definition, the data \ref{it:ivp-em-geom})-\ref{it:ivp-em-lambda}) above 
ranges: \ref{it:ivp-em-geom}) all the NEH geometries and the electromagnetic 
fields on $\ih$ such that (\ref{Fconstr1}, \ref{Fconstr2}, \ref{Fconstr3}) and, 
\ref{it:ivp-em-lambda}) all the possible functions $\lambda:\N_{v_0}\rightarrow\mathbb{C}$ 
which satisfy the appropriate boundary condition at $\ih\cap\N_0$ listed in point 
\ref{it:ivp-em-geom}) and one-forms $\bn^\mu F_{\mu a}$ and $\bn^\mu\star_\M F_{\mu a}$ 
defined on $\N_{v_0}$ also satisfying the appropriate boundary condition at $\ih\cap\N_0$ 
listed in point \ref{it:ivp-em-lambda}).

Below we present the expansion, specification of the data and the boundary conditions 
in detail using the Newman-Penrose framework.

\subsection{The null frame, the metric expansion} 
  \label{sec:null-frame}  

Our starting point is the Bondi-like extension of the structures and coordinates 
defined in Section~\ref{sec:neigh-invars} on a NEH $\ih$: a null, nowhere vanishing 
vector field $\ell^a$ tangent to $\ih$, a function $v:\ih\rightarrow \mathbb{R}$ such 
that $\ell^aD_av=1$, and coordinates $(x^A,v)$ on $\ih$ such that $\ell^aD_ax^A=0$. 
Whereas $v$ is globally defined, the coordinates  $x^A$ are defined locally, form  
an atlas, the pullback by $\ppi^*$ of atlas $\hat{x}^A$ defined on the base manifold $\bas$. 
The function $v$ defines on $\ih$ the covector $n=-\rd v$, which in the neighborhood 
$\M'$ of $\ih$, defines in particular the null vector field $\bn^\mu$. Using this 
structure we define below a null frame (see Appendix \ref{app:NP} for the
basic properties of null frames) $ (e_1,\,e_2,\,e_3,\,e_4) = (m,\,\bar{m},\, \bn,\, \ell)$ 
such that 
\begin{equation}\label{eq:frame-e3def}
  (e_3)^\mu = \bn^\mu \ \ {\rm in}\ \M',
  \qquad  
  (e_4)^\mu\ =\ \ell^\mu\ \ {\rm at}\  \ih,
\end{equation}
and $e_1, e_2$ are tangent to the leaves of the foliation of $\ih$.
Whereas the vector fields $e_3$ and $e_4$ are defined on the entire neighborhood $\M'$, 
the domains of the vector fields $e_1$ and $e_2$ will coincide with those
of the coordinates $(x^A)$. By $(e^1,e^2,e^3,e^4)$ we will be denoting the dual coframe.  
This construction of the frame has been already discussed in \cite{l-sem,p-phd} and
subsequently presented in \cite{k-spc}.

The spacetime metric tensor $g_{\mu\nu}$ on $\M'$ and the degenerate metric tensor 
$q_{ab}$ induced on $\ih$ take in that frame the following form: 
\begin{subequations}\begin{align} 
  g_{\mu\nu}\ &=\ \left(e^1\otimes e^2 + e^2\otimes e^1 
      - e^3\otimes e^4 - e^4 \otimes e^3\right)_{\mu\nu} \ , \\
  q_{ab}\ := \ g_{ab}\ &=\  (e^1 \otimes e^2 + e^2\otimes e^1)_{ab}. 
\end{align}\end{subequations} 
where $(\hspace{0.5mm}\cdot\hspace{0.5mm})_{ab}$ stands for the pull-back to $\ih$ 
of a tensor originally defined onto $\M$.

\subsubsection{Geometry and constraints at the horizon, the invariants}
  \label{sec:frame-ih}

Let us now focus on the properties of the frame at $\ih$ itself. For this purpose 
through this sub-subsection we will adopt the shortened notation, using `$=$' for 
`$=|_{\ih}$'.
In the Bondi-like coordinates $(x^A,v,r)$ defined in Section \ref{sec:neigh-invars}, 
$r=0$ on $\ih$, and at $\ih$ the vector field $\partial_v$ is null
\begin{equation}\label{eq:frame-lv}
  \ell^a \ =\ (\partial_v)^a
\end{equation}
and has constant surface gravity.
The real vectors $\Re(m)^\mu$, $\Im(m)^\mu$ are (automatically) tangent 
to $\ih$. To adapt the frame further, we assume the vector fields
$\Re(m)^a, \Im(m)^a$  tangent to the constancy surfaces $\slc_v$ of 
the coordinate $v$  \eqref{eq:v-M} are Lie dragged by the flow $[\ell]$ 
\begin{equation}\label{m} 
  \lie_{\ell} m^a\ =\ 0 \ . 
\end{equation} 
This implies immediately, that the projection of $m^a$ onto $\bas$
uniquely defines on a horizon base space $\bas$ a null vector frame
$(\hm,\hmb)$ and the differential operators $\delta$, $\bar{\delta}$
\begin{subequations}\label{eq:base-frame}\begin{align} 
  (\ppi_* m)^A \ &=:\ \hm^A \ ,  &
  \m\ &:=\ \hm^A(x^B)\partial_A \ 
\tag{\ref{eq:base-frame}}\end{align}\end{subequations}
corresponding to the frame vectors.

The frame specified above is adapted to: the vector field $\ell^a$,
the flow of $\ell^a$  invariant foliation of $\ih$, and the null complex-valued
frame $\hm^A$ defined on the manifold $\bas$. Spacetime frames
constructed in this way on $\ih$ will be called {\it adapted}.


Since all the frame elements are Lie dragged by $\ell^a$, 
the connection $D$ induced on $\ih$ can be decomposed as
follows\footnote{The decomposition is consistent with the definition
  of connection coefficients presented in \ref{app:NP}.
}
\begin{subequations}\label{G}\begin{align} 
  \label{G12} 
    m^\nu D \bar{m}_\nu\ &=\ 
    \pi^* \left({\hm}^A\hD {\hmb}{}_A\right)\ =: \ppi^*\hat{\Gamma},\\ 
  \label{G43} 
    -n_\nu D \ell^\nu\ &=\ \w{\bsl} = 
    \pi e^2_{(\ih)} + \bar{\pi}e^1_{(\ih)} 
    + \sgr{\bsl} e^3_{(\ih)},\\ 
  \label{G32} 
    -\bar{m}^\nu D n_\nu\ &=\  \mu e^1_{(\ih)} +\lambda e^2_{(\ih)} 
    + \pi e^4_{(\ih)},\\ 
  m_\mu D \ell^\mu\ &= \ 0, 
\end{align}\end{subequations} 
where $\hat{\Gamma}$  is the Levi-Civita connection 
$1$-form corresponding to the covariant derivative $\hD$
defined by $\hq$ and to the null frame  $\hm^A$ defined on $\bas$
\begin{equation}\label{2gamma} 
  \hat{\Gamma}\ =:\ 2\bar{a} \hat{e}^1 + 2{a} \hat{e}^2 \ .
\end{equation} 

The rotation $1$-form potential $\w{\bsl}$ in the chosen frame takes
the form 
\begin{equation}\label{eq:w-frame} 
  \w{\bsl}\ =\  \pi e^2_{(\ih)} + \bar{\pi}e^1_{(\ih)} 
            - \sgr{\ell}e^4_{(\ih)}\ , 
\end{equation} 

In terms of the coordinates $(x^A,v)$ on $\ih$,  the functions $a$ and $\pi$ satisfy
\begin{equation}\label{eq:frame-api-evo}
  \partial_va\ =\ \partial_v\pi\ =\ 0 \ .
\end{equation}

The Ricci tensor is represented by the set of the Newman-Penrose coefficients 
$\Phi_{ij}$, $i,j=0,1,2$ \eqref{eq:NP-WR}. In terms of them, the constraints induced 
on the horizon geometry $(q_{ab,D_a})$ by the Einstein field equations described in
section \ref{ssec:degrees} are by the identity \eqref{eq:[l,D]1} equivalent to the 
following set of equations  
\begin{subequations}\label{eq:constr-evac-4D}\begin{align} 
  \label{Dmu} 
    8\pi G( T_{m\bar{m}} - \frac{1}{2}Tq_{m\bar{m}}) \  =\ -2(\Phi_{11}+3\Ricn{4}{})\  
    &=\ 2D\mu + 2\sgr{\bsl}\mu - \tdiv\tw{\ell} 
        - |\tw{\ell}|^2_{\tq} + \Ricn{2}_{m\bar{m}} \ ,\\ 
  \label{Dlambda} 
    8\pi G T_{\bar{m}\bar{m}} =\ -2\Phi_{20}\    
    &=\ 2D\lambda + 2\sgr{\ell}\lambda - 2\mb \pi 
        - 4a\pi - 2\pi^2 \ , 
\end{align}\end{subequations}  
where $D:=\ell^{a}\partial_a$, $\m:=m^a\partial_a$, $\Ricn{2}_{m\bar{m}}:=(\ppi^{\star}\Ricn{2})_{ab}m^a\bar{m}^b$ 
and $(\tdiv\,\tw{\ell})$ is and the divergence of projected rotation $1$-form 
\eqref{eq:omega_def}. As functions 
of the variables $x^A$ (\ref{eq:xA-M}), by \eqref{eq:lie_q} and \eqref{eq:lieellomega} 
the latter two objects equal their counterparts $(\Ricn{2}_{m\bar{m}},\hdiv\,\hw{\ell})$ 
defined on the horizon base space 
\begin{subequations}\label{eq:Kdw}\begin{align} 
  \Ricn{2}_{m\bar{m}} \ &:= \ 2\m a + 2\mb \bar{a} - 8 a  \bar{a} \ ,  \\
  \hdiv\hw{\ell}\ 
    &= \ \m\pi + \mb\bar{\pi} - 2a\bar{\pi} - 2\bar{a}\pi \ .
\end{align}\end{subequations} 
\begin{rem}
  In terms of the Newman-Penrose coefficients, Definition~\ref{def:natural} 
  of the natural vector field $\ell^a$ of a NEH geometry reads: $\ell^a$ is tangent 
  to $\ih$, null, $\sgr{\ell}=1$ and
  \begin{equation} \label{eq:frame-muevo}
    \ell^a \mu_{,a}\ =\ 0.
  \end{equation}
  The invariant foliation listed in Definition \ref{invgen} and the corresponding 
  invariant variable $v$ is defined by the following condition 
  \begin{equation}
    \m\pi + \mb\bar{\pi} - 2a\bar{\pi} - 2\bar{a}\pi\ =\ 0.
  \end{equation}
\end{rem}

Furthermore, the condition (\ref{eq:lie_T-2}) reads
\begin{equation}\label{eq:lie_T-2-NP}
  D\Phi_{20}\ =\ D(\Phi_{11}+3\Ricn{4}{})\ =\ 0 \ ,
\end{equation}
whereas the functions $\Phi_{11}$, $\Phi_{20}$ and $\Ricn{4}{}$ are 
determined by the electromagnetic field (which we will see below).

Some components the energy-momentum and Weyl tensor vanish due to the Stronger 
Energy Condition \ref{c:energy}. Indeed, the following Ricci tensor components 
(listed in \eqref{eq:NP-stronger}) vanish on $\ih$ due to \eqref{RicclX}, and 
the Weyl tensor components (listed in \eqref{eq:NP-weaker}) due to the definition 
of NEH and the Bianchi equalities (see \eqref{eq:NP-WR} for the definition of 
components) 
\begin{subequations}\label{eq:NP-matt-gauge}\begin{align}  
  \label{eq:NP-stronger}
    \Phi_{00}\ =\ \Phi_{01}\ = \Phi_{10}\ &=\ 0 \ .\\
  \label{eq:NP-weaker}
    \Psi_0\ =\ \Psi_1\  &=\ 0 \ , 
\end{align}\end{subequations}
Moreover, the horizon geometry and the matter fields at $\ih$ determine the values 
at the horizon of the Weyl tensor components $\Psi_2$ and $\Psi_3$ (via the NP 
equations \eqref{eq:E-Psi2} and \eqref{eq:E-Psi3} respectively),
\begin{subequations}\begin{align}  
  \label{Psi2horizon}
    \Psi_2\ &=\ -\frac{1}{4}\Ricn{2} -\frac{1}{2}(\delta\pi-\bar{\delta}\bar{\pi}) 
    -a\bar{\pi}+\bar{a}\pi +\Phi_{11}+\frac{1}{24}\Ricn{4}\\
  \label{Psi3horizon}
    \Psi_3\ &=\ \bar{\delta}\mu-\delta\lambda +\pi\mu 
    +(4\bar{a}-\bar{\pi})\lambda+\Phi_{12} \ , 
\end{align}\end{subequations}
The remaining one: $\Psi_4$, is constrained by Bianchi identity \eqref{eq:B-DPsi4}
which at the horizon reads
\begin{equation}\label{eq:DPsi4}\begin{split}
  D\Psi_4\ =\ &-2\sgr{\bsl}\Psi_4 + \mb \Psi_3
    + (5\pi + 2a) \Psi_3 - 3\lambda \Psi_2 \\
  \ &- \bar{\mu}\Phi_{20} + (\pi+2a)\Phi_{21} - 2\lambda\Phi_{11}
    - \Phi_{20}{}_{,r} \ .
\end{split}\end{equation}
In this way, $\Psi_4$ at $\ih$ is uniquely determined by the value of
$\Psi_4$ on chosen section, the horizon geometry and the matter fields.

\subsubsection{Extension to the spacetime neighbourhood}
  \label{sec:frame-ext}

Given the coframe $(e^1,..,e^4)$ dual to the frame $(m,\bar{m},n,\ell)$
defined above at $\ih$, the condition 
\begin{equation}\label{eq:frame-ext}
  \nabla_{n}e^{\mu} = 0
\end{equation}
defines its unique extension to the spacetime neighbourhood $\M'$.
The corresponding connection coefficients $g(e_\gamma,\nabla_\alpha e_\beta)$
are defined in (\ref{eq:NP-con}).  In
the Bondi-like coordinate system this adapted (co)frame extended by
\eqref{eq:frame-ext} takes the form 
\begin{subequations}\label{eq:frame-def}\begin{align}
  e_1 &= m = \bar{e_2} = m^A(\partial_A+Z_A\partial_r) \ ,  &
  e^1 &= \bar{e}^2 = \bar{m}_A\rd x^A+X\rd v \ , \label{eq:m-def} \\
  e_3 &= n = -\partial_r \ , &
  e^3 &= -\rd r + Z_A\rd x^A + H\rd v \ , \\
  e_4 &= \ell = \partial_v-\bar{X}e_1-Xe_2+H\partial_r \ , &
  e^4 &= \rd v \ ,
\end{align}\end{subequations}
where $Z_A,H$ are real functions, $m^A,X$ are complex. At the horizon
these functions take the following values
\begin{subequations}\label{eq:fc-H}\begin{align}
  X|_{\ih}\ =\ H|_{\ih}\ =\ Z_A|_{\ih}\ &=\ 0  &
  m_A|_{\ih}\ &=\ \pi^*\hat{m}_A \ . \tag{\ref{eq:fc-H}}
\end{align}\end{subequations}

The condition \eqref{eq:frame-ext} (consistent with $\nabla_nn^\mu=0$ and 
$n^\mu n_\mu=$const) imposes on the connection coefficients (defined via 
\eqref{eq:NP-con}) corresponding to it the following constraints true
in $\M'$ 
\begin{equation}\label{eq:NP-zeros}
  \tau\ =\ \gamma\ =\ \nu\ =\ \mu-\bar{\mu}\
    =\ \pi-(\alpha+\bar{\beta})\ =\ 0 \ .
\end{equation}
In particular the last constraint allows us to express the
coefficients $(\alpha,\beta)$ in terms of $\pi$ and 
\begin{equation}\label{eq:NP-adef} 
  a\ :=\ \fracs{1}{2}(\alpha-\bar{\beta}) \ .
\end{equation}
The commutators of the differential operators corresponding to the
frame vectors can be expressed in terms of the functions
$(H,X,m^A,Z_A)$ and their derivatives. On the other hand they are
determined by the connection coefficients via \eqref{eq:NP-comm}.
That correspondence leads to the constraints on the frame coefficients
which determine their evolution of the functions $(X,H,m_A,Z_A)$ along the 
transversal to $\ih$ null geodesics:
\begin{subequations}\label{eq:N-fc-gen}\begin{align}
  \label{eq:n-X-gen}
    -\nr X\ &=\ \bar{\pi} + \mu X + \bar{\lambda}\bar{X} \\
  \label{eq:n-H-gen}
    \nr H\ &=\ (\epsilon+\bar{\epsilon}) + \pi X + \bar{\pi}\bar{X}\\
  \label{eq:n-mA-gen}
    \nr m_A\ &=\ \bar{\lambda}\bar{m}_A + \mu m_A \\
  \label{eq:n-mZ-gen}
  \nr Z_A\ &=\ \pi m_A + \bar{\pi}\bar{m}_A
\end{align}\end{subequations}
This set is supplemented by analogous evolution equations for the spin (connection) 
coefficients (\ref{eq:ODE-n-mu}-\ref{eq:ODE-n-varpi}, \ref{eq:ODE-n-rho}, 
\ref{eq:ODE-n-sigma}, \ref{eq:ODE3}) and Weyl tensor components (\ref{eq:ODE-n-Psi3}, 
\ref{eq:ODE-n-Psi2}, \ref{eq:n-Psi0-N}, \ref{eq:ODE-n-Psi1}).

The global structure of the resulting frame is as follows: The neighborhood $\M'$ 
of a given NEH $\ih$ is covered by open sets ${\cal U}_I$, $I=1,\ldots,K$ obtained 
from a covering ${\cal \hat{\cal U}_I}$, $I=1,\ldots,K$ of the base $\bas$. Each open set  
${\cal U}_I$ is the union of the null geodesics tangent to $n^\mu$ or  to $\ell^\mu$, 
and intersecting the set $\hat{\cal U}_I$, for every $I=1,\ldots,K$.

\subsubsection{Metric expansion at the horizon}
  \label{sec:frame-exp}

Since in this sub-subsection we consider objects on $\ih$ only, we again adopt the 
notation $ `='\ \equiv\ `=|_\ih'$.

It is a straightforward observation that the horizon geometry $(q_{ab},D_a)$ already
determines the frame components at $\ih$ (through \eqref{eq:fc-H}) as
well as their $1$-order radial derivative $\partial_r$ (via \eqref{eq:N-fc-gen}),
\begin{subequations}\label{eq:N-fc-ih}\begin{align}
  \label{eq:n-X-ih}
    X_{,r}\ &=\ -\bar{\pi}  \\
  \label{eq:n-H-ih}
    H_{,r}\ &=\ \sgr{\ell}\\
    \label{eq:n-mZ-ih}
  Z_{A,r}\ &=\ \pi m_A + \bar{\pi}\bar{m}_A\\
  \label{eq:n-mA-ih}
     m_{A,r}\ &=\ \bar{\lambda}\bar{m}_A + \mu m_A \\
\end{align}\end{subequations} 
The second order of the frame expansion, following directly from 
(\ref{eq:n-epsilon}, \ref{eq:n-mu}, \ref{eq:n-lambda}, \ref{eq:n-pi}) is
\begin{subequations}\begin{align}
	X_{,rr}\ &=\ -\bar{\Psi}_3 - \Phi_{12},\\
  H_{,rr}\ &=\ \bar{\Psi}_2 + \Psi_2 + 2\Phi_{11}
    - \fracs{1}{12}\Ricn{4}\\
	Z_A{}_{,rr}\ &=\ (\Psi_3+\Phi_{21})m_A
    + (\bar{\Psi}_3+\Phi_{12})\bar{m}_A \label{ddZ-gen}\\
  m_A{}_{,rr}\ &=\ - \Phi_{22}m_A - \bar{\Psi}_4\bar{m}_A \ . \label{ddm-gen}
\end{align}\end{subequations}
Note that the derivatives $H_{,rr}, X_{,rr}, Z_{A,rr}$ on $\ih$ are
determined directly by $(q_{ab},D_a)$ and the Ricci tensor 
[see (\ref{Psi2horizon},\ref{Psi3horizon},\ref{eq:DPsi4})]. The last derivative, 
${m}_{A,rr}$, involves a solution $\Psi_4$ to the equation \eqref{eq:DPsi4} uniquely 
determined by the initial value of $\Psi_4$ on chosen section and the horizon geometry.

To summarize, by direct inspection of the system of equations used here we see, 
that the data which is not determined, thus must be specified, consists of the 
following components: 
\begin{enumerate}[ (i)]
  \litem $\Phi_{21}, \Phi_{22}, \Ricn{4}, \Phi_{20,r}$ given on the
    entire $\ih$, and
  \litem $\Psi_4$ given on an initial slice $\slc$.
\end{enumerate}

\begin{rem}\label{locality} 
	As it is pointed out at the end of the previous Section \ref{sec:frame-ext}, 
	the elements frame $e_1$  and $e_2$ of the frame are defined locally, on the 
	sets $\ppi^{-1}(\hat{U}_I)$, $I=1,...,K$ covering the horizon $\ih$. On each 
	intersection between two sets, say $\ppi^{-1}{\cal U}_I$ and $\ppi^{-1}{\cal U}_J$, 
	there is an obvious transformation law, 
	\begin{equation} 
		e_1{}^{(I)}\ =\ u(x^A)^{(IJ)}e_1{}^{(J)} 
	\end{equation}
	where $u(x^A)^{(IJ)}\in {\rm U}(1)$. On the other hand, $e_3$ and $e_4$ are 
	defined globally, at every point of $\ih$. Now, the functions $\Phi_{21}, \Phi_{22}$ 
	and $\Phi_{20}$, as components of a tensor, are also defined locally, on each 
	set $\ppi^{-1}(\hat{U}_I)$, and satisfy the corresponding transformation laws 
	on the intersections. On the other hand, the Weyl tensor component  $\Psi_4$ 
	is not sensitive to the frame transformations preserving $e_3$ and $e_4$, and 
	$\Ricn{4}$ is just a scalar. The derivative $\partial_r=e_4^\mu\partial_\mu$, hence
	it is defined globally and commutes with the transformations.
\end{rem}

Finally, we address the question, what data is required to determine the general
order derivatives $\partial^n_r(e^{\mu})_{\nu}$. It turns out, that the general 
case is described by the following:
\begin{cor}\label{cor:data-gen}
  Given a NEH $\ih$  in a $4$-dimensional spacetime satisfying the Einstein field 
  equations with a general kind of matter, the Bondi-like coordinates $(x^A,v,r)$ 
  defined in Section~\ref{sec:neigh-invars} and a null frame $(e_1,e_2,e_3,e_3)$ 
  defined in Section~\ref{sec:frame-ih} and~\ref{sec:frame-ext},
  the following data
  \begin{enumerate}[ (i)]
    \litem the value of the constant $\sgr{\ell}$, 
    \litem on the initial slice $\slc$: ${m}^\alpha$ (which is tangent to the 
			slice by the construction), ${\pi}$, $\mu$, $\lambda$ and     
			$\nr^{k}\Psi_4$ $\forall k\in\{0,\cdots,n-2\}$
    \litem on $\ih$: $\nr^{k}\Phi_{11}$, $\nr^{k}\Phi_{21}$,
			$\nr^{k}\Phi_{22}$, $\nr^{k}\Ricn{4}$, $\nr^{k+1}\Phi_{20}$
			$\forall k\in\{0,\cdots,n-2\}$  
  \end{enumerate}
  determines uniquely all the radial derivatives 
  $\partial_r^ke_1^\mu, ..., \partial_r^ke_4^\mu$ (at $\ih$) of the frame components  
  up to the order $k=n$.
  The data is free, that is it is not subject to any extra constraints, modulo
  Remark \ref{locality}. Also, in the current work we make the additional  assumption \eqref{eq:lie_T-2} which in terms of the Newman-Penrose coefficients reads: $\partial_v\Phi_{02}=\partial_v\Phi_{20}=\partial_v(\Phi_{11}+3\Ricn{4})=0$.
\end{cor}
For the detailed proof of the above corollary the reader is referred
to Appendix~\ref{app:4d-neigh}. At this point one has to remember though, that 
Corollary~\ref{cor:data-gen} is not an existence or a uniqueness statement. 
For that, the data on $\ih$  has to be completed by suitable data defined on 
another null surface. Also the Einstein equations on $g_{\mu\nu}$ have to be 
completed by equations satisfied by the matter which contributes to the energy-momentum 
tensor (see Section~\ref{sec:cauchy-geom}).

The above expansion has been discussed in \cite{p-phd} and
subsequently presented up to a $2$nd order (also specifically in Einstein-Maxwell 
case) in \cite{k-spc}.

\subsection{4-dimensional electrovac NEH}
  \label{sec:evac-spc}

Let us now restrict  our studies to the case, when $\M'$ admits electromagnetic 
field as a sole matter content. The geometry of a non-expanding horizon in that case 
was analyzed already in \cite{abl-g}. Here we extend these studies by analysis of the
properties of an electrovac NEH's spacetime neighbourhood. First in
Section~\ref{sec:evac-exp} we introduce the necessary geometric
objects used for the description of the Maxwell fields, discuss their
properties and describe how the Maxwell evolution equations influence
the set of data necessary to determine the metric expansion at the
horizon. Next in Section~\ref{sec:ivp} we discuss the characteristic
initial value problem for the system under consideration in context of Bondi-like
coordinate system introduced in Section~\ref{sec:Bondi}. 

Given a NEH $\ih$ of a geometry $(q_{ab},D_a)$ we use throughout
this subsection the following objects: 
\begin{itemize}
	\item the Bondi-like coordinates $(x^A,v,r)$ adapted to $\ih$ and such that
		$\ell^a=(\partial_v)^a$ at $\ih$ is a null vector of a constant surface gravity 
		$\sgr{\ell}$,
	\item the null tangent frame $(e_1,...e_4)=(m^\mu,\bar{m}^\mu,n^\mu,\ell^\mu)$ 
		of the form \eqref{eq:frame-def} and the dual coframe $(e^1,...,e^4)$. 
\end{itemize}

\subsubsection{Constraints and the metric expansion}
  \label{sec:evac-exp}


Given a null frame specified above the electromagnetic field can be represented 
by the field coefficients defined in the following (equivalent to \eqref{eq:NP-F}) 
way\footnote{The decomposition is valid for a general Newman-Penrose null frame.}
\begin{equation}\label{eq:dummy}
  F\ :=\ \frac{1}{2}F_{\mu\nu}e^\mu\wedge e^\nu\
  =\ \ -\Phi_0e^4\wedge e^1 + \Phi_1(e^4\wedge e^3 + e^2\wedge e^1) 
  - \Phi_2 e^3\wedge e^2 + c.c. \ .
\end{equation}
The components of the energy-momentum tensor corresponding to the
field are just products of the respective field coefficients
\eqref{eq:NP-F-T} via \eqref{eq:NP-WR}. Their structure implies immediately 
that $T_{\mu\nu}\ell^{\mu}\ell^{\nu}\geq 0$. Whence from the Raychaudhuri equation 
it follows that $\Phi_{00}$ vanishes on $\ih$ , so does
\begin{equation}\label{eq:phi0-van}
\Phi_{0}|_{\ih}\ =\ 0 
\end{equation}
and so do all the components of $T_{\mu\nu}$ containing $\Phi_0$
\begin{equation}
  \Phi_{01}|_{\ih} = \Phi_{02}|_{\ih} = \Phi_{10}|_{\ih}= \Phi_{20}|_{\ih} = 0 \ .
\end{equation}
In consequence the Stronger Energy Condition \ref{c:energy} holds
for this kind of matter and (\ref{eq:lie_T-1}) is satisfied at $\ih$ automatically.

The component $\Phi_1$ is encoded into the pullback onto $\ih$ of 
$F_{\alpha\beta}-i*_\M F_{\alpha\beta}$
\begin{equation}
F_{ab}-i*_\M F_{ab}\ =\ \Phi_1\left(e^2\wedge e^1 \right)_{ab} 
\end{equation}

The electromagnetic field $F_{\mu\nu}$ is subject to the Maxwell equations
which in the null frame can be written in the form
\eqref{eq:Maxwell-NP}. On $\ih$ these equations reduce to 
\begin{subequations}\begin{align}
  \label{Phi01}
    \Phi_0|_{\ih}\ &=\ 0 \ , \quad D\Phi_1|_{\ih}\ =\ 0 \ ,  \\
  \label{DPhi2-nei}
    D\Phi_2|_{\ih}\ &=\ -\sgr{\bsl}\Phi_2 + (\mb+2\pi)\Phi_1 \ .
\end{align}\end{subequations}
The values of $\Phi_1, \Phi_2$ given on the chosen initial slice $\slc$
are then sufficient to determine the field $F_{\alpha\beta}$ at $\ih$
(provided that all the necessary frame and connection components are
given). Also the pull-back $F_{ab}$ of $F$ to $\ih$ is determined
just by $\Phi_1$ which furthermore can be represented as a pull-back
$\Phi_1 = \ppi^\star\hat{\Phi}_1$ of scalar $\hat{\Phi}_1$ defined on
$\bas$. 

The contribution of the Maxwell field to the frame expansion derived
in section \ref{sec:frame-exp} can be summarized as follows: since the 
Ricci tensor components (so the Maxwell field tensor) do not
contribute to the $0$th and $1$st order of expansion, the set of data
required to determine the expansions will be modified only for $n\geq
2$. In such case the modification can be summarized as:
\begin{cor}\label{cor:data-evac}
  Suppose $\ih$ is a non-expanding horizon embedded in $4$-dimensional
  electrovac spacetime. Let $(x^A,v,r)$ be the Bondi-like coordinates defined in 
  Section~\ref{sec:neigh-invars}, $(e_1,e_2,e_3,e_3)$ be a null frame defined in 
  ~\ref{sec:frame-ih} and~\ref{sec:frame-ext},
  and $\Phi_I$, $I=0,1,2$ be the electromagnetic field coefficients defined above.
  Then the value of the constant $\sgr{\ell}$ and the following data defined on the 
  initial slice $\slc$ 
	\begin{itemize}
		\item horizon geometry: ${m}^\alpha, {\pi},   \mu, \lambda$, 
		\item electromagnetic field: ${\Phi}_1, \ \ {\rm and}\ \ \nr^{k}\Phi_2, 
      \ \ \forall k\in\{0,\cdots,n-2\}$,  
		\item the Weyl tensor component: $\nr^{k}\Psi_4\ \ \forall k\in\{0,\cdots,n-2\}$
	\end{itemize}
	determines on $\ih$ uniquely all the radial derivatives 
	$\partial_r^ke_1^\mu, ..., \partial_r^ke_4^\mu$  of the frame components up to 
	the order $k=n$.
	This data is free, it is not subject to any constraints, modulo
  Remark~\ref{locality}. 
\end{cor}
\noindent The proof of this corollary, being a modification to the proof of
Corollary~\ref{cor:data-gen}, is presented in Appendix~\ref{app:evac-spc}. 
As in the case of Corollary~\ref{cor:data-gen}, the existence or uniqueness
is not guaranteed (see the discussion below Corollary~\ref{cor:data-gen}). 

\medskip

Corollary~\ref{cor:data-evac} and the form which the Maxwell-Eistein equations 
take on a NEH imply quite interesting property of the spacetime metric at the horizon. 
There is a well defined limit in which a given spacetime metric tensor $g_{\mu\nu}$ 
defines perturbatively --in terms of the expansion at a NEH $\ih$-- a new, ``would be'' 
(that is provided that it exists) stationary solution to the Einstein questions in 
\emph{all} the orders in the transversal variable $r$.  
Indeed,  we can easily determine the dependence on $v$ of all the data listed
in Corollary \ref{cor:data-evac}. 
In particular, the frame and rotation components are 
(by definition) $v$ independent, whereas $\mu,\lambda$ are (due to the reduction to 
the horizon of (\ref{eq:D-mu}, \ref{eq:D-lambda})) exponential (when $\sgr{\bsl}\neq 0$) 
or linear (otherwise) in $v$ respectively. 
Acting with $\partial^n_r$ on the transversal evolution
equations (\ref{eq:n-kappa}-\ref{eq:n-alpha}) (expressed in more
convenient form as (\ref{eq:ODE1}-\ref{eq:ODE3}) one can 
show that the $n$th transversal derivative of metric (represented by
the respective derivative of the frame components) behaves like
\begin{equation}
	\partial^n_rg_{\alpha\beta}|_{\ih}\ \sim\ g_{\alpha\beta}^{(n)} e^{-n\sgr{\bsl}v}+\ ...\ +g_{\alpha\beta}^{(0)}
\end{equation}
\cite{l-sem} if $\sgr{\bsl}\neq 0$.
Since the Bondi-like variable always can be  chosen in such a way that $\sgr{\ell}\not=0$,
this result can be interpreted in the way, that the horizon neighborhood geometry 
settles down to the geometry representing a Killing horizon with $\partial_v$ as 
a Killing vector. Note however, that the result does not mean that the horizon 
neighbourhood approaches the symmetric spacetime as $(i)$ we do not know whether or not 
there is a metric tensor, solution to the Einstein equations, which satisfies 
the limit expansion and $(ii)$ solutions to the characteristic initial Cauchy problem
do not need to be analytic.

\subsubsection{Characteristic Cauchy problem}
  \label{sec:ivp}

Now, we can complete the data of Corollary \ref{cor:data-evac} to characteristic
Cauchy data. This will be the null frame version of the Cauchy data introduced in 
a geometric way in Section \ref{sec:cauchy-geom}. Here, we provide a formulation 
in terms of the null tangent frame $(m^\mu,\bar{m}^\mu,\bn^\mu,\ell^\mu)$, 
the corresponding Newman-Penrose coefficients of the connection, curvature, 
and the electromagnetic field. We apply the results of sections~\ref{sec:neigh-invars} 
and~\ref{sec:evac-exp} to specify the class of the reduced Friedrich data
\cite{Friedrich-evac1,*Friedrich-evac2,Rendall}
corresponding to the case at hand and the NEHs in question.

As in Section \ref{sec:cauchy-geom}, in addition to a given NEH $\ih$, we use 
another null surface $\N_{v_0}$ orthogonal to a slice $\slc$ of $\ih$ such that  
$ v|_{\slc}=v_0.$  
Now, in the Bondi-like coordinates $(x^A,v,r)$ the null surfaces $\ih$ and 
$\N_{v_0}$ satisfy:
\begin{equation}
  r|_{\ih}\ =\ 0,\ \ \ \ v|_{\N_{v_0}}\ =\ v_0 \ . 
\end{equation}
The NEH horizon geometry and the component $\Phi_1$ of the electromagnetic field 
defined on $\ih$, coupled to  the component $\Psi_4$ of the Weyl tensor and $\Phi_2$ 
set freely on the entire $\N_{v_0}$ provide at the slice $\slc$ the data of 
Corollary~\ref{cor:data-evac}. Furthermore, they determine uniquely all the (spacetime) 
frame, connection, Maxwell field and Riemann tensor components at $\N_{v_0}$. 
The key idea of the proof is the observation that the Einstein-Maxwell equations 
and Bianchi identities form on $\N_{v_0}$ a hierarchy of the ordinary differential 
equations. For the readers convenience the proof of this fact is presented 
in Appendix~\ref{app:ivp}. The consequence of the above observations is
\begin{cor}\label{cor:Cauchy-data}
	Given a non-expanding horizon $\ih$ and the transversal null
	surface $\N_{v_0}$ the following   data  is the Friedrich reduced data
	\cite{Friedrich1,*Friedrich2}:
	\begin{enumerate}[ (i)]
		\item the surface gravity $\sgr{\bsl}\in\mathbb{R}$ ($=0$ or $=1$);
		\item on $\slc_{v_0}=\ih\cap\N_{v_0}$: $m^\alpha$ (by construction, 
			tangent to $\slc$), $\pi$, 	$\mu$, $\lambda$, ${\Phi}_1$; \label{it:Fr-data}
		\item on $\N_{v_0}$:  $\Phi_2$, $\Psi_4$.
	\end{enumerate}
\end{cor}
The data is freely defined modulo Remark \ref{locality}. Given this data,
in the domain of dependence of $\ih\cup \N_{v_0}$, the NP equations 
(\ref{eq:D-rho}-\ref{eq:D-mu},\ref{eq:n-epsilon}-\ref{eq:n-alpha}) coupled with the 
Einstein-Maxwell field equations (\ref{eq:NP-WR}, \ref{eq:NP-F}), Maxwell evolution 
equations \eqref{eq:Maxwell-NP}, Bianchi identities (\ref{eq:BI-DPsi1}-\ref{eq:nPsi3}),
frame components evolution equations \eqref{eq:N-fc-gen}
and with the gauge choice equations (\ref{eq:frame-e3def}, \ref{eq:frame-lv}, %
\ref{2gamma}, \ref{eq:frame-api-evo}, \ref{eq:frame-muevo}, \ref{eq:frame-def}, %
\ref{eq:fc-H}, \ref{eq:NP-zeros}, \ref{eq:NP-adef}) define   
a unique electrovac spacetime ($e_1,...,e_4$, $\Phi_0$-$\Phi_2$, $\Psi_0$-$\Psi_4$, %
$X$, $H$, $Z_A$, $m_A$). 
In the spacetime defined in $\M'^\pm$ (the future/past to $\ih$ part of the domain of dependence) by the resulting solution,
$\ih$ is a non-expanding horizon, $(x^A,v,r)$ is an adapted Bondi-like coordinate
system and $(m^\mu,\bar{m}^\mu,\bn^\mu,\ell^\mu)$ is a null frame of the properties
of the frame (\ref{eq:frame-e3def}-\ref{eq:DPsi4}). In particular, the vector field 
\begin{equation} 
	\bn^\mu\ := -(\partial_r)^\mu 
\end{equation}
is null, satisfies 
\begin{equation}
	\nabla_nn=0 
\end{equation}
and is orthogonal to the slices 
\begin{equation} 
	v=\const 
\end{equation}
of the horizon. Also, the vector field 
\begin{equation}
	\xi^\mu\ :=\ (\partial_v)^\mu
\end{equation}
satisfies 
\begin{equation}
	\xi^\mu|_{\ih}\ =\ \ell^\mu, \qquad \lie_n\xi\ =\ 0.
\end{equation} 
Therefore, the current vector fields $\bn^\mu$ and $\xi^\mu$ coincide with the vector 
fields $n^\mu$ and $\xi^\mu$ induced in a neighborhoud of $\ih$ 
in Section~\ref{sec:neigh-invars}.  

This structure will be applied in the next section to identify and characterize 
the necessary and sufficient consitions for the existence of a timelike Killing 
field on the domain $\M'^\pm$.

\section{Electrovac Killing horizon}
  \label{sec:evac-kvf}

\subsection{The induced structures, the Bondi-like coordinates and the adapted null frame} 

We now restrict our interest to the situation when an electrovac spacetime $\M$ admits 
a Killing vector field $K^\mu$ tangent to and null at a horizon $\ih$. 
We also assume that $K^\mu$ is an infinitesimal symmetry of the electromagnetic field,
that is 
\begin{equation} \lie_KF_{\mu\nu}\ =\ 0.\end{equation}
On the horizon
\begin{equation}
	(K|_{\ih})^a\ =\ \ell^a
\end{equation}
is an infinitesimal symmetry. We also assume, that  the surface gravity 
of $\ell^a$ (necessarily constant) is not zero
\begin{equation}\label{eq:sgr-evac-kill}
	\sgr{\ell}\not=0.
\end{equation}
We will apply the general results of Sections~\ref{sec:sym-known}, \ref{sec:neigh-kill} 
and~\ref{sec:neigh-killnonin} as well as we will employ the adapted null frames introduced 
in Section~\ref{sec:4d-neigh}.   
 
If the NEH $\ih$ geometry is invariant-generic (see Section~\ref{sec:neigh-kill}) 
and $\xi^\mu$ is the $\ih$  neighbourhood invariant vector field, then due to 
Theorem~\ref{cor:kill-form} the Killing vector necessarily coincides with $\xi^\mu$ 
modulo a rescaling by a constant factor
\begin{equation}
	K^\mu\ =\ \xi^\mu.
\end{equation}
Otherwise,  the results of Section~\ref{sec:neigh-killnonin} apply. 
In either case, there are on $\ih$ coordinates $(x^A,v)$ such that 
\begin{equation}
	\ell^a\ =\ (\partial_v)^a.
\end{equation} 
We are assuming they are given and use the corresponding  Bondi-like extension
and the related Bondi-like coordinates $(x^A,v,r)$. Then, owing to 
Theorem~\ref{thm:Kill-hel}, the Killing vector $K$ in the neighbourhood $\M'$ 
necessarily is
\begin{equation}
	K^\mu\ =\ (\partial_v)^\mu.
\end{equation} 

It turns out that the null frame $(e_1,...,e_4)$ \eqref{eq:frame-ext} adapted to 
the structures introduced on $\ih$, and the adapted Newman-Penrose framework defined 
in Section~\ref{sec:4d-neigh} are surprisingly compatible with the Killing vector fields:    
\begin{lem}\label{KNP=0} 
	Suppose $K^\mu$ is a Killing vector field tangent to a NEH $(\ih,q_{ab},D_c)$. 
	The components of the null frame $(e_1,e_2,e_3,e_4)$ and the dual coframe 
	$(e^1,...,e^4)$ introduced in Section~\ref{sec:null-frame}
	are Lie dragged by $K^\mu$, that is
	\begin{equation}
		\lie_Ke^1=...=\lie_Ke^4\ =\ 0\ =\ \lie_Ke_1=...=\lie_Ke_4,
	\end{equation}
	provided 
	\begin{equation}
		\lie_K e^1|_{\ih} = ... = \lie_Ke_4|_{\ih}=0. 
	\end{equation}
\end{lem}
Indeed, it follows from the following calculation true in the neighborhood 
$\M'$ of $\ih$ for every value of $\mu=1,...,4$ (no abstract index)
\begin{equation}\label{calc}
	0\ =\ \lie_K(\nabla_{\bn} e_\mu)\ =\ \nabla_{\bn} \lie_Ke_\mu
\end{equation}
where the first equality follows from $\nabla_{\bn} e_\mu=0$, whereas the second 
follows from 
\begin{equation}
	[\bn,K]\ =\ 0
\end{equation}
which for a Killing vector $K^\mu$ implies that the parallel transport along the 
integral lines of $\bn^\mu$ commutes with the flow of $K^\mu$. The second equation 
in \eqref{calc} combined with the initial condition 
\begin{equation}
	\lie_K e_\mu|_{\ih}\ =\ 0
\end{equation}  
completes the proof. 

\begin{cor}\label{cor:kvf-form}
	Consider a NEH $\ih$ such that its neighbourhood admits a Killing vector field
	tangent to $\ih$ and null thereon. If $\ih$ is invariant-generic introduce
	on $\ih$ coordinates $(x^A,v)$ such that $\partial_v$ is the invariant vector on $\ih$. 
	Otherwise, assume that $K^\mu$ is not zero restricted to any null generator of $\ih$ 
	and introduce on $\ih$ coordinates $(x^A,v)$ such that 
	\begin{equation}
		K^\mu|_{\ih} = (\partial_v)^{\mu} .
	\end{equation}
	Extend $(x^A,v)$ to the Bondi-like coordinates in the neighbourhood $\M'$.
	Introduce the null frame of Section \ref{sec:frame-ih} and Section \ref{sec:frame-ext}. 
	
	Then, in all the $\M'$ the Killing vector $K^\mu$ is of the form
	\begin{equation}
		K^{\mu}=(\partial_v)^{\mu}.
	\end{equation}
	Furthermore, as a consequence of Lemma~\ref{KNP=0}, all the frame coefficients 
	$e_\mu^A, e_\mu^v,e_\mu^r$
	as well as all the Newman-Penrose coefficients of the Levi-Civita connection, the Maxwell 
	field and the Weyl tensor are constant along the orbits of $K^\mu$, that is they are 
	independent of the variable $v$.
\end{cor}

\subsection{Necessary conditions: data at $\ih$, and the metric expansion}
  \label{sec:evac-kvf-exp}
  
The expansion in the radial coordinate\footnote{The affine parameter of transversal 
	null geodesics.
} 
$r$ at the given NEH $\ih$ of the coefficients of the  frame $e_1,...e_4$ in the 
general case assuming the vacuum Einstein-Maxwell equations was developed 
in Section~\ref{sec:evac-exp}. The presence of the Killing field imposes new 
constraints on the data considered on the horizon $\ih$ following from the substitution 
of the symmetry condition 
\begin{equation}
	\ell^\alpha\partial_\alpha(n^\beta\partial_\beta)^k f\ =\ 0, 
\end{equation}
where $f$ is any component of the frame $e^\alpha_1,...,e^\alpha_4$ in the Bondi-like 
coordinates $(x^A,v,r)$, and any component (in that frame) of the Levi-Civita connection,
the Weyl tensor, and the Maxwell field, and $k\in\mathbb{N}$. 
Indeed, the equation~\eqref{eq:constr-evac-4D} with $0$ substituted for $\partial_v \mu$ 
and $\partial_v\lambda$ determines the expansion and shear  $(\mu,\lambda)$ of $n$ as  
functionals of the remaining elements of the horizon geometry ---the complex vector $m^a$, 
the surface gravity $\sgr{\ell}$ and the component $\pi$ of the rotation 1-form potential--- 
and the component $\Phi_1$ of the electromagnetic field,
\begin{subequations}\label{eq:mul-kill}\begin{align}
	\mu\ &=\ \frac{1}{\sgr{\ell}} \left[ m^a\partial_a + |\pi|^2 -2\pi\bar{a} 
		+ \Psi_2 \right] \ , \\
	\lambda\ &=\ \frac{1}{\sgr{\ell}} \left[ \bar{m}^a\partial_a + |\pi|^2 
		+ 2a\pi \right] \ .
\end{align}\end{subequations}
Also the Maxwell field equation \eqref{DPhi2-nei} upon the assumption
$\partial_v \Phi_2=0$ determines the value of $\Phi_2$ 
at $\ih$ for known $(\Phi_0,\Phi_1)$ and connection coefficients,
\begin{equation}
	\Phi_2\ =\ \frac{1}{\sgr{\ell}} \left[ \bar{m}^a\partial_a + 2\pi\Phi_1 \right] \ .
\end{equation}
Hence the value of a whole Energy-Momentum tensor at $\ih$ is known. Therefore given  
$(m,\sgr{\ell},\pi,\Phi_1)$ one can calculate all the connection coefficients as 
well as  the Weyl tensor components $\Psi_0,\Psi_1,\Psi_2,\Psi_3$ (see the analysis 
in Section~\ref{sec:evac-exp}). 
In the analogous way the last component $\Psi_4$ is  determined as a functional of 
$(m,\sgr{\ell},\pi,\Phi_1)$ via the substitution of $0$ for $\partial_v\Psi_4$
in \eqref{eq:DPsi4} and expressing $\Phi_{20,r}$ therein by a suitable functional
of $(m,\sgr{\ell},\pi,\Phi_1)$ following from $\Phi_0=0$ and the Maxwell equation 
\eqref{eq:n-Phi0},
\begin{equation}
	\Psi_4\ =\ \frac{1}{2\sgr{\ell}} \left[ 
		\bar{m}^a\partial_a\Psi_3 - 3\lambda\Psi_3 + (5\pi+2a)\Psi_3 -\Phi_{20,r}
		+ \bar{m}^a\partial_a\Phi_{21} + (\pi+2a)\Phi_{21} - 2\lambda\Phi_{11} 
	\right]
\end{equation}
In this way the free degrees of freedom (represented by the triple $(m^a,\pi,\Phi_1)$ 
defined on the horizon $\ih$) determine then the frame expansion up to $2$nd order. 
Furthermore, given the frame $e_1^\mu,...,e_4^\mu$ at $\ih$ and the derivatives $\partial_r^ke_1^\mu,...,\partial_r^ke_4$ for $k=1,...,n$, the
values of $\nr^{n-1}\Phi_2$ and $\nr^{n-1}\Psi_4$ necessary for the  $n+1$th
order of expansion can be derived by differentiating the equations \eqref{eq:D-Phi2} 
and \eqref{eq:B-DPsi4} in the radial direction respectively (see Appendices 
\ref{app:frame-exp} and \ref{app:evac-exp} for the details). Finally the following 
is true:
\begin{cor}
  Suppose $(\ih,\ell)$ is a  Killing horizon in a $4$-dimensional electrovac spacetime 
  and the surface gravity $\sgr{\ell}\not=0$ \eqref{eq:sgr-evac-kill}. Then, the 
  following data defined on the horizon: the complex vector  $m^a$ \eqref{eq:m-def}, 
  the rotation $1$-form potential  $\omega_a$ \eqref{eq:omega_def}, and the pull-back 
  onto $\ih$ of $(F-i*_\M F)_{ab}$, 
  uniquely determine all the derivatives $\nr^n e^{\mu}{}_{\nu}$ of the frame 
  \eqref{eq:frame-def} coefficients as well as all the derivatives $\partial_r^n\Phi_I$, 
  $I=0,1,2$ of the components of the electromagnetic field \eqref{eq:NP-F} 
  at the horizon $\ih$, for all $n\in\mathbb{N}$.
  \label{cor:evac-kill}
\end{cor}
In particular, in the Einstein vacuum case, that is in the absence of the Maxwell 
field, is a NEH $\ih$ is a Killing horizon, and the Killing vector field is given 
at the horizon as $\ell$, then all the transversal derivatives $(\lie_n)^k g_{\mu\nu}$, 
$k\in \mathbb{N}$, where $n$ is a transversal to $\ih$ vector field, are determined 
by the  degenerate metric tensor $q_{ab}$ induced in $\ih$ and the rotation $1$-form 
potential $\omega_a$ \eqref{eq:omega_def} provided $\ell^a\omega_a\neq0$. 

From Corollary \ref{cor:evac-kill} it follows immediately that, provided the spacetime
metric $g_{\mu\nu}$ and the electromagnetic field tensor $F_{\mu\nu}$
are analytic, they are uniquely determined in the space-time 
by the data $(m,\sgr{\ell},\pi,\Phi_1)$.
In a vacuum case  the analyticity is ensured in the region where KVF is timelike 
\cite{mzh-static,mzh-stationary} and \eqref{eq:zeta-norm}.  This is the case   
outside (in direction against $\bn^\mu$) (when $\sgr{k}>0$) or inside (for $\sgr{k}<0$) 
the horizon in the connected region\footnote{We assume here that the edge of this 
	region has non-empty intersection with the horizon
}. However, we still do not know whether the metric is analytic up to the horizon $\ih$.

This requirement is satisfied in particular by non-degenerate Killing horizons in 
the static vacuum spacetime \cite{ch-static}. The listed data representing a non-rotating 
horizon\footnote{The rotation of a static Killing horizon necessarily vanishes.}
uniquely determines then the metric and Maxwell field of a static spacetime in the 
connected region in which the Killing field is timelike.

\subsection{Necessary conditions: data on the transversal surface $\N_{v_0}$}
  \label{sec:evac-kvf-exist}

In subsection \ref{sec:evac-kvf-exp} we characterized  the $\ih$ part of the 
characteristic Cauchy data of Section \ref{sec:ivp}. Now, we turn to the data 
defined on the null surface $\N_{v_0}$ transversal to the horizon $\ih$ which data 
consist of the functions $(\Phi_2,\Psi_4)$. The presence of the Killing field 
inducing the null symmetry at the horizon imposes some constraints on these (otherwise 
free) data. The specifics of the construction of the Bondi-like coordinates imply that 
(see Sec.~\ref{sec:neigh-kill}), given coordinates $(x^A,v,r)$ such that the Killing 
field (if present) null at the horizon $\ih$  has the form $K|_{\ih}=\partial_v$, it takes 
this form (i.e. $K^{\mu}=(\partial_v)^{\mu}$) in the neighborhood $\M'$ covered  by 
the coordinates (see Theorem~\ref{cor:kill-form}~\eqref{it:cor-kill-form-null}). If 
the horizon is invariant-generic, then the coordinate $v$ is  a priori given as the 
invariant one and $(\partial_v)^\mu=\zeta^\mu$, the invariant vector field of the 
neighbourhood of $\ih$. Otherwise, the problem reduces to finding a suitable $v$ on $\ih$. 
Therefore, probing for a Killing vector field reduces to the question whether 
the field $\partial_v$ is a symmetry of $g_{\mu\nu}$ and $F_{\mu\nu}$ at the horizon 
neighbourhood.
We also remember from Corollary \ref{cor:kvf-form}, that if $\partial_v$ is a KVF 
then it preserves all the frame coefficients. The converse statement is straightforward: 
the independence of $v$ of all the frame, Maxwell field, connection and Weyl tensor 
coefficients is equivalent to the fact that $\partial_v$ is a KVF. The condition
\begin{subequations}\label{eq:pre-KVF-constr}\begin{align}
  \partial_v\Psi_4|_{\N_{v_0}}\ &=\ 0 \ ,  &
  \partial_v\Phi_2|_{\N_{v_0}}\ &=\ 0 \tag{\ref{eq:pre-KVF-constr}}
\end{align}\end{subequations}
is then a necessary condition for $\partial_v$ to be the KVF.

Using the Bianchi identity \eqref{eq:B-DPsi4} and the Maxwell field equation 
\eqref{eq:D-Phi2} [and expressing the operator $D:=\ell^\mu\partial_{\mu}$ 
in the Bondi-like coordinate system via \eqref{eq:frame-def}] one can rewrite these 
conditions as the differential constraints involving the derivatives in the directions 
tangent to $\N_{v_0}$ only
\begin{subequations}\label{eq:KVF-constr}\begin{align}
  \label{eq:KVF-Psi4}
  \begin{split}
    (H\partial_r- \bar{X}\m-X\mb) \Psi_4\ &=\ - (4\epsilon -\rho)\Psi_4 + \mb\Psi_3
      + (5\pi + 2a)\Psi_3 - 3\lambda\Psi_2 -\kappa_0\bar{\mu}\Phi_2\bar{\Phi}_0 \\
    &\hphantom{=}\ + \kappa_0((\pi+2a)\Phi_2\bar{\Phi}_1-2\lambda\Phi_1\bar{\Phi}_1
      +\partial_r\Phi_2\bar{\Phi}_0 + \bar{\sigma}\Psi_2\bar{\Phi}_2) \ ,
  \end{split} \\
  \label{eq:KVF-Phi2}
  (H\partial_r - \bar{X}\m-X\mb)\Phi_2\ &=\ \mb\Phi_1 - \lambda\Phi_0 + 2\pi\Phi_1 
    + (\rho-2\epsilon)\Phi_2 \ .
\end{align}\end{subequations}
The conditions \eqref{eq:pre-KVF-constr} became then the well defined on $\N_{v_0}$ 
constraint \eqref{eq:KVF-constr} for the geometry components. Note that this system
involves the transversal to $\ih$ derivatives $(\nr\Psi_4,\nr\Phi_2)$. It can be then 
treated as the completion of the 'evolution' equations (\ref{eq:ODE1}-\ref{eq:n-Psi0-N}). 
However, we do not know whether  the resulting system of equations has a well defined 
Cauchy problem  on $\ih$. The first difficulty is, that the $r$-dependent  
coefficient $H$ vanishes at $\ih$. 
Secondly, the completed system constitutes now a system of partial differential 
equations (PDE's) instead of the ordinary ones. The 
structure of this system is not manifest, however the action of the Killing flow 
allows to recast it into the system defined on the Cauchy surface where an equivalent 
system is elliptic in the region where the KVF is timelike \cite{mzh-stationary}. 
Unlike in the static case \cite{ch-static} the question about the ellipticity of 
the system \emph{at} the horizon remains open.

\subsection{The necessary and sufficient conditions}

In this subsection we  formulate the set of necessary and sufficient conditions
for the existence of a Killing vector tangent to and null at horizon in the 
$4$-dimensional, electrovacuum case. As above, we will use the transversal null
surface $\N_{v_0}$. 

\begin{theorem}\label{thm:KVF-suffic} 
  Suppose $(\ih,q_{ab},D_a)$ is an
  invariant-generic non-expanding horizon contained in 4-dimensional spacetime
  $(\M, g_{\alpha\beta})$ which satisfies the vacuum Einstein-Maxwell equations
  with an electromagnetic field $F_{\mu\nu}$. Let $\zeta^\mu$ be the invariant 
  vector field of the neighbourhood of $\ih$ and $(x^A,v,r)$ be the invariant 
  Bondi-like coordinate system (see Section~\ref{sec:neigh-invars}). Each of the 
  conditions $(i)$ and $(ii)$ below is equivalent to the  local existence (in the 
  domain of dependence of $\ih\cup\N_{v_0}$ in the point $(ii)$ below) of a Killing 
  vector $K^\mu$ tangent to and null at $\ih$, and such that $\lie_KT_{\mu\nu}=0$:
  \begin{enumerate}[(i)]
    \item $\hphantom{.}$\vspace{-0.60cm}
      \begin{equation} 
				\lie_{\zeta}g_{\mu\nu}\ =\ 0 \ , \qquad 
				\lie_{\zeta}{T_{\mu\nu}}\ =\ 0 \ , \label{liezeta} 
      \end{equation}
    \item
      The Cauchy data defined on $\ih$ and $\N_{v_0}$ satisfy:
      \begin{enumerate}[(a)]
				\item on $\ih$:  $\zeta^a$ is an infinitesimal symmetry of $(q_{ab},D_a)$ and 
					$\lie_{\zeta}F_{\mu\nu}\ =\ 0$     
				\item on $\N_{v_0}$: the conditions \eqref{eq:KVF-constr} are satisfied, 
					provided (for the necessary condition) the null frame \eqref{eq:frame-def} 
					is constructed such that 
					$\lie_\zeta e_1|_{\ih}\ =\ldots = \  \lie_\zeta e_4|_{\ih}\ =\ 0.$
      \end{enumerate}
  \end{enumerate}
\end{theorem}
\begin{rem}
  The condition $(ii)$(a) above is equivalent to (\ref{m}, \ref{G}, \ref{eq:frame-api-evo},
  \ref{eq:NP-matt-gauge}, \ref{Phi01}, \ref{eq:mul-kill}).
\end{rem}
The necessary conditions follow from the fact, that in the invariant Bondi-like 
coordinate system the invariant vector $\zeta$ of the neighborhood of $\ih$ has the 
form $\zeta^\mu=(\partial_v)^\mu$, from Theorem~\ref{cor:kill-form}, and from the 
previous section.    
To complete the proof of sufficiency, it is enough to show that at $\ih\cup \N_{v_0}$ 
the vector field $\zeta^\mu$ satisfies the Racz conditions \cite{Racz1,*Racz2}
\begin{subequations}\begin{align}
  \label{Racz1}\lie_{\zeta}g_{\mu\nu}\ =\ \nabla_\alpha \lie_{\zeta}g_{\mu\nu}\ 
  =\ \lie_{\zeta}T_{\mu\nu}\ &=\ 0 \ ,  \\
  \label{Racz2}\nabla^{\mu}\nabla_{\mu}\zeta_{\nu}
   + \Ricn{4}_{\nu}{}^{\mu}\zeta_{\mu} \ &=\ 0 \ .
\end{align}\end{subequations}
These conditions are ensured by the Lemma~\ref{lem:kvf-zeta} (cond.~\eqref{Racz1}) 
and Lemma~\ref{lem:k'k} (cond.~\eqref{Racz2}) below. Once they hold, $\zeta$ is 
necessarily a Killing field by Racz theorem \cite{Racz1,*Racz2} (which we quote 
in Appendix~\ref{app:Racz}).
\begin{lem}\label{lem:kvf-zeta}
	Suppose $(\ih, q_{ab},D_a)$ is a non-expanding horizon contained in 4-dimensional 
	spacetime $(\M, g_{\alpha\beta})$ which satisfies the vacuum Einstein-Maxwell 
	equations with an electromagnetic field $F_{\mu\nu}$. Let $(x^A,v,r)$ be the 
	Bondi-like coordinate system of Section~\ref{sec:neigh-invars}. Suppose the 
	conditions $(ii)$(a),(b) of Theorem~\ref{thm:KVF-suffic} are satisfied (however, 
	$\partial_v$ is not assumed here to be the invariant vector field). 
	Then, the vector field $\partial_v$ satisfies at $\ih \cup \N_{v_0}$ the following 
	condition for arbitrary $n\in\mathbb{N}$,
  \begin{equation}\label{eq:kvf-nabla}
    \nabla^{(n)}_{\alpha_1,\cdots,\alpha_n}\lie_{\partial_v}
    g_{\mu\nu}|_{\Sigma}\ =\ 0 \ .
  \end{equation}
\end{lem}

\begin{lem}\label{lem:k'k}
  Suppose a NEH $(\ih,q_{ab},D_c)$ and a vector field $\partial_v$ satisfy all the
  assumptions of Lemma \ref{lem:kvf-zeta}. Then the solution to the
  initial value problem 
  \begin{equation}\label{eq:KVF-cand-lem}
    \nabla^{\mu}\nabla_{\mu}K'_{\alpha}
      + \Ricn{4}_{\alpha}{}^{\mu}K'_{\mu}\ =\ 0 \ , \qquad  
    K'_{\mu}|_{\ih\cup\N}\ =\ \zeta_{\mu} \ 
  \end{equation}  
  agrees at $\Sigma=\ih\cup\N$ with the vector field $\partial_v$ 
  up to arbitrary order
  \begin{equation}\label{eq:k'k}
    \forall_{n\in\mathbb{N}}\quad
    \nabla^{(n)}_{\alpha_1,\cdots,\alpha_n}K'_{\mu}\ =\
    \nabla^{(n)}_{\alpha_1,\cdots,\alpha_n}\zeta_{\mu} \ .
  \end{equation}
\end{lem}
The proofs of the above Lemmas \ref{lem:kvf-zeta} and \ref{lem:k'k} is presented in the 
Appendices~\ref{app:kvf-zeta} and \ref{app:k'k} respectively.

In Theorem \ref{thm:KVF-suffic} we assumed that $\ih$ was invariant-generic. Owing 
to that assumption and to Theorem \ref{cor:kill-form} the only candidate for the 
Killing vector field was the invariant vector $\zeta^\mu$ which in the invariant 
Bondi-like coordinates equals $\partial_v$. On the other hand, if we relax the 
invariant-genericity assumption, we still have Theorem~\ref{sec:neigh-killnonin}. 
Combined with Lemma~\ref{lem:kvf-zeta} and Lemma~\ref{lem:k'k} it leads to the 
following non-invariant version of Theorem~\ref{thm:KVF-suffic}:   
\begin{theorem}\label{thm:KVF-suffic'} 
  Suppose $(\ih, q_{ab},D_a)$ is a NEH contained in 4-dimensional spacetime
  $(\M, g_{\alpha\beta})$ which satisfies the vacuum Einstein-Maxwell equations
  with an electromagnetic field $F_{\mu\nu}$. Let $(x^A,v,r)$ be the Bondi-like 
	coordinate system (see Section \ref{sec:neigh-invars}) such that $\sgr{\partial_v}\ %
	=\ \const\ \neq 0$.
  Each of the conditions $(i)$ and $(ii)$ below is equivalent to the  local
  existence (in the domain of dependence of $\ih\cup \N_{v_0}$ in the point (ii) 
  below) of a Killing vector $K^\mu$ tangent to $\ih$, such that
  $K|_{\ih}=\partial_v$ and such that $\lie_KT_{\mu\nu}=0$:
  \begin{enumerate}[(i)]
    \item $\hphantom{.}$\vspace{-0.60cm}
      \begin{equation}
	\lie_{\partial_v}g_{\mu\nu}\ =\ 0 \ , \qquad 
	\lie_{\partial_v}{T_{\mu\nu}}\ =\ 0 \ , \label{liepartial} 
      \end{equation}
    \item
      The Cauchy data defined on $\ih$ and $\N_{v_0}$ satisfy:
      \begin{enumerate}[(a)]
	\item on $\ih$:  $\partial_v$ is an infinitesimal symmetry of $(q_{ab},D_a)$ and 
	  $\lie_{\partial_v}F_{\mu\nu}|_\ih\ =\ 0$     
	\item on $\N_{v_0}$: the conditions (\ref{eq:KVF-constr}) are satisfied,
	  provided (for the necessary condition) the null frame \ref{eq:frame-def} is 
	  constructed such that 
	  $\lie_\zeta e_1|_{\ih}\ = \ldots = \  \lie_\zeta e_4|_{\ih}\ =\ 0.$
      \end{enumerate}
  \end{enumerate}
\end{theorem}

\section{Axial and helical Killing fields in $4$D electrovac spacetime}
  \label{sec:evac-kvf-axihel}

The Killing fields null at the horizon are not the only possible type of spacetime 
symmetries. There are two more classes possible: axial KVF  and helical KVFs. 
In this section we formulate
the set of necessary and sufficient conditions for their existence analogous to 
Theorem~\ref{thm:KVF-suffic}. We still assume that the studied horizons are NEHs 
embedded in a $4$-dimensional electrovac spacetime.

\subsection{Axial KVF}

If the spacetime neighbourhood of a NEH $(\ih,q_{ab},D_c)$ admits a rotational 
Killing field $\Phi^\mu$ tangent to $\ih$, then  
one can choose at $\ih$ a null vector field $\ell^a$ such that 
\begin{equation}
	[\Phi,\ell]\ =\ 0 \ , \qquad \sgr{\ell}=\const\neq 0
\end{equation}
and a foliation of $\ih$ by spacelike slices, each  preserved by the symmetry generated 
by $\Phi^a|_{\ih}$ \cite{lp-symm}. In the case of an invariant-generic NEH $\ih$, 
the invariant vector $\ell^a$ and the invariant foliation have this property. 
Otherwise, we will be assuming that $\ell^a$ and the foliation are given. 
In the corresponding Bondi-like coordinates $(x^A,v,r)=(\theta,\phi,v,r)$, such 
that $(\theta,\phi)$ are the spherical coordinates defined on the spheres $v=$const, 
$r=$const, and such that 
\begin{equation}
	\Phi^a|_{\ih}\ =\ (\partial_\phi)^a
\end{equation} 
owing to Theorem~\ref{cor:kill-form} in the invariant-generic case, and 
Theorem~\ref{thm:Kill-hel} otherwise, in all the domain of the Bondi-like extension, 
it is true that
\begin{equation} 
	\Phi^\mu=(\partial_\phi)^\mu.
\end{equation}         
The Bondi-like coordinates are determined by the coordinates $(\theta,\phi,v)$ 
defined on $\ih$ by using only $\Phi|_{\ih}$. We also assume, that the KVF $\Phi$ 
is a symmetry of the Maxwell field. Due to Lemma \ref{KNP=0}, in the adapted null 
frame \ref{eq:frame-def} such that 
\begin{equation} 
	\lie_\Phi e^1|_{\ih}\ =\ \ldots\ = \lie_\Phi e^4|_{\ih}\ 
	=\ \lie_\Phi e_1|_{\ih}\ =\ \ldots\ = \lie_\Phi e_4|_{\ih}\ =\ 0
\end{equation}         
the Cauchy data defined on the surfaces $\ih\cup \N_{v_0}$ ($\N_{v_0}$ such that 
$v=v_0$) as specified in Corollary~\ref{cor:Cauchy-data} is invariant with respect 
to $\Phi=\partial_\phi$.

The opposite statement is obviously true. Suppose the Cauchy data defined on 
$\ih\cup \N_{v_0}$ has an infinitesimal axial symmetry $\partial_\phi$. Then, 
the vector field $\Phi^\mu$ defined  in the corresponding Bondi-like coordinates as
\begin{equation}
	\Phi^\mu\ =\ (\partial_\phi)^\mu
\end{equation}
is a Killing vector field.

\subsection{Helical KVF}

If the spacetime neighbourhood of an IH admits a helical Killing vector field $X$ 
(see Theorem~\ref{thm:Kill-hel} for the definition) a result analogous to Theorem~%
\ref{thm:KVF-suffic} and Theorem~\ref{thm:KVF-suffic'} can be established. In the 
current case, however, the structure of a spacetime symmetries is much richer. 
If present, the KVF $X$ induces (see Theorem~\ref{thm:hel-decomp}) at $\ih$ both,
null and axial symmetry. Then, if it exists, by Theorem~\ref{thm:Kill-hel} we 
construct on $\M'$ the Bondi-like coordinate system, using as the boundary condition 
at $\ih$ the assumption, that restriction to $T(\ih)$ of $X$ is a linear combination 
of a null $\partial_v$ and axial $\partial_{\varphi}$ symmetry generators.      

Knowing the expected form of KVF one can repeat the steps performed in the
proof of Theorem \ref{thm:KVF-suffic}, just inserting as a candidate for KVF the 
field $X^{\mu}:=a\zeta^{\mu}+b(\partial_{\varphi})^{\mu}$ (where $a,b$ are
constants) instead of $\zeta^{\mu}$. The frame coefficients have then to satisfy
an additional condition, namely that at the horizon they are invariant with respect 
to axial symmetry induced at it
\begin{equation}\label{eq:hk-frame-inv}
  \partial_{\varphi} e^{\mu}|_{\ih}\ =\ 0 \ .
\end{equation}
Upon that assumption the generalization of Theorem~\ref{thm:KVF-suffic}
to the case of helical KVF is almost straightforward. The only step
which require certain attention if the proof of \eqref{eq:kvf-nabla}
(part of the proof of Lemma~\ref{lem:kvf-zeta}) at $\ih$ as the evolution 
of higher order derivatives of $g_{\mu\nu}$ along the orbits of $X$ is a priori 
not known.
We know however that the components of $\ih$ internal geometry (so all
the elements of the set of (external) geometry data $\bar{\chi}$
except $\Psi_4,\Phi_2$) are preserved by the flow of $\partial_v$ and
$\partial_{\varphi}$ as well as the flow of $X$. Whence, applying the
method used in proof of Lemma~\ref{lem:kvf-zeta} we can show, that also
all the transversal derivatives (i.e. the derivatives over radial coordinate $r$) 
of all the elements of $\bar{\chi}$ are preserved by the flow of $\partial_v$, 
$\partial_{\varphi}$ and $X$.

Note that if the spacetime metric is analytic at $\ih\cup\N$ the invariance of all 
the transversal metric (and Maxwell field) derivatives with respect to $\partial_{\varphi}$ 
implies immediately, that the Bondi-like extensions $\zeta, \partial_{\varphi}$ of 
the null and axial symmetry induced at $\ih$ are Killing fields at the horizon
neighbourhood. This is true also in higher dimension and for any compact topology 
of the horizon base space \cite{hiw-rigid} It is a-priori not known whether this 
statement will remain true if we drop the analyticity assumption.

The necessary condition for $X$ to be KVF 
\begin{equation}
  X^{\mu}\partial_{\mu}\Psi_4|_{\N}
  \ =\ X^{\mu}\partial_{\mu}\Phi_2|_{\N}\ =\ 0
\end{equation} 
can be expressed in the form similar to \eqref{eq:KVF-constr}
\begin{subequations}\label{eq:KVF-constr-hel}\begin{align}
  \label{eq:KVF-Psi4-hel}
  \begin{split}
    (H\partial_r-\bar{X}\m-X\mb-C\partial_{\varphi}) \Psi_4\
    &=\ - (4\epsilon -\rho)\Psi_4 + \mb\Psi_3
      + (5\pi + 2a)\Psi_3 
      - 3\lambda\Psi_2
      - \kappa_0 ( \bar{\mu}\Phi_2\bar{\Phi}_0 \\
      &\ - (\pi+2a)\Phi_2\bar{\Phi}_1 ) 
      +\kappa_0 ( -2\lambda\Phi_1\bar{\Phi}_1
      +\partial_r\Phi_2\bar{\Phi}_0
      + \bar{\sigma}\Psi_2\bar{\Phi}_2 )
    \ ,
  \end{split} \\
  \label{eq:KVF-Phi2-hel}
  (H\partial_r - \bar{X}\m-X\mb-C\partial_{\varphi})\Phi_2\
    &=\ \mb\Phi_1 - \lambda\Phi_0 + 2\pi\Phi_1 + (\rho-2\epsilon)\Phi_2
    \ ,
\end{align}\end{subequations}
where $C$ is a constant of $\N$ such that $X=c_1(\partial_v +
C\partial_{\varphi})$ (with $c_1$ being a constant) at $\ih\cap\N$.
As a consequence one can formulate necessary and sufficient conditions for the 
existence of a helical HVF in the neighbourhood of the NEH in forms of the 
analogs of Theorems~\ref{thm:KVF-suffic} and~\ref{thm:KVF-suffic'}. The only 
difference with respect to those theorems is $(i)$ different differential condition 
for data at $\N_{v_0}$, namely \eqref{eq:KVF-constr} is replaced with 
\eqref{eq:KVF-Psi4-hel} and $(ii)$ additional condition at $\ih$: that the adapted 
null frame \eqref{eq:frame-def} is preserved by the flow of a cyclic symmetry of 
$\ih$ induced on it owing to Theorem~\ref{thm:hel-decomp}.

\section{Conclusions}\label{sec:concl}

In the article we explored the possibility of constructing 
a well defined and 
convenient to use description of a non-expanding horizon spacetime neighbourhood.  
This goal has been achieved by a construction of a preferred coordinate system, 
built with use of the geometric invariant of the horizon geometry: invariant null 
vector field and the foliation compatible with its flow. For the class of the horizons
named here  'generic-invariant' (and covering all horizon geometries except for special
non-generic cases) such structure is unique, whereas the small class of non-generic 
horizons may allow for several such structures on one horizon. This structure allowed to 
arrive to the following results true for the neighbourhood of the horizon of arbitrary 
dimension and arbitrary compact spatial slice (or equivalently base space) topology:
\begin{enumerate}
	\item The distinguished coordinate system defined at non-expanding horizon has been 
		extended to the horizon neighbourhood analogously to the construction of the Bondi
		coordinate system near the null SCRI: the horizon coordinates are transported along 
		the (defined uniquely for generic-invariant horizon) geodesics generated
		by the null field transversal to the slices of distinguished foliation of the horizon.
		These coordinates are then supplemented by the affine parameter along the above mentioned
		geodesics. Specifically, the coordinates are determined uniquely by the distinguished 
		null flow and foliation of the horizon via (\ref{eq:nM}, \ref{eq:r-M}, \ref{eq:v-M}, 
		\ref{eq:xA-M}).
		\label{it:concl-Bondi}
	\item The specific construction of the above coordinate system proves to be very convenient
		in case when the non-expanding horizon is a Killing horizon, that is there exists at 
		the horizon neighbourhood a Killing vector field tangent to the horizon and null at it. 
		The Killing vector field  takes in specified coordinate system a particular very 
		simple form given by \eqref{eq:k-l}. 
		
		That allows in particular to test immediately if the non-expanding 
		horizon is a Killing horizon as one needs only to verify whether the fields of the form 
		\eqref{eq:k-l} are Killing fields.
		\label{it:concl-Killing1}
	\item The quasi-local version of the Hawking rigidity theorem \cite{lp-symm}, see also 
		\cite{hiw-rigid} allows to generalize the above results to the Killing fields, which 
		are not necessarily null at the horizon. Then again, such fields take in Bondi-like 
		coordinate system a specific form given by \eqref{eq:Kill-chiral-form}. 
		\label{it:concl-Killing2}
\end{enumerate}

n the case of non-expanding horizons in four dimensional spacetime (but still for 
the arbitrary compact topology of the horizon base space) and arbitrary (up to energy 
condition \ref{c:energy}) matter content the known Newman-Penrose formalism
allowed to construct on the horizon neighbourhood an 
invariant null frame, compatible with the horizon invariant structure and the 
Bondi-like coordinate system. In case of a generic-invariant horizon such frame 
is again defined uniquely. An application of this frame to describe the spacetime 
metric near the horizon led to the following result:
\begin{enumerate}
	\setcounter{enumi}{4}
	\item The invariant Bondi-like coordinate system allowed to define in an 
		invariant way a radial expansion of the spacetime metric about the horizon. This 
		and representation of the geometric data in a distinguished null frame allowed 
		in turn to identify a free data needed to determine the expansion of the 
		spacetime metric at the horizon up to desired order. This data is specified
		by the Corollary~\ref{cor:data-gen}. It does not require to know the evolution 
		equations of the matter fields present at the horizon neighbourhood.
		
		However, if one considers specific matter content the matter field equations 
		may induce additional constraints on the above (otherwise free) data.
		\label{it:concl-expansion}
\end{enumerate}

The matter of studies was subsequently restricted to the horizons in $4$-dimensional 
spacetime admitting Maxwell field only. In this case the distinguished null frame 
introduced before allows for a convenient representation of the Maxwell field equations.
This in turn allowed to improve the result of point~\ref{it:concl-expansion} as well 
as to establish new ones. In particular:
\begin{enumerate}
	\setcounter{enumi}{5}
	\item The Maxwell field equations coupled to Einstein-Maxwell ones allowed to 
		reduce the set of data at the horizon required to determine the expansion of 
		the spacetime metric up to desired order. Such expansion along the whole horizon
		is now determined by the appropriate data (specified in Corollary~\ref{cor:data-evac})
		on a single spatial slice of the horizon.
		\label{it:concl-exp-evac}
	\item Known formulation of the characteristic initial value problem for 
		Einstein-Maxwell field equations (together with Maxwell evolution equations) 
		and known specification of the free data (so called Friedrich reduced data 
		on the boundary surfaces: in our case the horizon and the null surface transversal 
		to it) in terms of the components of geometry and matter field in the distinguished 
		null frame allowed in turn to determine the necessary and sufficient conditions
		for the existence on the non-expanding horizon neighbourhood of the Killing vector
		field tangent to and null at the horizon. this condition takes a form of 
		a differential constraint involving the (otherwise free) data on the null surface 
		transversal to the horizon. In general electro-vacuum case it is the set of 
		constraints \eqref{eq:KVF-constr}, whereas in the vacuum case it is a single 
		constraint \eqref{eq:KVF-Psi4} with $\Phi_0, \Phi_1, \Phi_2$ set to zero.

		These constraints allow in particular to probe the spacetime geometries determined 
		numerically near the horizon for stationarity. Also the found constraints allow 
		to construct in a straightforward way an invariant quantity which vanishes when 
		the a non-expanding horizon is a Killing horizon and constituting the measure 
		of departure from stationarity otherwise. It may then prove useful in description 
		of the spacetime near black hole horizon in its final stages of evolution (like 
		final stages of black hole merging process).
		\label{it:concl-KVF}
\end{enumerate}

The developments presented in this paper, although being of considerable potential 
use by themselves, are essentially an exploration of possibilities to build and use the
description of a black hole neighbourhood via natural expansion of the distinguished 
geometry structures of its horizon. Therefore, thy are meant to be mainly a methodology 
example rather than a complete studies. For that reason, for example only the Maxwell field 
as a matter content has been considered when developing the results listed in points 
\ref{it:concl-exp-evac} and \ref{it:concl-KVF} above. These results can be easily extended 
to other types of matter, provided the appropriate formulation of characteristic 
initial value problem for any such type of matter exists and the free initial data for 
such problem are identified. In particular, the theorems and construction of the whole 
article can be extended in a straightforward way to admit a nonvanishing cosmological 
constant, thus opening the applications in description of black holes in asymptotically 
anti-DeSitter spacetimes.

Similarly, the results presented in point \ref{it:concl-expansion} rely heavily 
on the Newman-Penrose complex tetrad formalism, tailored specifically to $4$-dimensional 
spacetimes. One can however introduce in arbitrary dimension a real orthonormal 
vielbain (see for example \cite{apv-ads}) playing the same role. Rewriting the Einstein 
equations in such vielbain, while more involved than in dimension $4$ does not pose 
any new qualitative challenges. As a consequence, at least the description of the radial
metric expansion at the horizon can be extended in a systematic way to higher dimension.

The result of point \ref{it:concl-Bondi} can be extended in two ways. On one hand, 
the invariant geometric structure of the horizon (although a bit modified) is still 
present on the dynamical horizon. Since to build the Bondi-like coordinate system on 
the horizon neighbourhood one only needs this structure and the congruence of the null 
geodesics transversal to the horizon, it is straightforward to introduce such description 
in the neighbourhood of a dynamical black hole. A similar to ours construction (based on
general foliations by marginally trapped tubes) has been for example proposed in
\cite{b-spc}.

On the other hand, the specified coordinate system construction does not employ the 
Einstein field equations outside of the energy condition \ref{c:energy}, which in 
turn can be formulated in terms of purely geometric quantities (instead of the matter
stress-energy tensor). As a consequence this construction can be extended also 
to modified theories of gravity, although for such extension certain care  is needed 
to probe how the distinguished geometry structure at the horizon changes.

\begin{acknowledgments}
	We would like to thank Piotr Chru\'sciel and Vincent Moncrief for enlightening 
	discussions and helpful comments. TP also thanks Rodrigo Olea-Aceituno for helpful 
	advises on the references.
	This work has been 
	supported in parts by the Chilean FONDECYT regular grant 1140335, the National 
	Center for Science (NCN) of Poland research grants 2012/05/E/ST2/03308, 
	2011/02/A/ST2/00300, 
	and the Spanish MICINN research grant FIS2011-30145-C03-02.
	TP also acknowledges the financial support of UNAB via internal project no DI-562-14/R.
\end{acknowledgments}

\appendix

\section{The Newman-Penrose formalism}
  \label{app:NP}
  
Here we briefly sumarize the Newman-Penrose formalism, as specified 
in \cite{exact}.

\subsection{The NP null frame}
\label{app:NP-frame}

The complex null vector frame $(e_1,\dots,e_4)=(m,\bar{m},n,\ell)$
(where $n,l$ are real vectors and $m$ is complex) of a four-dimensional
spacetime consists the
Newman-Penrose null tetrad if the following scalar products
\begin{subequations}\label{eq:NP-prod}\begin{align}
  g^{\mu\nu}m^{\mu}\bar{m}^{\nu} &= 1  &
  g^{\mu\nu}n^{\mu}\ell^{\nu} &= -1  \tag{\ref{eq:NP-prod}}
\end{align}\end{subequations}
are the only nonvanishing products of the frame components. The dual
frame corresponding to the tetrad will be denoted as
$(e^1,\dots,e^4)$. In terms of this coframe components the metric
tensor takes the form:
\begin{equation}\label{eq:NP-metric}
  g_{\mu\nu}\ =\ e^1{}_{\mu}e^2{}_{\nu} + e^2{}_{\mu}e^1{}_{\nu}
    - e^3{}_{\mu}e^4{}_{\nu} - e^4{}_{\mu}e^3{}_{\nu}
\end{equation}

The torsion-free spacetime connection is determined by the 1-forms defined
as follows
\begin{subequations}\label{eq:NP-con-def}\begin{align}
  \Gamma_{\alpha\beta}\ &=\ - \Gamma_{\beta\alpha} &
  \rd e^{\alpha} + \Gamma^{\alpha}{}_{\beta}\wedge e^{\beta}\ &=\ 0
    \tag{\ref{eq:NP-con-def}}
\end{align}\end{subequations}
which can be decomposed onto the following complex coefficients
\begin{subequations}\label{eq:NP-con}\begin{align}
  - \Gamma_{14} &= \sigma e^1 + \rho e^2 + \tau e^3 + \varkappa e^4 &
  - \fracs{1}{2}(\Gamma_{12} + \Gamma_{34})
    &= \beta e^1 + \alpha e^2 + \gamma e^3 + \epsilon e^4  \\
    \Gamma_{23} &= \mu e^1 + \lambda e^2 + \nu e^3 + \pi e^4 
\end{align}\end{subequations}
called the spin coefficients.

The Riemann tensor is given by the equation
\begin{equation}\label{eq:NP-Riem}
  \fracs{1}{2}
    \Riem{4}^{\alpha}{}_{\beta\gamma\delta} e^{\gamma}\wedge e^{\delta}
  = \rd\Gamma^{\alpha}{}_{\beta}
    + \Gamma^{\alpha}_{\gamma}\wedge\Gamma^{\gamma}{}_{\beta}
\end{equation}

The Ricci and Weyl tensor
\begin{subequations}\begin{align}
  \Ricn{4}_{\alpha\beta} &:= \Riem{4}^{\gamma}{}_{\alpha\gamma\beta}
    \hspace{9cm}
  \Ricn{4} := \Ricn{4}^{\gamma}{}_{\gamma} \\
	\Weyl{4}_{\alpha\beta\gamma\delta}
		&:= \Riem{4}_{\alpha\beta\gamma\delta}
				+\fracs{1}{6}\Ricn{4} ( g_{\alpha\gamma}g_{\beta\delta}
																-g_{\alpha\delta}g_{\beta\gamma} ) 
		-\fracs{1}{4} ( g_{\alpha\gamma}\Ricn{4}_{\beta\delta}
											-g_{\beta\gamma}\Ricn{4}_{\alpha\delta}
											+g_{\beta\delta}\Ricn{4}_{\alpha\gamma}
											-g_{\alpha\delta}\Ricn{4}_{\beta\gamma} )
\end{align}\end{subequations}
is described in the formalism by the complex coefficients defined as:
\begin{subequations}\label{eq:NP-WR}\begin{align}
  \Psi_0 &= -\Weyl{4}_{4141}
    & \Phi_{00} &= -\fracs{1}{2}\Ricn{4}_{44}
    & \Phi_{12} &= -\fracs{1}{2}\Ricn{4}_{31} \\
  \Psi_1 &= -\Weyl{4}_{4341}
    & \Phi_{01} &= -\fracs{1}{2}\Ricn{4}_{41}
    & \Phi_{20} &= -\fracs{1}{2}\Ricn{4}_{22} \\
  \Psi_2 &= -\Weyl{4}_{4123}
    & \Phi_{02} &= -\fracs{1}{2}\Ricn{4}_{11}
    & \Phi_{21} &= -\fracs{1}{2}\Ricn{4}_{32} \\
  \Psi_3 &= -\Weyl{4}_{4323}
    & \Phi_{10} &= -\fracs{1}{2}\Ricn{4}_{42}
    & \Phi_{22} &= -\fracs{1}{2}\Ricn{4}_{33} \\
  \Psi_4 &= -\Weyl{4}_{3232}
    & \Phi_{11} &= -\fracs{1}{4}(\Ricn{4}_{43}+\Ricn{4}_{12})
    & \fracs{1}{24}\Ricn{4} &= \fracs{1}{12}(\Ricn{4}_{43}-\Ricn{4}_{12})
\end{align}\end{subequations}

The equation \eqref{eq:NP-Riem} written in terms of the components
(\ref{eq:NP-con},\ref{eq:NP-WR}) takes the form
\begin{subequations}\label{eq:NP-eq}\begin{align}
  \label{eq:E-Psi1}
  \m\rho-\mb\sigma
    &= \rho(\bar{\alpha}+\beta)-\sigma(3\alpha-\bar{\beta})
     +\tau(\rho-\bar{\rho})+\varkappa(\mu-\bar{\mu})-\Psi_1+\Phi_{01} \\
  \label{eq:E-Psi2}
    \m\alpha-\mb\beta
    &= (\mu\rho-\lambda\sigma)+\alpha\bar{\alpha}+\beta\bar{\beta}
     -2\alpha\beta+\gamma(\rho-\bar{\rho}) 
    +\epsilon(\mu-\bar{\mu})-\Psi_2+\Phi_{11}+\fracs{1}{24}\Ricn{4} \\
  \label{eq:E-Psi3}
    \m\lambda-\mb\mu
      &=\nu(\rho-\bar{\rho})+\pi(\mu-\bar{\mu})
       +\mu(\alpha+\bar{\beta}) 
      +\lambda(\bar{\alpha}-3\beta)-\Psi_3+\Phi_{21} \\
  \label{eq:D-rho}
  D\rho-\mb\varkappa
    &= (\rho^2+\sigma\bar{\sigma})+\rho(\epsilon+\bar{\epsilon})
     -\bar{\varkappa}\tau-\varkappa(3\alpha+\bar{\beta}-\pi)+\Phi_{00} \\
  \label{eq:D-sigma}
  D\sigma-\m\varkappa
    &= \sigma(\rho+\bar{\rho}+3\epsilon-\bar{\epsilon})
     -\varkappa(\tau-\bar{\pi}+\bar{\alpha}+3\beta)+\Psi_0 \\
  \label{eq:D-alpha}
  D\alpha-\mb\epsilon
    &=\alpha(\rho+\bar{\epsilon}-2\epsilon)+\beta\bar{\sigma}
     -\bar{\beta}\epsilon-\varkappa\lambda-\bar{\varkappa}\gamma
     +\pi(\epsilon+\rho)+\Phi_{10} \\
  \label{eq:D-beta}
  D\beta-\m\epsilon
    &= \sigma(\alpha+\pi)+\beta(\bar{\rho}-\bar{\epsilon})
     -\varkappa(\mu+\gamma)
     -\epsilon(\bar{\alpha}-\bar{\pi})+\Psi_1 \\
  \label{eq:D-lambda}
  D\lambda-\mb\pi
    &=(\rho\lambda+\bar{\sigma}\mu)+\pi(\pi+\alpha-\bar{\beta})
     -\nu\bar{\varkappa}-\lambda(3\epsilon-\bar{\epsilon})+\Phi_{20} \\
  \label{eq:D-mu}
  D\mu-\m\pi
    &=(\bar{\rho}\mu+\sigma\lambda)
     +\pi(\bar{\pi}-\bar{\alpha}+\beta)
     -\mu(\epsilon+\bar{\epsilon})-\nu\varkappa+\Psi_2+\fracs{1}{12}\Ricn{4} \\
  \label{eq:n-epsilon}
    D\gamma-\n\epsilon
      &=\alpha(\tau+\bar{\pi})+\beta(\bar{\tau}+\pi)
       -\gamma(\epsilon+\bar{\epsilon})
       -\epsilon(\gamma+\bar{\gamma}) 
      +\tau\pi-\nu\varkappa+\Psi_2+\Phi_{11}-\fracs{1}{24}\Ricn{4} \\
  \label{eq:n-kappa}
  D\tau-\n\varkappa
    &=\rho(\tau+\bar{\pi})+\sigma(\bar{\tau}+\pi)
     +\tau(\epsilon-\bar{\epsilon})-\varkappa(3\gamma+\bar{\gamma})
     +\Psi_1+\Phi_{01} \\
  \label{eq:n-pi}
  D\nu-\n\pi
    &=\mu(\pi+\bar{\tau})+\lambda(\bar{\pi}+\tau)
     +\pi(\gamma-\bar{\gamma})
     -\nu(3\epsilon+\bar{\epsilon})+\Psi_3+\Phi_{21} \\
  \label{eq:n-lambda}
  \n\lambda-\mb\nu
    &=-\lambda(\mu+\bar{\mu}+3\gamma-\bar{\gamma})
     +\nu(3\alpha+\bar{\beta}+\pi-\bar{\tau})-\Psi_4 \\
  \label{eq:n-mu}
  \m\nu-\n\mu
    &=(\mu^2+\lambda\bar{\lambda})+\mu(\gamma+\bar{\gamma})
     -\bar{\nu}\pi
     +\nu(\tau-3\beta-\bar{\alpha})+\Phi_{22} \\
  \label{eq:n-beta}
  \m\gamma-\n\beta
    &= \gamma(\tau-\bar{\alpha}-\beta)+\mu\tau-\sigma\nu
     -\epsilon\bar{\nu}-\beta(\gamma-\bar{\gamma}-\mu)
     +\alpha\bar{\lambda}+\Phi_{12} \\
  \label{eq:n-sigma}
  \m\tau-\n\sigma
    &= (\mu\sigma+\bar{\lambda}\rho)+\tau(\tau+\beta-\bar{\alpha})
     -\sigma(3\gamma-\bar{\gamma})-\varkappa\bar{\nu}+\Phi_{02} \\
  \label{eq:n-rho}
  \n\rho-\mb\tau
    &= -(\rho\bar{\mu}+\sigma\lambda)
       +\tau(\bar{\beta}-\alpha-\bar{\tau})
       +\rho(\gamma+\bar{\gamma})+\nu\varkappa-\Psi_2-\fracs{1}{12}\Ricn{4} \\
  \label{eq:n-alpha}
  \n\alpha-\mb\gamma
    &= \nu(\rho+\epsilon)-\lambda(\tau+\beta)+\alpha(\bar{\gamma}
       -\bar{\mu}) +\gamma(\bar{\beta}-\bar{\tau})-\Psi_3
\end{align}\end{subequations}
where the differential operators $\m,D,\n$ correspond to the
null vectors:
\begin{subequations}\label{eq:NP-op}\begin{align}
  \m &:= m^{\mu}\partial_{\mu}  &
  \n &:= n^{\mu}\partial_{\mu}  &
  D   &:= \ell^{\mu}\partial_{\mu}  \tag{\ref{eq:NP-op}}
\end{align}\end{subequations}
The equalities \eqref{eq:NP-eq} are called the Newman-Penrose equations.

The commutators of the operators \eqref{eq:NP-op} can be expressed by the
spin coefficients;
\begin{subequations}\label{eq:NP-comm}\begin{align}
  [\n D - D \n]
    &= (\gamma+\bar{\gamma})D + (\epsilon+\bar{\epsilon})\n
      -(\tau+\bar{\pi})\mb - (\bar{\tau}+\pi)\m \\
  [\m D - D \m]
    &= (\bar{\alpha}+\beta-\bar{\pi})D + \varkappa\n
      -\sigma\mb - (\bar{\rho}+\epsilon-\bar{\epsilon})\m \\
  [\m\n - \n\m]
    &= -\bar{\nu}D + (\tau-\bar{\alpha}-\beta)\n
       +\bar{\lambda}\mb + (\mu-\gamma+\bar{\gamma})\m \\
  [\mb\m-\m\mb]
    &= (\bar{\mu}-\mu)D + (\bar{\rho}-\rho)\n
       -(\bar{\alpha}-\beta)\mb - (\bar{\beta}-\alpha)\m
\end{align}\end{subequations}

The Bianchi identity written in terms of the coefficients defined by
(\ref{eq:NP-con},\ref{eq:NP-WR}) consists the system of complex PDEs:
\begin{subequations}\begin{align}
  \begin{split}
    0 &= -\mb\Psi_0+D\Psi_1+(4\alpha-\pi)\Psi_0
         -2(2\rho+\epsilon)\Psi_1+3\varkappa\Psi_2 
      - D\Phi_{01}+\m\Phi_{00}\\
      &+ 2(\epsilon+\bar{\rho})\Phi_{01}
        +2\sigma\Phi_{10}-2\varkappa\Phi_{11} 
      - \bar{\varkappa}\Phi_{02}
        +(\bar{\pi}-2\bar{\alpha}-2\beta)\Phi_{00}
  \end{split} \label{eq:BI-DPsi1}\\
  \begin{split}
    0 &= \mb\Psi_1-D\Psi_2-\lambda\Psi_0+2(\pi-\alpha)\Psi_1
         +3\rho\Psi_2-2\varkappa\Psi_3 
      + \mb\Phi_{01}-\n\Phi_{00} \\
      &-2(\alpha+\bar{\tau})\Phi_{01}+2\rho\Phi_{11} 
      + \bar{\rho}\Phi_{02}-(\bar{\mu}
         -2\gamma-2\bar{\gamma})\Phi_{00}-2\tau\Phi_{10}-\fracs{1}{12}D\Ricn{4}
  \end{split} \label{eq:BI-DPsi2}\\
  \begin{split}
    0 &= -\mb\Psi_2+D\Psi_3+2\lambda\Psi_1-3\pi\Psi_2
         +2(\epsilon-\rho)\Psi_3+\varkappa\Psi_4 
      - D\Phi_{21}+\m\Phi_{20}  \\
      &+ 2(\bar{\rho}-\epsilon)\Phi_{21}
       -2\mu\Phi_{10}+2\pi\Phi_{11} 
      - \bar{\varkappa}\Phi_{22}-(2\bar{\alpha}-2\beta
         -\bar{\pi})\Phi_{20}-\fracs{1}{12}\mb\Ricn{4}
  \end{split} \label{eq:BI-DPsi3}\\
  \begin{split}
    0 &= \mb\Psi_3-D\Psi_4-3\lambda\Psi_2+2(2\pi+\alpha)\Psi_3
         -(4\epsilon-\rho)\Psi_4 - \n\Phi_{20}+\mb\Phi_{21} \\
      &+ 2(\alpha-\bar{\tau})\Phi_{21}+2\nu\Phi_{10} 
      + \bar{\sigma}\Phi_{22}-2\lambda\Phi_{11}
         -(\bar{\mu}+2\gamma-2\bar{\gamma})\Phi_{20}
  \end{split} \label{eq:B-DPsi4}\\
  \label{eq:nPsi0}
  \begin{split}
    0 &= -\n\Psi_0+\m\Psi_1+(4\gamma-\mu)\Psi_0-2(2\tau+\beta)\Psi_1
         +3\sigma\Psi_2 - D\Phi_{02}+\m\Phi_{01} \\ 
      &+ 2(\bar{\pi}-\beta)\Phi_{01}
         -2\varkappa\Phi_{12}-\bar{\lambda}\Phi_{00} 
      + 2\sigma\Phi_{11}
         +(\bar{\rho}+2\epsilon-2\bar{\epsilon})\Phi_{02}
  \end{split} \\
  \label{eq:nPsi1}
  \begin{split}
    0 &= -\n\Psi_1+\m\Psi_2+\nu\Psi_0+2(\gamma-\mu)\Psi_1-3\tau\Psi_2
         +2\sigma\Psi_3 + \n\Phi_{01}-\mb\Phi_{02} \\
      &+ 2(\bar{\mu}-\gamma)\Phi_{01}-2\rho\Phi_{12} 
      - \bar{\nu}\Phi_{00}+2\tau\Phi_{11}
         +(\bar{\tau}-2\bar{\beta}+2\alpha)\Phi_{02}+\fracs{1}{12}\m\Ricn{4}
  \end{split} \\
  \label{eq:nPsi2}
  \begin{split}
    0 &= -\n\Psi_2+\m\Psi_3+2\nu\Psi_1-3\mu\Psi_2+2(\beta-\tau)\Psi_3
         +\sigma\Psi_4 - D\Phi_{22}+\m\Phi_{21} \\
      &+ 2(\bar{\pi}+\beta)\Phi_{21}
         -2\mu\Phi_{11}-\bar{\lambda}\Phi_{20} 
      + 2\pi\Phi_{12}
         +(\bar{\rho}-2\epsilon-2\bar{\epsilon})\Phi_{22}
         -\fracs{1}{12}\n\Ricn{4}
  \end{split} \\
  \label{eq:nPsi3}
  \begin{split}
    0 &= -\n\Psi_3+\m\Psi_4+3\nu\Psi_2-2(\gamma+2\mu)\Psi_3
         -(\tau-4\beta)\Psi_4 + \n\Phi_{21}-\mb\Phi_{22} \\
      &+2(\bar{\mu}+\gamma)\Phi_{21}-2\nu\Phi_{11} 
      - \bar{\nu}\Phi_{20}+2\lambda\Phi_{12}
         +(\bar{\tau}-2\alpha-2\bar{\beta})\Phi_{22}
  \end{split} \\
  \begin{split}
    0 &= -D\Phi_{11} + \m\Phi_{10} + \mb\Phi_{01} - \n\Phi_{00}
         -\fracs{1}{8}D\Ricn{4} + (2\gamma-\mu+2\bar{\gamma}-\bar{\mu})\Phi_{00} \\
      &+ (\pi-2\alpha-2\bar{\tau})\Phi_{01}
         +(\bar{\pi}-2\bar{\alpha}-2\tau)\Phi_{10}
         +2(\rho+\bar{\rho})\Phi_{11} 
      + \bar{\sigma}\Phi_{02} + \sigma\Phi_{20}
         -\bar{\varkappa}\Phi_{12} - \varkappa\Phi_{21}
  \end{split} \\
  \begin{split}
    0 &= -D\Phi_{12} + \m\Phi_{11} + \mb\Phi_{02} - \n\Phi_{01}
         -\fracs{1}{8}\m\Ricn{4}
         +(-2\alpha+2\bar{\beta}+\pi-\bar{\tau})\Phi_{02} \\
      &+ (\bar{\rho}+2\rho-2\bar{\epsilon})\Phi_{12}
         +2(\bar{\pi}-\tau)\Phi_{11}
         +(2\gamma-2\bar{\mu}-\mu)\Phi_{01} 
      + \bar{\nu}\Phi_{00} - \bar{\lambda}\Phi_{21}-\varkappa\Phi_{22}
  \end{split}\\
  \label{eq:DPhi22}
  \begin{split}
     0 &= -D\Phi_{22} + \m\Phi_{21} + \mb\Phi_{12} - \n\Phi_{11}
          -\fracs{1}{8}\n\Ricn{4}
          +(\rho+\bar{\rho}-2\epsilon-2\bar{\epsilon})\Phi_{22} \\
       &+ (2\bar{\beta}+2\pi-\bar{\tau})\Phi_{12}
          +(2\beta+2\bar{\pi}-\tau)\Phi_{21} - 2(\mu+\bar{\mu})\Phi_{11} 
       + \nu\Phi_{01} + \bar{\nu}\Phi_{10} - \bar{\lambda}\Phi_{20}
          -\lambda\Phi_{02}
  \end{split}
\end{align}\end{subequations}

\subsection{The Einstein-Maxwell field equations}
\label{app:NP-EM}

Given an electromagnetic field 2-form $F_{\mu\nu}$ we define the following
complex coefficients
\begin{subequations}\label{eq:NP-F}\begin{align}
  \Phi_0 &:= F_{41} & \Phi_1 &:= \fracs{1}{2}(F_{43}+F_{21})
  & \Phi_2 := F_{23}  \tag{\ref{eq:NP-F}}
\end{align}\end{subequations}
which completely determine it.

The Maxwell field equations expressed in terms of the field
and spin coefficients take the form:
\begin{subequations}\label{eq:Maxwell-NP}\begin{align}
  \label{eq:D-Phi1}
  D\Phi_1 - \mb\Phi_0
    &= (\pi-2\alpha)\Phi_0 + 2\rho\Phi_1 -\varkappa\Phi_2 \\
  \label{eq:D-Phi2}
  D\Phi_2 - \mb\Phi_1
    &= -\lambda\Phi_0 + 2\pi\Phi_1 + (\rho-2\epsilon)\Phi_2 \\
  \label{eq:n-Phi0}
  \m\Phi_1 - \n\Phi_0
    &= (\mu-2\gamma)\Phi_0 + 2\tau\Phi_1 - \sigma\Phi_2 \\
  \label{eq:n-Phi1}
  \m\Phi_2 - \n\Phi_1
    &= -\nu\Phi_0 + 2\mu\Phi_1 + (\tau-2\beta)\Phi_2
\end{align}\end{subequations}

In an electrovac spacetime with vanishing cosmological constant the Ricci tensor is 
related with the field energy-momentum via the Einstein field equations of the 
following form:
\begin{equation}\label{eq:NP-F-T}
  \Phi_{\tilde{\alpha}\tilde{\beta}}
  = 16\pi G \Phi_{\tilde{\alpha}}\bar{\Phi}_{\tilde{\beta}}
  \qquad
  \tilde{\alpha},\tilde{\beta} \in \{ 0,1,2 \}
\end{equation}

\section{$4$-dimensional neighbourhood of the horizon}
\label{app:4d-neigh}

In this section we present the detailed proof of the Corollary
\ref{cor:data-gen} of section \ref{sec:frame-exp}.

\subsection{Proof of corollary \ref{cor:data-gen}}
  \label{app:frame-exp}

The proof is based on an explicit construction of the algorithm which allows to 
calculate the derivatives $\partial^ne^{\mu}$ of the frame components $m_A, Z_A, H$, 
provided the results for all the lower orders are given. This algorithm allows thus 
to establish the conclusion via induction.  

In general the higher radial derivatives of the frame components at the horizon can 
be obtained by differentiation (of appropriate order) over $r$ of the identity 
\eqref{eq:DPsi4} and transversal parts (i.e. the contraction of the considered 
equations with $n$) of (\ref{eq:NP-con-def}b, \ref{eq:NP-Riem}).
In particular the 0th order is given directly by \eqref{eq:fc-H}, whereas the first 
radial derivatives of ${e^\mu}_\nu$ are determined via \eqref{eq:N-fc-ih} by $(q,D)$. 
To demonstrate the method for $n>1$ we start with an explicit calculation of the $2$nd 
order before presenting the general derivation of $n+1$st order.

The presented method is applicable to any kind of matter field, however here we 
assume that for each order of the expansion the required Ricci tensor components 
and their radial derivatives are given on $\ih$. That assumption is true for example 
in the Maxwell and/or scalar /and/or dilaton case where the necessary Ricci tensor 
components are determined via the matter field equations by the respective data 
defined on the initial slice $\slc$. The explicit expansion in the Einstein-Maxwell
case was provided in section~\ref{sec:evac-exp}.

\subsubsection{The $2$nd order}

Assume now that the geometry $(q,D)$ and results of the first order evaluation are 
at our disposal. Then, acting with $\partial_r$ on \eqref{eq:N-fc-gen} one can derive 
$H_{,rr}, X_{,rr}, {Z}_A{}_{,rr},{m}_A{}_{,rr}$ in terms of the first radial 
derivatives of the connection coefficients, which in turn are given by the equations 
(\ref{eq:n-epsilon}, \ref{eq:n-mu}, \ref{eq:n-lambda}, \ref{eq:n-pi}).
The resulting formula for the second radial derivatives of the frame coefficients 
on $\M'$ reads
\begin{subequations}\label{eq:fc-r2}\begin{align}
  H_{,rr}\ &=\ \bar{\Psi}_2 + \Psi_2 + 2\Phi_{11}
    - \fracs{1}{12}\Ricn{4} + (\Psi_3+\Phi_{21})X
    + (\bar{\Psi_3}+\Phi_{12})\bar{X} \label{ddH}\\
  X_{,rr}\ &=\ -\bar{\Psi}_3 - \Phi_{12} - \Phi_{22}X
    - \Psi_4\bar{X} \label{ddX}\\
  Z_A{}_{,rr}\ &=\ (\Psi_3+\Phi_{21})m_A
    + (\bar{\Psi}_3+\Phi_{12})\bar{m}_A \label{ddZ}\\
  m_A{}_{,rr}\ &=\ - \Phi_{22}m_A - \bar{\Psi}_4\bar{m}_A \ . \label{ddm}
\end{align}\end{subequations}
The values of these derivatives at the horizon are given by substituting the frame
coefficients with their values on $\ih$. In particular:
\begin{subequations}\begin{align}
  H_{,rr}|_{\ih}\ &=\ \bar{\Psi}_2 + \Psi_2 + 2\Phi_{11}
    - \fracs{1}{12}\Ricn{4}\\
  X_{,rr}|_{\in}\ &=\ -\bar{\Psi}_3 - \Phi_{12},
\end{align}\end{subequations}
Note that the derivatives $H_{,rr}, X_{,rr}, Z_{A,rr}$ on $\ih$ are
determined directly by $(q,D)$ and the Ricci tensor. The last
derivative, ${m}_{A,rr}$, involves the solution $\Psi_4$ to the
equation 
\begin{equation}\label{eq:DPsi4-1}
  D\Psi_4|_{\ih}\ =\ -2\sgr{\bsl}\Psi_4 + \mb \Psi_3
    + (5\pi + 2a) \Psi_3 - 3\lambda \Psi_2 
   - \bar{\mu}\Phi_{20} + 2\alpha\Phi_{21} - 2\lambda\Phi_{11}
    - \Phi_{20}{}_{,r}
\end{equation}
(which is the restriction to $\ih$ of the equation \eqref{eq:DPsi4}).
The value of this solution is uniquely determined by the value of
$\Psi_4$ on chosen section and the horizon geometry.

To summarize, by direct inspection of the system of equations used here we see, 
that the data which is not determined, thus must be specified, consists of the 
following components: 
\begin{enumerate}[ (i)]
  \litem $\Phi_{21}, \Phi_{22}, \Ricn{4}, \Phi_{20,r}$ given on the
    entire $\ih$, and
  \litem $\Psi_4$ given on an initial slice $\slc$.
\end{enumerate}

\subsubsection{The $n+1$th order}

In this step we assume that at our disposal are the results of the derivation
up to $n$th order, that is:
\begin{enumerate}[ (i)]
  \litem the components of the frame and their radial derivatives up
    to the $n$th order,
  \litem the components of the connection and their radial derivatives
    up to the $n-1$st order,
  \litem the components of the Ricci tensor, and their derivatives up
    to the $n-2$nd order, as well as the following higher derivatives
    $\nr^{n-1} \Phi_{00}, \nr^{n-1} \Phi_{01}, \nr^{n-1} \Phi_{20},
    \nr^{n-1} (\Phi_{11} - \fracs{1}{8}\Ricn{4})$,
  \litem all the components of the Weyl tensor and their derivatives
    up to the $n-2$nd order.
\end{enumerate}

The following higher derivatives: $\nr^{n-1}\Psi_0, \nr^{n-1}\Psi_1,
\nr^{n-1}(\Psi_2 + \frac{1}{12}\Ricn{4}), \nr^{n-1}(\Psi_3-\Phi_{21})$
are then determined by the Bianchi identities (\ref{eq:nPsi0}-\ref{eq:nPsi3}).
Furthermore, the $n+1$th radial derivatives of the frame components are given by 
differentiating the equations \eqref{eq:fc-r2}, whereas the $n$th derivatives of 
the connection coefficients are given by differentiating the equations 
(\ref{eq:n-kappa}-\ref{eq:n-alpha}).
These data are also sufficient to derive the $n$th derivatives of
$\Psi_0,\ldots,\Psi_3, \Phi_{00}, \Phi_{10}, \Phi_{20}$ via the equations 
obtained by differentiating (\ref{eq:E-Psi1}-\ref{eq:E-Psi3},
\ref{eq:D-rho}-\ref{eq:D-mu})\footnote{provided the rest of the Ricci
  components appearing in the equations is given
} sufficiently many times.

The remaining Weyl tensor component $\nr^{n-1}\Psi_4$ is given by the equation 
derived by the differentiation of the Bianchi identity \eqref{eq:B-DPsi4}. The 
equation has the following form
\begin{equation}
  D\,\nr^{n-1}\Psi_4 {}\ =\ - (n+1)\sgr{\bsl} \nr^{n-1} \Psi_4 \
   +\ {\cal P}_n(e, \Gamma^{(n-1)},\Psi^{(n-1)}, \nr^n\Phi_{20},
                 \Phi^{(n-1)}, \nr^n \Ricn{4}) \ ,
\end{equation}
where $e$ represents the components of the frame, whereas $\Gamma^{(n-1)}$, $\Phi^{(n-1)}$ 
and $\Psi^{(n-1)}$ stands for the components and their derivatives up to the order 
of $n-1$ of, respectively, the connection, traceless part of the Ricci tensor, and
all of the Weyl tensor except $\Psi_4$.

In summary, given the results up to $n$th order we need to specify the
following data:
\begin{enumerate}[ (i)]
  \litem $\nr^{n-1}\Phi_{11}$, $\nr^{n-1}\Phi_{21}$, $\nr^{n-1}\Phi_{22}$, 
    $\nr^{n-1}\Ricn{4}$, $\nr^{n-1}\Phi_{02}$ given on the entire $\ih$, and
  \litem $\nr^{n-1}\Psi_4$ given on the initial slice $\slc$.
\end{enumerate}
to determine the $n+1$st order of the expansion.

\section{4-dimensional electrovac NEH}
  \label{app:evac-spc}

This appendix contains the derivation of technical results used in
section \ref{sec:evac-spc}: the proof of Corollary \ref{cor:data-evac}
and the detailed description of the derivation of Friedrich reduced
data at the transversal surface $\N$ used in section \ref{sec:ivp}.

\subsection{Proof of corollary \ref{cor:data-evac}}
  \label{app:evac-exp}

The structure of the proof is analogous to the one presented in
section \ref{app:frame-exp}, that is we again construct an algorithm
of derivation of $n+1st$ order expansion for given all lower orders
and use the mathematical induction. The only difference is that now
part of previously undetermined data can be now determined by Maxwell
fied equations. The modifications to the proof of section
\ref{app:frame-exp} go as follows:
\begin{itemize}
  \item 
    The Ricci tensor components (so the Maxwell field tensor) doesn't
    contribute to the $0$th and $1$st order of expansion, so for these
    orders we can directly apply analogous part of the proof in section
    \ref{app:frame-exp}. 
  \item 
    In the 2nd order the components $\Phi_{11},\Phi_{12},\Phi_{22}$ are
    determined by $\Phi_1,\Phi_2$ given on the distinguished initial slice
    $\slc$. The value of $\nr\Phi_{02}$ (needed to derive $\Psi_4$ via
    \eqref{eq:DPsi4-1}) at the horizon is determined by \eqref{eq:n-Phi0}
    \begin{equation}
      \nr\Phi_{20}|_{\ih}\ =\ -\ 8\pi G\Phi_2\mb\bar{\Phi}_1 \ .
    \end{equation}
  \item 
    Finally, given the frame and the Maxwell field expanded to an $n$th
    order and $\partial^{n+1}_r\Phi_2|_{\slc}$ (where $\slc$ is a slice from 
    the previous point) the derivative $\partial^{n+1}_r\Phi_2$ on $\ih$ is
    the solution to the equation obtained by action of $\partial^{n+1}_r$
    on \eqref{eq:D-Phi2}. This completes the set of the data needed for
    calculation of the $n+1$ order of the expansion.
\end{itemize}

\subsection{Characteristic IVP: data derivation on $\N$}
  \label{app:ivp}

Let $\N$ be a transversal null surface defined as in section
\ref{sec:Bondi} and intersecting a non-expanding horizon $\ih$ at the
slice $\slc$. The part of Friedrich reduced data in characteristic
initial value problem which corresponds to $\N$ consists of the
following components: Newman-Penrose frame and connection
coefficients, Weyl tensor components $\Psi_1,\cdots,\Psi_4$ and
Maxwell field tensor components $\Phi_1,\Phi_2$. We show here that 
once we specify $\Psi_4$ and $\Phi_2$ on $\N$ all the remaining data 
are determined by the data at $\slc$. 

Indeed, provided the coefficients $(\Psi_4,\Phi_2)$ are known,
the evolution equations (\ref{eq:n-mA-gen}, \ref{eq:n-mZ-gen}), the
Newman-Penrose equations (combined with the appropriate Einstein field
eqs) (\ref{eq:n-mu}, \ref{eq:n-lambda}, \ref{eq:n-pi}, \ref{eq:n-alpha},
\ref{eq:n-beta}), the Maxwell equation \eqref{eq:n-Phi1} and the
Bianchi identity \eqref{eq:nPsi3} form on $\N$ the following ODE system
\begin{subequations}\label{eq:ODE1}\begin{align}
  \nr m^A\ &=\ \bar{\lambda}\bar{m}^A + \mu m^A \\
  \nr (m^A Z_A)\ &=\ \pi + \mu(m^A Z_A) + \bar{\lambda}(\bar{m}^AZ_A)\\
  \label{eq:ODE-n-mu}
  \nr\mu \ &= \
    (\mu^2+\lambda\bar{\lambda}) + \kappa_0\Phi_2\bar{\Phi}_2 \\
  \label{eq:ODE-n-lambda}
  \nr\lambda\ &=\
    2\mu\lambda + \Psi_4 \\
  \label{eq:ODE-n-pi}
  \nr\pi \ &=\
   \pi\mu + \bar{\pi}\lambda + (\Psi_3+\kappa_0\Phi_2\bar{\Phi}_1) \\
  \label{eq:ODE-n-varpi}
  \nr\varpi\ &=\
    \mu\varpi - \lambda\bar{\varpi} + (\Psi_3-\kappa_0\Phi_2\bar{\Phi}_1) \\
  \label{eq:ODE-n-Phi1}
  -\nr\Phi_1\ &=\ \m\Phi_2
    - 2\mu\Phi_1 + (\bar{\pi}-\bar{\varpi})\Phi_2 \\
  \begin{split}\label{eq:ODE-n-Psi3}
    -\nr(\Psi_3-\kappa_0\Phi_2\bar{\Phi}_1)\
      &=\ \m\Psi_4 - \kappa_0\mb\Phi_{2}\bar{\Phi}_{2}
         - 4\mu(\Psi_3-\kappa_0\Phi_2\bar{\Phi}_1)  \\
      &\ + 2(\bar{\pi}-\bar{\varpi})\Psi_4
         - 2\kappa_0\left(\pi\Phi_{22} + \mu\Phi_{21}
                          - \lambda\Phi_{12}\right)
  \end{split}
\end{align}\end{subequations}
for the coefficients $(m^A,m^AZ_A,\mu,\lambda,\pi,a,\phi_1,\Psi_3)$.\footnote{Note, 
  that the functions $(m^A,m^AZ_A)$ appear in the system as coefficients of $\m,\mb$.
} This system has a unique solution for given initial data on $\slc$ (which in turn 
is already determined by the horizon geometry, see previous section).

Known solution of \eqref{eq:ODE1} can be next applied to the system
formed by (\ref{eq:n-rho}, \ref{eq:n-sigma}, \ref{eq:n-Phi0},
\ref{eq:nPsi2}), (where the Bianchi identity \eqref{eq:DPhi22} was
used to determine the value of $D\Phi_{22}$ in \eqref{eq:nPsi2})
\begin{subequations}\label{eq:ODE2}\begin{align}
  \label{eq:ODE-n-rho}
  \nr\rho \ &=\
    (\rho\mu+\sigma\lambda) + (\Psi_2-\fracs{1}{12}\Lambda)\\
  \label{eq:ODE-n-sigma}
  \nr\sigma\ &=\
    (\mu\sigma+\bar{\lambda}\rho) + \Phi_{02}\\
  \label{eq:ODE-n-Phi0}
  \nr\Phi_0\ &=\ \m\Phi_1 - \mu\Phi_0 + \sigma\Phi_2 \\
  \begin{split}\label{eq:ODE-n-Psi2}
    -\nr(\Psi_2+\kappa_0\Phi_1\bar{\Phi}_1)\
      &=\ \m\Psi_3 - \kappa_0\mb\Phi_1\bar{\Phi}_2
         + (\bar{\pi}-\bar{\varpi}) + \sigma\Psi_4
         - 3\mu(\Psi_2+\kappa_0\Phi_1\bar{\Phi}_1) \\
      &\ - \kappa_0\left( \rho\Phi_2\bar{\Phi}_2
           + (\pi-\varpi)\Phi_1\bar{\Phi}_2 - 5\mu\Phi_1\bar{\Phi}_1
           - \lambda\Phi_0\bar{\Phi}_2 \right)
  \end{split}
\end{align}\end{subequations}
which is then the system of ODEs for $(\rho,\sigma,\Phi_0,\Psi_2)$. The solution 
to it (also unique for a given initial data on $\slc$) determines furthermore the 
value of $\epsilon$ via \eqref{eq:n-epsilon}
\begin{equation}\label{eq:ODE3}
  \nr\epsilon \ = \
    \pi\bar{\pi} + \fracs{1}{2}(\bar{\pi}\varpi-\pi\bar{\varpi})
    + (\Psi_2 + \kappa_0\Phi_1\bar{\Phi}_1 + \fracs{1}{12}\Lambda) \
\end{equation}
so does determine the pair $(X,H)$ via (\ref{eq:n-X-gen},\ref{eq:n-H-gen}). 

The last remaining part of the reduced data: $\Psi_1$ is determined via 
eq.~\eqref{eq:nPsi1} 
\begin{equation}\label{eq:ODE-n-Psi1}
  \partial_r \Psi_1\ =\ \delta\Psi_2 - 2\mu\Psi_1 + 2\sigma\Psi_3 
  +\kappa_0(\partial_r\bar{\Phi}_0\Phi_1-\mb\bar{\Phi}_0\Phi_2)
  +2\kappa_0(\bar{\mu}\bar{\Phi}_0\Phi_1-\rho\bar{\Phi}_1\Phi_2)  \ ,
\end{equation}
where the radial derivatives of $\Phi_0, \Phi_1$ are determined via 
eqs.~(\ref{eq:ODE-n-Phi1},~\ref{eq:ODE-n-Phi0}).

The Weyl tensor $\Psi_0$ component is not a part of Friedrich reduced data, however 
we also describe its evolution since it (as well as the evolution of $D\Phi_0$) is 
needed in section \ref{sec:evac-kvf}.
The analysis of it requires a little more effort than the other components, as we 
do not have direct analog of ``radial evolution'' equations~\eqref{eq:ODE1} for 
this component. We overcome this problem applying the Bianchi identity \eqref{eq:nPsi0} 
which implies, that the transversal derivative of $\Psi_0$ is equal to:
\begin{equation}\label{eq:n-Psi0-N}\begin{split}
  -\nr\Psi_0 - \m\Psi_1
    + \kappa_0(D\Phi_{0}\bar{\Phi}_2-\m\Phi_{0}\bar{\Phi}_1)\
   &=\ - \mu\ \Psi_0 - (\bar{\pi}-\bar{\varpi})\ \Psi_1 + 3\sigma\ \Psi_2 
    + \kappa_0 (2(\epsilon-\bar{\epsilon})+\bar{\rho})\Phi_0\bar{\Phi}_2 \\
   &\hphantom{\ =} + \kappa_0( (\bar{\pi}-2\bar{a})\Phi_0\bar{\Phi}_1 )
    + 2\sigma\Phi_1\bar{\Phi}_1 - 2\varkappa\Phi_1\bar{\Phi}_2
    - \bar{\lambda}\Phi_0\bar{\Phi}_0 ) \ ,
\end{split}\end{equation}
where the component $D\Phi_2$ is determined by the Maxwell equation
\eqref{eq:D-Phi2}.

To compute the value of $D\Phi_0$ one can apply the equations
involving the derivatives of the connection components along
$\ell$. Acting by operator $D$ on the equation \eqref{eq:n-Phi0} and
substituting the values $D\Phi_2,D\Phi_1,D\mu,D\sigma,D\pi$ via the
equations (\ref{eq:D-Phi2}, \ref{eq:D-Phi1}, \ref{eq:D-mu},
\eqref{eq:D-sigma}) and the combination of (\ref{eq:D-alpha},
\ref{eq:D-beta}) (to extract $D\pi$) one gets the following equation
\begin{equation}\label{eq:nD-Phi0}
  -\nr D\Phi_0\ =\ -\mu D\Phi_0 + \Phi_2\Psi_0 + \mathcal{F} \ ,
\end{equation}
where $\mathcal{F}$ is a functional of the Friedrich reduced data (frame,
connection, field and Riemann except $\Psi_0$) and their
derivatives along the directions tangent to $\N$ up to $2$nd
order. Both the equations (\ref{eq:n-Psi0-N}, \ref{eq:nD-Phi0}) form
the system of ODEs which (similarly to the systems considered above) has a
unique solution for given initial data $(\Psi_0, D\Phi_0)|_{\slc}$.\footnote{As
  $\ih$ is a NEH both the components vanish on it.
}

\section{$4$-dimensional electrovac Killing horizon}\label{app:4d-kill}

Here we present the original form of Racz theorem of \cite{Racz1,*Racz2} and 
the proofs of Lemma \ref{lem:k'k} and \ref{lem:kvf-zeta} necessary 
to prove, that the necessary condition for the vector field $\zeta$ to be a Killing 
field (null at the horizon) are also the sufficient one.

\subsection{Racz theorem}\label{app:Racz}

Consider the field $K'_{\mu}$ defined a the solution to the following initial value problem
\begin{equation}\label{eq:KVF-cand}
  \nabla^{\mu}\nabla_{\mu}K'_{\alpha}
    + \Ricn{4}_{\alpha}{}^{\mu}K'_{\mu} = 0 \qquad
  K'_{\mu}|_{\ih\cup\N} = K'_{0\mu} \ .
\end{equation}

\begin{theorem}\cite{Racz1,*Racz2}\label{thm:Racz} 
  Suppose $(M,g_{\mu\nu})$ is a spacetime equipped with a metric tensor 
  and admitting matter fields represented by the set of tensor fields 
  $\mathcal{T}_{(I)\alpha_1,\cdots, \alpha_k}$ satisfying the a quasi-linear 
  hyperbolic system
  \begin{equation}\label{eq:Racz-field}
    \nabla^{\mu}\nabla_{\mu}\mathcal{T}_{(I)}\
    =\ \mathcal{F}_{(I)}(\mathcal{T}_{(J)},\nabla_{\nu}\mathcal{T}_{(J)},
	g_{\alpha\beta})
  \end{equation}
  and such that the energy-momentum tensor is a smooth function of the
  fields, their covariant derivatives and the metric, thus
  \begin{equation}\label{eq:Racz-energy}
    \Ricn{4}_{\mu\nu}\
    =\ \Ricn{4}_{\mu\nu}(\mathcal{T}_{(J)},\nabla_{\nu}\mathcal{T}_{(J)},
	g_{\alpha\beta}) \ .
  \end{equation}
  Then on the domain of dependence of some initial hypersurface $\Sigma$ (within 
  an appropriate initial value problem) there exists a non-trivial Killing vector 
  field $K'{}^{\mu}$ (such that $\lie_{K'}\mathcal{T}_{(I)}=0$) if and only if 
  there exists a non-trivial initial data set $K'{}^{\mu}|_{\Sigma}$ for
  (\ref{eq:KVF-cand}a) which satisfies
  \begin{equation}\label{eq:KVF-condit}
    0\ =\ \lie_{K'}g_{\mu\nu}|_{\Sigma}\ =\
    \nabla_{\alpha}\lie_{K'}g_{\mu\nu}|_{\Sigma}\ =\
    \lie_{K'}\mathcal{T}_{(I)}|_{\Sigma}  \ .
  \end{equation}
\end{theorem}

\subsection{Proof of Lemma \ref{lem:kvf-zeta}}\label{app:kvf-zeta}

Let us start with $0$th order first.
In the chosen null frame the symmetric tensor $A_{\mu\nu}$ such that
\begin{equation}
  A_{\mu\nu} := \zeta_{\mu;\nu}+\zeta_{\nu;\mu} \ .
\end{equation}
can be expressed at $\ih\cup\N$ directly by the frame coefficients
$(H,X)$ and their first derivatives. Due to \eqref{eq:N-fc-ih} the
components $A_{33}, A_{34}, A_{31}$ vanish identically. The rest of
components is equal to
\begin{subequations}\label{eq:A}\begin{align}
  A_{44} &= -DH + (\epsilon+\bar{\epsilon})H
            +\bar{\varkappa}X + \varkappa\bar{X} \\
  A_{14} &= DX-\m H + 2\pi H - \varkappa
            +(\bar{\rho}-\epsilon+\bar{\epsilon})X + \sigma\bar{X} \\
  A_{11} &= \m X + 2\bar{a}X + \bar{\lambda}H - \bar{\sigma} \\
  A_{12} &= \m\bar{X} + \mb X - 2\bar{a}\bar{X} -2aX
            + 2\mu H - (\rho-\bar{\rho})
\end{align}\end{subequations}
respectively.

At the horizon all the components of $A_{\mu\nu}$ are identically zero, as $H=X=0$ 
there. Furthermore, the constraint \eqref{eq:KVF-Phi2} implies, that $D\Phi_2=0$ 
at $\ih\cap\N$, thus (due to \eqref{DPhi2-nei}) $\Phi_2$ is Lie-dragged by $\zeta$ 
at the entire horizon. Since the frame components are preserved by the flow of $\bsl$ 
at $\ih$ the field $\bsl^{\mu}=\zeta^{\mu}|_{\ih}$ at the horizon is a symmetry of
$F_{\mu\nu}$.

An analysis of the behavior of $\zeta$ at the transversal surface $\N$ requires 
a little bit more effort. To proceed, let us re-express the components of $A_{\mu\nu}$ listed
in \eqref{eq:A} in terms of the commutators $[\delta,\partial_v]$.
Decomposing $\zeta$ in the null frame and applying the equations \eqref{eq:NP-comm} 
one can express the commutator $[\m,\partial_v]$ in terms of the connection coefficients 
and the derivatives of $(X,H)$. On the other hand, the same commutator is determined 
by the derivatives of the coefficients $(m^A,Z_A)$ with respect to $v$
\begin{equation}\label{eq:A-comm}\begin{split}
  [\m,\partial_v] \
    &= \ \left( (\m\bar{X}) - X\varpi + H\mu - \bar{\rho}
           - (\epsilon-\bar{\epsilon}) \right)\m   
    +  \left( (\m X) + X\bar{\varpi}
           + \bar{\lambda}H - \sigma \right)\mb  
    +  \left( (\m H) + X(\rho-\bar{\rho})
           - \bar{\pi}H + \varkappa \right)\n \\
    &= \ (m^AZ_A)_{,v}\partial_{r} - m^A{}_{,v}\partial_A
\end{split}\end{equation}
Comparing the expressions above with \eqref{eq:A} one realizes, that $A_{\mu\nu}$ 
can be written down in terms of the components of $[\m,\partial_v]$
\begin{subequations}\begin{align}
  A_{44} \ &=\ -\partial_vH + X\overline{[\m,\partial_v]_{\n}}
               + \bar{X}[\m,\partial_v]_{\n}    \\
  A_{14} \ &=\ -X[\m,\partial_v]_{\mb}
               - \bar{X}\overline{[\m,\partial_v]_{\m}}
               - \partial_vX - [\m,\partial_v]_{\n}  \\
  A_{11} \ &= \ [\m,\partial_v]_{\mb}                \\
  A_{21} \ &= \ [\m,\partial_v]_{\m} + \overline{[\m,\partial_v]_{\m}}
\end{align}\end{subequations}
where
\begin{equation}
  [\m,\partial_v] =: [\m,\partial_v]_{\m}\m + [\m,\partial_v]_{\mb}\mb
                   + [\m,\partial_v]_{\n}\n \ .
\end{equation}
Due to the second equality in \eqref{eq:A-comm} tensor $A_{\mu\nu}$ vanishes at $\N$ 
if and only if the derivatives over $v$ of the frame coefficients vanish there.

To verify whether the latter holds let us consider on $\N$ the set $\dot{\chi}_{(I)}$ 
formed by the derivatives over $v$ of the following data: the frame coefficients,
the connection coefficients $\{\mu, \lambda, \pi, a, \rho, \sigma, \epsilon \}$, 
the Maxwell field and Weyl tensor components except $\Psi_0$. At the intersection 
$\ih\cap\N$ all these data vanish. On the other hand, one can build a system of 
the ``evolution equations'' for $\dot{\chi}_{(I)}$ allowing to evolve the initial values
at $\ih\cap\N$ along null geodesics spanning $\N$. Such system can be obtained 
by action of the operator $\partial_v$ on the system (\ref{eq:ODE1}-\ref{eq:ODE3}, 
\ref{eq:n-H-gen}). 
Similarly to the original system (\ref{eq:ODE1}-\ref{eq:ODE3}, \ref{eq:n-H-gen}) it
constitutes a hierarchy of quasilinear ODEs polynomial in the data it involves. 
As the operator $\partial_v$ commutes with $\partial_r,\partial_A$ the resulting
system for $\dot{\chi}_{(I)}$ inherits the properties of the original one
(\ref{eq:ODE1}-\ref{eq:ODE3}, \ref{eq:n-H-gen}). In particular, for known values of the
frame, connection, Maxwell field and Weyl tensor coefficients and known
$\dot{\Psi}_4,\dot{\Phi}_2$ the new system forms the hierarchy of ODEs analogous to 
the hierarchy represented by (\ref{eq:ODE1}-\ref{eq:ODE3}, \ref{eq:n-H-gen}).
In consequence to solve our system one can apply the same algorithm as the one applied 
to (\ref{eq:ODE1}-\ref{eq:ODE3}, \ref{eq:n-H-gen}). If the geometry data is known 
at $\N$ then the system describing the evolution of $\dot{\chi}_{(I)}$ has a unique 
solution for given $(\dot{\Psi}_4,\dot{\Phi}_2)$. In particular, in the case when 
$\dot{\Psi}_4 = \dot{\Phi}_2 = 0$ at $\N$, the equations are linear and homogeneous in
$\dot{\chi}_{(I)}$, thus $\dot{\chi}_{(I)}=0$ is a unique solution. It means that
$\dot{H}=\dot{X}=\dot{m}^A=\dot{Z}_A=0$ and $\dot{\Phi}_0 = \dot{\Phi}_1=\dot{\Phi}_2=0$, 
which implies vanishing of $A_{\mu\nu}$ and $\lie_{\zeta}F_{\mu\nu}$ at $\N$.

To show \eqref{eq:kvf-nabla} for arbitrary $n$, it is sufficient to show vanishing 
of the derivatives in directions transversal to $\ih$ and $\N$ respectively. This 
requirement is equivalent to the following condition
\begin{subequations}\label{eq:kvf-c-m}\begin{align}
  \nr^n \lie_{\zeta}g_{\mu\nu} |_{\ih}\ &=\ 0 \ , &
  \partial_v^n \lie_{\zeta}g_{\mu\nu} |_{\N}\ &=\ 0 \ .
  \tag{\ref{eq:kvf-c-m}}
\end{align}\end{subequations}
As \eqref{eq:kvf-nabla} holds at $\Sigma$ for $n=0$ the above
condition is automatically satisfied at the intersection $\ih\cap\N$.
At $\ih$ the satisfaction of the condition \eqref{eq:kvf-c-m} can be shown by direct 
inspection of the equations analogous to the ones used in the development of the 
metric expansion in section \ref{sec:evac-exp} (and earlier in section \ref{sec:frame-exp}). 
Indeed the only coefficients in the set $\bar{\chi}$ of geometry data\footnote{The set
  consists of all the components of frame, connection, Weyl tensor and
  Maxwell field.}
which are not automatically constant along the horizon null geodesics
are $\Psi_4, \Phi_2$. Their derivatives over $v$ are constant
(exponential)\footnote{I.e. $\Psi_4, \Phi_2$ are exactly of the form
  $\Psi_4 = e^{-2\sgr{\bsl}v}\hat{\Psi}_4$, $\Phi_2 =
  e^{-\sgr{\bsl}v}\hat{\Phi}_2$ where $\hat{\Psi}_4$, $\hat{\Phi}_2$
  are constant along the generators of $\ih$.}
along null geodesics for $\sgr{\bsl}\neq 0$
($\sgr{\bsl}=0$) respectively due to (\ref{eq:DPsi4-1}, \ref{DPhi2-nei}) (as
all the data on the right-hand side except $\Psi_4,\Phi_2$ is constant
along integral curves of $\bsl$). Whence, they vanish as \eqref{eq:KVF-constr} 
hold in particular at $\ih\cap\N$.

Provided all the derivatives $\partial_v\nr^k\bar{\chi}_{(\mathcal{I})}$ up to $n$th 
order vanish, the derivatives of the next order of all the components of 
$\bar{\chi}_{(\mathcal{I})}$ except $\Psi_4,\Phi_2$ also vanish, as they are 
determined by $\partial^k_v\bar{\chi}_{(\mathcal{I})}$ (where $k\in\{0,\cdots,n\}$) 
via the respective derivatives (namely $\partial_v\nr^n$) of equations
(\ref{eq:ODE1}-\ref{eq:nD-Phi0}). The similar derivatives of
(\ref{eq:DPsi4-1}, \ref{DPhi2-nei}) imply then the constancy
(exponentiality)\footnote{In the same strict sense as for $\Psi_4,\Phi_2$}
of $\partial_v\nr^{n+1}\Psi_4, \partial_v\nr^{n+1}\Phi_2$. They satisfy
then $\partial_v\nr^{n+1}\Psi_4|_{\ih} = \partial_v\nr^{n+1}\Phi_2|_{\ih}
= 0$ due to satisfaction of these conditions at $\ih\cap\N$. Finally
by induction $\partial_v\nr^n\bar{\chi}_{(\mathcal{I})}=0$ for arbitrary $n$.

At the horizon the $n$th order transversal derivatives of $A_{\mu\nu}$ can be 
expressed as a functionals homogeneous in the elements of $\dot{\chi}$ and its 
transversal derivatives up to $n$th order. It was shown above that these components 
vanish at $\ih$. Whence
\begin{equation}
  \partial_v\nr^nA_{\mu\nu}|_{\ih}\ = \ 0 \ .
\end{equation}
The condition \eqref{eq:kvf-c-m} holds then at $\ih$, as it is satisfied at $\ih\cap\N$.

To obtain an analogous result at $\N$ we can apply the method used previously to 
verify the condition $A_{\mu\nu}|_{\N} = 0$. Now instead of initial value problem 
for $\dot{\chi}_{(\mathcal{I})}$ we need to consider IVP for 
$\partial^n_v\chi_{(\mathcal{I})}$ (where $\chi := \bar{\chi}\setminus\{\Psi_0\}$) 
analogous to it. This IVP has however the same properties as the one for 
$\dot{\chi}_{(\mathcal{I})}$. This statement completes the proof.

\subsection{Proof of Lemma \ref{lem:k'k}}\label{app:k'k}

The equation \eqref{eq:k'k} is already satisfied for $n=0$
by \eqref{eq:KVF-cand-lem}. Suppose now it is satisfied up to $n$th
order. Lemma \ref{lem:kvf-zeta} implies then, that
\begin{equation}
  \nabla^{(n)}_{\alpha_1,\cdots,\alpha_n}
  (  \nabla^{\mu}\nabla_{\mu}\zeta_{\nu}
   + \Ricn{4}_{\nu}{}^{\mu}\zeta_{\mu} )\ =\ 0 \
\end{equation}
at $\Sigma$. In consequence, due to \eqref{eq:KVF-cand} and the inductive assumption 
the following holds
\begin{equation}
  \nabla^{(n)}_{\alpha_1,\cdots,\alpha_n}
  \nabla^{\mu}\nabla_{\mu}(K'{}_{\nu}-\zeta_{\nu})|_{\Sigma}\ =\ 0 \ .
\end{equation}
Let us consider above equation at $\ih$ first. Its contraction (in all the $\alpha_i$) 
with $\bn^{\mu}$ produces the following condition
\begin{equation}
  (\lie_{\bn})^n \lie_{\bsl} \lie_{\bn} (K'{}_{\mu}-\zeta_{\mu})
  + (\lie_{\bn})^n \lie_{\bn} \lie_{\bsl} (K'{}_{\mu}-\zeta_{\mu})
  \ =\ \mathcal{F} (\lie_{\bn})^n \lie_{\bn} (K'{}_{\mu}-\zeta_{\mu}) \ ,
\end{equation}
where we decomposed $g^{\mu\nu}$ in $\nabla^{\mu}\nabla_{\mu}
= g^{\mu\nu}\nabla_{\mu}\nabla_{\nu}$ via \eqref{eq:NP-metric}
and applied the inductive assumption to express covariant derivatives
as Lie ones. $\mathcal{F}$ is a functional of metric derivatives up to
$n+2$nd order.

Due to Jacobi identity and $\lie_{\ell}\bn=0$ this formula can be
re-expressed in the following way
\begin{equation}
  \lie_{\bsl} (\lie_{\bn})^{n+1} (K'{}_{\mu}-\zeta_{\mu})\
  =\ \fracs{1}{2} \mathcal{F} (\lie_{\bn})^{n+1}
     (K'{}_{\mu}-\zeta_{\mu}) \ .
\end{equation}
It constitutes a linear homogeneous ODE for $(\lie_{\bn})^{n+1}
K'_{\mu}$. As at $\ih\cap\N$ $(\lie_{\bn})^{n+1}
(K'{}_{\mu}-\zeta_{\mu}) = 0$ its only solution at $\ih$ is
\begin{equation}
  (\lie_{\bn})^{n+1} (K'{}_{\mu}-\zeta_{\mu}) |_{\ih}\ =\ 0 \ .
\end{equation}
The similar initial value problem can be formulated for
$(\lie_{\ell})^{n+1}(K'{}_{\mu}-\zeta_{\mu})$ at $\N$, whence
\begin{equation}
  (\lie_{\ell})^{n+1} (K'{}_{\mu}-\zeta_{\mu}) |_{\N}\ =\ 0 \ .
\end{equation}
As a consequence, provided \eqref{eq:k'k} is satisfied for $k\in\{0,\cdots,n\}$,
it holds also for $n+1$. Whence, by induction \eqref{eq:k'k} holds for arbitrary 
$n\in\mathbb{N}$.

\bibliography{LP-neigh}{}
\bibliographystyle{apsrev4-1}

\end{document}